\documentclass[11pt]{article}
\usepackage[letterpaper, margin=1in]{geometry}
\usepackage{bbm}
\usepackage{mathrsfs}
\usepackage{epsfig}
\usepackage{soul,color}
\usepackage{graphicx}
\usepackage{enumitem}
\usepackage{amsfonts}
\usepackage{amsthm}
\usepackage[figuresright]{rotating}
\usepackage{amssymb}
\usepackage{amsmath}
\usepackage{dcolumn}
\usepackage{physics}
\usepackage{float} 
\usepackage{bm}
\usepackage{verbatim}
\usepackage[normalem]{ulem}
\usepackage[ruled,vlined,linesnumbered]{algorithm2e}
\usepackage{setspace}
\usepackage[colorlinks,linkcolor=blue,anchorcolor=blue,citecolor=blue,urlcolor=blue]{hyperref}
\usepackage{cleveref}
\usepackage{stackengine}
\usepackage{qcircuit}
\usepackage{mathrsfs}
\usepackage{makecell}

\usepackage{tikz}
\usetikzlibrary{shapes.geometric, calc}

\usepackage{xcolor}
\newtheorem{theorem}{Theorem}[section]
\newtheorem{lemma}[theorem]{Lemma}
\newtheorem{corollary}[theorem]{Corollary}
\newtheorem{example}[theorem]{Example}
\newtheorem{proposition}[theorem]{Proposition}
\newtheorem{definition}{Definition}[section]
\newtheorem{fact}{Fact}[section]
\newtheorem{remark}{Remark}[section]
\crefname{fact}{fact}{facts}
\Crefname{fact}{Fact}{Facts}

\newcommand{\poly}[1]{\mathrm{poly}\left(#1\right)}

\newcommand{\pr}[1]{\mathrm{Pr}\left[#1\right]}
\newcommand{\lr}[1]{\left(#1\right)}

\newcommand{\bH}{\mathbb{H}}
\newcommand{\C}{\mathbb{C}}
\newcommand{\R}{\mathbb{R}}

\newcommand{\cC}{\mathcal{C}}
\newcommand{\cD}{\mathcal{D}}
\newcommand{\cE}{\mathcal{E}}

\newcommand{\cL}{\mathcal{L}}
\newcommand{\cM}{\mathcal{M}}
\newcommand{\cN}{\mathcal{N}}
\newcommand{\cP}{\mathcal{P}}
\newcommand{\cS}{\mathcal{S}}
\newcommand{\cT}{\mathcal{T}}

\newcommand{\udt}{\mathcal{U}_{\delta t}}

\newcommand{\bnorm}[1]{\Big\|#1\Big\|}

\newcommand{\tmix}{T_\mathrm{mix}}

\newcommand{\Tset}{\mathscr{T}}
\newcommand{\trange}{\operatorname{range}_{\Tset}}
\newcommand{\dist}{\operatorname{dist}}
\newcommand{\range}{\operatorname{range}}

\def\ba#1\ea{\begin{equation}#1\end{equation}}

\definecolor{Magenta}{rgb}{1,0,1}

\begin{document}
\title{Slow Mixing of Quantum Gibbs Samplers}

\author{
{\sf David Gamarnik}\thanks{Operations Research Center, Statistics and Data Science Center,  Sloan School of Management, MIT, Cambridge, MA, USA; e-mail: {\href{mailto:gamarnik@mit.edu}{\texttt{gamarnik@mit.edu}}}} 
\and
{\sf Bobak T.\ Kiani}\thanks{John A. Paulson School of Engineering and Applied Sciences, Harvard, Cambridge, MA, USA}
\thanks{Department of Computer Science, Bowdoin College, Brunswick, ME, USA; e-mail: {\href{mailto:b.kiani@bowdoin.edu}{\texttt{b.kiani@bowdoin.edu}}}.}
\and
{\sf Alexander Zlokapa}\thanks{Center for Theoretical Physics, MIT, Cambridge, MA, USA; e-mail: {\href{mailto:azlokapa@mit.edu}{\texttt{azlokapa@mit.edu}}} } 
}
\date{ }

\maketitle
\begin{abstract} 

Thermal (Gibbs) state preparation is a common task in physics and computer science. Recent works~\cite{chen2023efficient,gilyen2024quantum,ding2024efficient} propose algorithms based on open system dynamics associated with the cooling of a system coupled to a bath. The cost of these algorithms is given by a mixing time, analogous to the mixing time of classical Metropolis-like sampling algorithms. Few methods are available to bound mixing times of non-commuting quantum Hamiltonians or to show slow mixing of quantum systems. In contrast, a standard set of tools exists classically to lower-bound mixing times of local samplers by showing clustering of the Gibbs measure into distant subspaces separated by a bottleneck region.

Here, we introduce a \emph{strong clustering} property of Gibbs states that is uniquely quantum and permits clusters to be defined in a different basis from the Hamiltonian, including non-commuting models whose eigenbases are beyond the reach of perturbation theory. We show how to construct these clusters under two suitable measures of distance, defined in terms of Bohr spectrum jumps and in terms of locality. We prove that strong clustering implies slow mixing, giving unconditional lower bounds on the mixing time of recently proposed Gibbs samplers for several Hamiltonian models:
\begin{enumerate}
    \item \emph{Classical Hamiltonians.} Given a classical Hamiltonian and a classical bottleneck with mixing time lower bound $\tmix = \exp[\Omega(n^\alpha)]$ for classical Metropolis-like algorithms at inverse temperatures $\beta > \beta_c$, we show that quantum Gibbs samplers also satisfy $\tmix = \exp[\Omega(n^\alpha)]$ for $\beta > \beta_c$ despite the ability to perform off-diagonal operations. Example Hamiltonians include random $K$-SAT instances and $p$-spin glasses.
    \item \emph{Good quantum code Hamiltonians.} We give a short proof that $\tmix = \exp[\Omega(n)]$ for Hamiltonians associated with good $n$-qubit stabilizer codes at sufficiently low constant temperature whose checks satisfy a small-set expansion property. 
    \item \emph{2D ferromagnetic transverse field Ising model (TFIM).} For geometrically local Hamiltonians, we show that bottlenecks in the probability distribution imply strong clustering in dimension $D \geq 3$. The 2D TFIM requires a more refined analysis via the Poisson Feynman-Kac representation, which we use to lift classical bottlenecks in the presence of a transverse field. We show a lower bound of $\tmix = \exp[n^{1/2-o(1)}]$ for the $\sqrt n \times \sqrt n$ 2D ferromagnetic TFIM given quasi-local Gibbs samplers at constant temperature.
\end{enumerate}

\end{abstract}

\pagebreak

\section{Introduction}

The problem of efficiently sampling from a thermal (Gibbs) distribution is a fundamental task in computer science and statistical mechanics. For a quantum system described by a Hamiltonian $H$ and an inverse temperature $\beta$, the task of a quantum Gibbs sampler is to prepare the Gibbs state
\begin{equation}
    \rho_\beta = \frac{e^{-\beta H}}{\text{Tr}(e^{-\beta H})}.
\end{equation}
In physics, $\rho_\beta$ represents the thermal equilibrium state of the system and is thus of interest in the context of quantum simulation, such as measuring properties of low-energy states or probing phase transitions with respect to temperature~\cite{alhambra2023quantum}. In computer science, 
sampling from the Gibbs state is a fundamental step in the algorithmic problem of counting in graphs~\cite{jerrum1996markov}. Also
sampling from $\rho_\beta$ is used as a primitive in tasks such as quantum semidefinite programming solvers~\cite{brandao2017quantum,van2017quantum,brandao2019quantum}, and is generally expected to be useful in optimization due to the applicability of its classical analogue of Metropolis sampling~\cite{metropolis1953equation}.

In nature, a physical system in thermal equilibrium with a heat bath can be modeled as a system interacting with its environment in the weak-coupling infinite-time limit. The dynamics of this open system are governed by the Lindblad equation; in particular, the weak-coupling infinite-time limit corresponds to a choice of quantum map $\cL$ known as the Davies generator~\cite{davies1974markovian,davies1976quantum,davies1979generators}, whose time evolution channel $e^{\cL t}$ maps any input state to the Gibbs state as $t\to\infty$. The Davies generator has a fixed point $\cL(\rho_\beta) = 0$ and satisfies detailed balance (i.e., reversibility). However, the implementation of such generators is challenging due to potential non-locality and lack of smoothness in the construction of the Lindbladian. Finding an efficient alternative to the Davies generator that is implementable on a quantum computer has been a long-standing challenge. Recent developments have achieved significant success in this regard~\cite{chen2023efficient,gilyen2024quantum,ding2024efficient}, constructing Lindbladians $\cL$ that overcome some of the limitations posed by Davies generators yet also satisfy the fixed point and detailed balance properties.

To guarantee that a Lindbladian Gibbs sampler prepares the Gibbs state efficiently from any initialization, one must upper-bound the mixing time associated with approximating the Gibbs distribution to within a specified accuracy. More formally, given a Lindbladian $\cL$, $\tmix(\cL)$ is the smallest time for which, for any states $\rho$ and $\sigma$, $\norm{e^{\cL T}[\rho-\sigma]}_1 \leq \epsilon \norm{\rho - \sigma}_1 $ for constant $\epsilon \leq \frac{1}{2}$:
\begin{equation}
\label{eq:tmix}
    \tmix := \max_{\rho, \sigma} \min \left\{ T: \norm{e^{\cL T}[\rho-\sigma]}_1 \leq \epsilon \norm{\rho - \sigma}_1 \right\}.
\end{equation}
In this work, we are interested in showing mixing time \emph{lower bounds} that apply to the recent line of Gibbs samplers inspired by the Davies generator~\cite{chen2023efficient,gilyen2024quantum,ding2024efficient} and the Metropolis algorithm~\cite{jiang2024quantum}. Prior work showing slow mixing is typically restricted to the special case of commuting Hamiltonians obtained explicitly from quantum stabilizer codes~\cite{dennis2002topological,alicki2010thermal,hong2024quantum}. Here, we move beyond code Hamiltonians to answer the following open questions.\footnote{Work that appeared simultaneously on arXiv also addresses the second open question. \cite{rakovszky2026bottlenecks} shows slow mixing under quasi-local Gibbs samplers of all-to-all commuting quantum code Hamiltonians with linear distance, subjected to additional non-commuting terms that do not break topological order. We discuss this result later.}
\begin{enumerate}[label=(\roman*)]
    \item Given a classical Hamiltonian that mixes slowly for classical Metropolis-like algorithms, is a quantum Gibbs sampler efficient?
    \item Given a non-commuting Hamiltonian, can superpolynomial mixing time lower bounds be shown for quantum Gibbs samplers?
\end{enumerate}
In this work, we will provide lower bounds --- in particular, lower bounds that are exponential in the system size for several models --- answering these questions. Our results will be obtained in the low-temperature regime, where the temperature is below some constant threshold (system size independent).

Our work relies heavily on developing a formal notion of clustering in quantum Gibbs states. Classically, proofs of hardness for a large class of optimization and sampling algorithms, such as stable, low-degree, or local algorithms, often rely on showing a clustered geometry of the solution space (sometimes referred to as ``glassiness" of the Hamiltonian). For example, a classical Gibbs measure is considered clustered if it is almost entirely supported on sets of bitstrings that are far apart in Hamming distance; the remaining bitstrings form a bottleneck subspace that is ultimately responsible for algorithmic obstruction. This and similar geometric pictures of clustering appear in many classical proof techniques throughout the computer science and physics literature, such as the overlap gap property~\cite{gamarnik2021overlap}, transport disorder chaos~\cite{el2022sampling}, and replica symmetry breaking~\cite{mezard1987spin}.

The aforementioned classical definitions of clustering have clear meaning in terms of the eigenbasis of the Hamiltonian. In this work, we develop a formal notion of clustering for quantum Hamiltonians in terms of subspaces that approximately block diagonalize the Hamiltonian, providing an analogue to these classical techniques. We show this clustering implies hardness for quantum Gibbs samplers. Several obstacles are introduced by moving to the quantum setting, and in particular in the non-commuting setting:
\begin{itemize}
    \item \emph{Measure of distance.} Classically, clusters of low energy states that are far apart in Hamming distance obstruct Gibbs samplers. What is the analogous quantum measure of distance? We desire exponential mixing time lower bounds for Gibbs samplers that can often act non-locally for non-commuting classical Hamiltonians. Naive light cone based arguments or local approximations of Lindbladians via Lieb-Robinson bounds do not directly apply in these settings. 
    \item \emph{Non-commuting Hamiltonians beyond perturbation theory.} Perturbation theory can directly address settings with a commuting Hamiltonian perturbed by a non-commuting Hamiltonian with asymptotically vanishing norm. In our setting with perturbations of constant-sized terms, the eigenbasis and spectrum change nontrivially and perturbation theory alone cannot extend slow mixing results from commuting Hamiltonians to this setting.
    \item \emph{Clustering in a different basis from the Hamiltonian.} Due to the lack of control over eigenstates, we resort to showing a clustered description of non-commuting Hamiltonians' Gibbs states in a basis distinct from the eigenbasis. It is not obvious (1) how to find the clustering basis, or (2) that such a basis even exists.
    \item \emph{Slow mixing from clustering.} Classically, slow mixing of local samplers can be shown by proving that there exists a bottleneck in the Gibbs measure which separates high probability clusters which have pairwise distance larger than the locality of the sampler. For non-commuting Hamiltonians, our proof uses a block-diagonal form of clustering, which we refer to as \emph{strong clustering}.
\end{itemize}

\subsection{Summary of our results}

Our technical contributions fall in three categories: developing a fundamentally quantum description of clustering in Gibbs states, showing a quantum non-commutative analogue of the classical bottleneck theorem for slow mixing, and applying the above methods to example systems. 
Difficulties in formalizing clustering can occur when the clusters are written in a different basis from the eigenbasis of the Hamiltonian. 
Consequently, we provide two frameworks for defining clustering, depending on if the eigenbasis is known.
\begin{itemize}
    \item \emph{Clustering in the Hamiltonian eigenbasis.} A measure of distance can be defined in terms of transition operators of the Lindbladian, which are directly expressible in terms of the Hamiltonian eigenbasis. This notion of clustering can give lower bounds on Gibbs samplers with highly non-local Lindblad operators. For example, we show slow mixing of quantum LDPC code Hamiltonians without assuming quasi-locality of the Lindbladian. The same argument shows that classical Hamiltonians can mix slowly with quantum Gibbs samplers via this notion of clustering.
    \item \emph{Clustering in quantum Hamming distance.} Here, we measure distance between two states as the minimum number of local unitaries that are needed to be applied to gain nonzero overlap between the two states. This yields clusters diagonal in a basis distinct from the eigenbasis. One of our main technical results shows that, under suitable clustering and locality assumptions, a bottleneck with sufficiently small Gibbs support is approximately block diagonal in the cluster basis. Given our slow mixing lemma (described below), this provides a concrete structural/geometric description of Gibbs states that mix slowly due to a bottleneck.
\end{itemize}

We formalize this block-diagonal structure of the Gibbs state below, formulated for a general notion of distance $d$.


\begin{definition}[Strong clustering of Gibbs states]
    Let $H$ be an $n$-qubit Hamiltonian acting on a $2^n$-dimensional Hilbert space and let 
    $\rho_\beta \propto e^{-\beta H}$
    be its Gibbs state. Suppose there exists a collection of $M \geq 2$ orthogonal projectors $\Pi_1,\dots,\Pi_M$ satisfying
    \begin{equation}
    d(\Pi_i,\Pi_j) \ge d_0(n) \quad \text{for all } i\neq j,
    \end{equation}
    for some distance measure $d$ and $d_0(n)>0$. Define the complementary projector (or \emph{bottleneck})
    \begin{equation}
    \Pi_\perp := I - \sum_{i=1}^M \Pi_i.
    \end{equation}
    Then $\rho_\beta$ is \emph{strongly clustered} under the distance measure $d$ if there exist positive constants $\gamma_\perp$, $\gamma_0$, and $\alpha$ such that the following conditions hold for sufficiently large $n$:
    \begin{align}
        \text{(low bottleneck weight): }&\Tr[\Pi_\perp \rho_\beta] \le \exp\bigl(-n^{\gamma_\perp}\bigr), \label{eq:weak1}\\
        \text{(high cluster weight): }&\Tr[\Pi_i \rho_\beta] \ge \frac{1}{n^{\gamma_0}}, \quad \forall\, i=1,\dots,M, \label{eq:weak2}\\
        \text{(cluster in blocks): }&\frac{1}{2}\Big\|\rho_\beta - \tilde{\rho}_\beta\Big\|_1 \le \exp\bigl(-n^\alpha\bigr), \text{ where } \tilde{\rho}_\beta = \Pi_\perp \rho_\beta \Pi_\perp + \sum_{i=1}^M \Pi_i \rho_\beta \Pi_i.\label{eq:weak3}
    \end{align}
\end{definition}

Classically, a geometric notion of clustering in Hamming distance heuristically corresponds to the ``glassiness" of systems that generically causes large families of algorithms, such as stable, low-degree, or local algorithms to fail \cite{gamarnik2021overlap,gamarnik2024hardness,huang2022tight}. An analogous understanding of algorithmic hardness due to clustering in quantum Hamiltonians is largely absent from the literature. In this work, we show that strong clustering obstructs Lindbladian Gibbs samplers. These algorithms prepare states of the form $\rho(t) = e^{\cL t}(\rho_0)$ via a linear map (Lindbladian) $\cL$. To compare the locality of the channel to the clustered geometry, we introduce the notion of the \emph{range} of the Lindbladian.

\begin{definition}[Range of a map (informal)]
Given a distance $d(A,B)$ between orthogonal subspaces $A,B$, a Hermiticity-preserving linear map $\cM$ acting on $n$ qubits has range at most $r$ if it admits a decomposition of the form
\begin{equation}
 \cM(X)=\sum_j s_jK_jXK_j^\dagger,\qquad s_j\in\{-1,1\},
\end{equation}
where every $K_j$ has range at most $r$. For the distance under consideration, this means that $\Pi_CK_j\Pi_A=0$ whenever $d(A,C)>r$.
For a quantum channel, range is defined through the channel's Kraus representation and all signs $s_j$ are positive.
\end{definition}
The decomposition $\cM(X)=\sum_j s_jK_jXK_j^\dagger$ can be constructed by diagonalizing the Choi matrix of $\cM$~\cite{watrous2018theory}. 
When using the term range for a Lindbladian Gibbs sampler, the range bounds the terms in the Lindbladian including any coherent part.

Our main slow mixing lemma resembles the classical bottleneck lemma for slow mixing of Markov chains. 
\begin{lemma}[Slow mixing from strong clustering, informal]
\label{lem:bottleneck}
    Let $\rho_\beta = e^{-\beta H}/\Tr(e^{-\beta H})$ be a strongly clustered Gibbs state under distance measure $d$, clusters $\Pi_1, \dots, \Pi_M$, and bottleneck subspace $\Pi_\perp$. Let $\cL$ be a Lindbladian Gibbs sampler with fixed point $\rho_\beta$. If for all distinct cluster subspaces $\Pi_i, \Pi_j$
    \begin{equation}\label{eq:lindbladcond}
        \range(\cL) < d(\Pi_i, \Pi_j),
    \end{equation}
    then the mixing time is lower-bounded by
    \begin{equation}
        T_\mathrm{mix} = \exp[\Omega(\poly{n})].
    \end{equation}
\end{lemma}

\begin{remark}
    We do not require $\cL$ to satisfy detailed balance.
\end{remark}

The slow mixing lemma can be seen as a generalization to quantum Gibbs samplers of bottleneck arguments for the slow mixing of classical Metropolis-like algorithms. In the classical setting of a Hamiltonian defined over length-$n$ bitstrings, a bottleneck argument (see, e.g., \cite{levin2017markov} for a pedagogical introduction) uses three orthogonal regions $A, B, C \subset \{0,1\}^n$. The regions are chosen such that $A$ and $C$ have large support in the Gibbs measure $\pi_\beta$, while the bottleneck region $B$ has support $\pi_\beta(B) = \exp\left[-n^\alpha\right]$ for some $\alpha>0$. For simplicity, we take for now $B = (A \cup C)^c$. If regions $A$ and $C$ are separated by Hamming distance $d$, and the classical Markov process flips fewer than $d$ bits per step, then an initial configuration in $A$ must necessarily pass through $B$. Since $\pi_\beta$ is the stationary distribution of the Markov process, it can be shown that time $\exp[n^\alpha]$ is required to traverse the bottleneck and acquire large support in $C$.

As mentioned earlier, commuting Hamiltonians are easier to analyze since we have explicit access to their eigenbasis. Lemma~\ref{lem:bottleneck} provides an almost immediate implication for the slow mixing of classical Hamiltonians.

\begin{theorem}[Slow mixing of classical Hamiltonians, informal]
\label{thm:classical}
    Let $H$ be a Hamiltonian diagonal in the $Z$ basis, and let $\Pi_A, \Pi_C$ be orthogonal projectors diagonal in the $Z$ basis such that $A$ and $C$ are separated by a Hamming distance that increases with system size. Let $\Pi_B = I - \Pi_A - \Pi_C$ and $\rho_\beta$ be a Gibbs state at inverse temperature $\beta$ such that
    \begin{equation}
        \Tr[\Pi_B \rho_\beta] = \exp[-\Omega(\poly{n})], \quad \mathrm{and} \quad \Tr[\Pi_A \rho_\beta], \Tr[\Pi_C \rho_\beta] = \Omega(1).
    \end{equation}
    Then a Lindbladian Gibbs sampler $\cL$ with fixed point $\rho_\beta$ and bounded jump operator range satisfies
    \begin{equation}
        \tmix(\cL) = \exp[\Omega(\poly{n})].
    \end{equation}
\end{theorem}
Informally, the above theorem states that if a classical Hamiltonian $H$ mixes slowly under Metropolis-like sampling algorithms due to a bottleneck, then it will necessarily mix slowly under a quantum Gibbs sampler. The conditions on regions $A, B, C$ in Theorem~\ref{thm:classical} are identical to those required to show slow mixing of classical Metropolis-like algorithms. To give concrete examples of systems, our results imply a superpolynomial mixing time of the Gibbs samplers of \cite{davies1974markovian,chen2023efficient,gilyen2024quantum} at sufficiently small constant temperature for the Curie-Weiss model, $(d\geq 2)$-dimensional Ising and hard-core models, a $p$-spin glass, and random $K$-SAT instances \cite{levin2017markov,gamarnik2021overlap,gamarnik2023shattering,achlioptas2006solution}.

\begin{remark}
     We postpone the precise definition of a Lindbladian Gibbs sampler until the main text; it suffices to know that the Gibbs samplers of \cite{davies1979generators,gilyen2024quantum,chen2023efficient} are included. We use the term \emph{constant-local jump operators} consistently with prior literature; these are typically chosen to be single-qubit Paulis. The resulting Lindblad operators may become quasi-local or non-local, e.g., due to twirling the jump operators with time evolution as in the Gibbs sampling algorithms of \cite{gilyen2024quantum}.
\end{remark}

We also establish slow mixing of good quantum code Hamiltonians by identifying clusters in the explicitly known eigenbasis.  This extends results from prior work \cite{hong2024quantum} to include slow mixing for a broader set of quantum Gibbs samplers.

\begin{theorem}[Slow mixing of good quantum code Hamiltonians, informal]
\label{thm:codes}
Let $H = \sum_i (I-C_i)/2$ for $C_i$ corresponding to the check operators of an $n$-qubit good quantum LDPC code. Assume the checks satisfy $(\gamma,\alpha)$ small-set expansion (Definition~\ref{def:small_set_expander}). At inverse temperature $\beta>\log(2)/(\gamma\alpha)$, a Gibbs sampler with fixed point $\rho_\beta$ and bounded jump-operator range, defined from constant-local input jump operators, satisfies
\begin{equation}
 \tmix(\cL)=\exp[\Omega(n)].
\end{equation}
Suitable quantum Tanner codes satisfy the required expansion property \cite{leverrier2022quantum,hopkins2022explicit,anshu2023nlts}.
\end{theorem}



We now turn to the more technically difficult setting of non-commuting Hamiltonians. 

\begin{theorem}[Sufficient conditions for strong clustering]\label{thm:mainstrong}
    Assume Hamiltonian $H$ with $\Tr[H]=0$ is $\ell$-local and its Gibbs state $\rho_\beta$ has a quantum Hamming distance $d_0/\ell\ge n^\eta$ for some constant $\eta>0$ between clusters and satisfies the cluster weight and bottleneck weight conditions
    \begin{equation}
        \Tr[\Pi_\perp \rho_\beta] \leq \exp[-n^{\gamma_\perp}], \quad \Tr[\Pi_i \rho_\beta] \geq n^{-\gamma_0} \,\,\, \forall i \in [M].
    \end{equation}
    If there exists positive integer $K \geq 1$ and such that for some constant $\delta > 0$,
    \begin{align}
        \frac{10\ell \beta \norm{H}}{d_0} < 2^K < \min\left\{n^{\gamma_\perp-\delta}, \exp\left[\frac{d_0}{2\ell}\right]\right\},
    \end{align}
    then $\rho_\beta$ is strongly clustered.
\end{theorem}

The proof of Theorem~\ref{thm:mainstrong} relies on the repeated application of matrix inequalities to a telescoping decomposition of $\Pi_i \rho_\beta \Pi_j$ for distinct clusters $i \neq j$. The telescoping decomposition is related to the Trotterization of imaginary time evolution between clusters into $2^K$ steps; we then control the event of jumping across the bottleneck in a single step. Although the parameter $K$ is slightly obscured by the proof, we consider two standard cases here.
\begin{itemize}
    \item \emph{Perturbed good codes.} If $d_0 = \Theta(n)$ and $\Tr[\Pi_\perp \rho_\beta] = \exp[-\Omega(n)]$, then $K=O(1)$ suffices to show strong clustering when the cluster-weight condition hold. This occurs naturally when adding local terms that preserve the topological order of LDPC codes \cite{hastings2011topological,yin2024eigenstate}.
    \item \emph{Geometrically local Hamiltonians ($D \geq 3$).} On a $D$-dimensional lattice, fault-line arguments typically produce $d_0 = \Theta(n^{1-1/D})$ and $\Tr[\Pi_\perp \rho_\beta] = \exp[-\Omega(n^{1-1/D})]$. Strong clustering follows for $D \geq 3$ lattices by choosing $K = \Theta(\log n^{1/D})$. Examples may include the 3D Ising model or a perturbed 4D toric code.
\end{itemize}

Interestingly, the parameters of Theorem~\ref{thm:mainstrong} do not suffice to show clustering of 2D Hamiltonians obtained from fault-line arguments as in the 2D Ising model with a constant-strength transverse field (which is strong enough for degenerate perturbation theory to fail). Nonetheless, we can prove strong clustering and show slow mixing for this model using a related proof technique that is unique to the stoquastic setting of this model.

\begin{theorem}[Slow mixing of the 2D transverse field Ising model, informal]
\label{thm:tfimslow}
    Let $\rho_\beta$ be the thermal state of the 2D transverse field Ising model
    \begin{equation}
        H = -\sum_{\langle i, j \rangle} Z_i Z_j - h\sum_{i=1}^n X_i,
    \end{equation}
    where $\langle i, j \rangle$ denotes unique pairs of neighboring vertices on a 2D lattice of dimension $\sqrt n \times \sqrt n$. 
    Let $\cL$ be a Lindbladian with fixed point $e^{\cL t}(\rho_\beta) = \rho_\beta$ such that each Lindblad operator (including coherent terms) acts nontrivially only on a geometrically local region of at most $R \times R$ qubits.
    There exists constant $\beta_*$ such that for all constant $\beta \geq \beta_*$, the mixing time to prepare $\rho_\beta$ satisfies
    \begin{equation}
        \tmix(\cL) = \exp[\Omega(\sqrt n)]
    \end{equation}
    for all $0 \leq h \leq h_*$ and $R \leq \lambda_*\sqrt n$, where $h_*(\beta), \lambda_*(\beta)$ are constants independent of $n$.
\end{theorem}
\begin{corollary}
\label{cor:tfimslowchen}
    For a 2D transverse field Ising model with constant $h, \beta$ satisfying the conditions of Theorem~\ref{thm:tfimslow}, the Gibbs sampler $\cL$ of~\cite{chen2023efficient} constructed with single qubit jump operators has superpolynomial mixing time $\tmix(\cL) \geq \exp[n^{1/2-o(1)}]$. 
\end{corollary}

We note that slow mixing here occurs for transverse fields of constant (with respect to $n$) strength but dependent on $\beta$. This slow mixing regime is within the ferromagnetic phase of the model where the coupling terms $Z_iZ_j$ dominate the energy contribution. This phase is predicted to occur at values of $h$ below a critical strength $h^*\approx 3$ at low temperatures \cite{pfeuty1971ising,PhysRev.124.768,blass2016test}.

The proof of Theorem~\ref{thm:tfimslow} uses a bottleneck region constructed in the computational basis from the classical bottleneck $B \subset \{0,1\}^n$. Our method is based on the Poisson Feynman-Kac formalism~\cite{kac1974stochastic,crawford2007thermodynamics,leschke2021free}, where the transverse field is treated as the generator of a continuous-time Markov chain. Historically, the underlying ideas of this probabilistic interpretation extend back at least to the path-integral quantum Monte Carlo method proposed by~\cite{suzuki} applicable to stoquastic Hamiltonians~\cite{bravyi2006complexity}. This allows the Gibbs state to be expressed as a collection of Poisson point processes. For any $\Pi$ diagonal in the computational basis, this permits us to compute $\Tr[\Pi \rho_\beta]$ for any Hamiltonian of the form $H = H_0 - h\sum_{i=1}^n X_i$, where $H_0$ is a Hamiltonian diagonal in the $Z$ basis. This method of computation enables us to control the magnitudes of matrix elements needed to establish strong clustering and slow mixing. Roughly, we do so by discretizing timesteps of the Poisson process over imaginary time with discretization $\sim \beta/n^2$ small enough so that the total contribution from transitions between clusters in one step is superpolynomially small. Then we union bound over all $O(n^2)$ timesteps.

We then apply a Peierls-like argument to bound the support of the bottleneck region, which is constructed from fault lines across the 2D lattice. Under the transverse field, fault lines become ``dressed'' according to the path integral given by the Poisson Feynman-Kac formula. We show that the dressed fault lines provide a bottleneck that separates clusters under the quantum Hamming distance. For a sufficiently local Lindbladian Gibbs sampler, the bottleneck $\Pi_B$ is unavoidable. For a concrete result, we show using Lieb-Robinson bounds that the Gibbs sampler of~\cite{chen2023efficient}, which is quasi-local on the lattice, mixes slowly.


\subsection{Related work}

Tools to upper or lower bound the mixing times for quantum Gibbs samplers remain limited; specific results known in the literature will be given below. Most works focus on upper bounds and either address commuting Hamiltonians (using log-Sobolev inequalities~\cite{kastoryano2013quantum,bardet2024entropy,bardet2023rapid,capel2020modified,bardet2022approximate,bardet2021modified,gao2022complete,kochanowski2024rapid}) or irreversible quantum Gibbs samplers (using hypocoercivity~\cite{fang2024mixing,li2024quantum}). Few results are known showing provable lower bounds on the mixing time of these quantum Gibbs samplers. Available proofs are largely restricted to mixing times of commuting Hamiltonians such as those for quantum codes~\cite{dennis2002topological,alicki2010thermal,hong2024quantum}.

Bounds on the mixing time for classical Hamiltonians 
are understood rather well for classical
sampling algorithms such as Markov Chain Monte Carlo (MCMC) and its specific
variant called Glauber dynamics. Polynomial upper
bounds on mixing time can be obtained by considering isoperimetric properties
of the underlying chain state space, specifically using 
Poincar\'e/Cheeger/log-Sobelev inequalities,
or other methods, such as coupling~\cite{levin2017markov}.
Lower bounds on mixing time can be obtained
by identifying subsets of the state space for which the (probabilistic) flow
out of the subsets is low, often exponentially low. These subsets then
create
barriers to fast mixing sometimes called free energy barrier. A canonical example
of slow mixing example is the 
Ising model on a lattice~\cite{thomas1989bound,levin2017markov},
and hard-core model, also on a lattice
\cite{randall2006slow}. The Ising model on the lattice
is one of the models we consider in this paper.

Slow mixing has been analyzed also in the context of random structures such as 
random graphs, random constraint satisfaction problems (random CSPs) and (classical) spin glasses.
In this case, identifying sets with low flow-out directly is difficult. Instead,
the existence of such sets is proven probabilistically roughly along the following
lines. First, the property known as \emph{overlap gap property} (OGP) is proven to take
place with high probability. The property states that
the set of pairwise distances between states with low energy  
exhibit  a gap~\cite{mezard2009information,gamarnik2021overlap,gamarnik2022disordered,gamarnik2023disordered}. The presence of the
gap implies that the low energy states can be partitioned into clusters, separated
by large distance and large energy barriers, thus creating
free energy barriers~\cite{ben2018spectral}. The presence of such clusters
in random CSPs was established in~\cite{achlioptas2006solution} and~\cite{mezard2005clustering}, and 
in spin glasses it was established 
in~\cite{talagrand2000rigorous,gamarnik2023shattering}. For a book
treatment see~\cite{talagrand2010mean}.
Such clusters constitute examples
of subsets with low flow-out rates, thus proving slow mixing. The clustering
structure of the solutions space is often called dynamical replica symmetry
breaking in the statistical physics literature~\cite{mezard2009information}. 

Our results concern mixing time lower bounds that apply to the successful line of recent work on Lindbladian Gibbs samplers~\cite{rall2023thermal,chen2023quantum,chen2023efficient,wocjan2023szegedy,ding2024single,ding2024efficient,gilyen2024quantum} based on the Davies generator~\cite{davies1974markovian,davies1976quantum,davies1979generators}. At the same time though we should acknowledge the existence of Gibbs samplers based on other primitives such as quantum phase estimation and singular value transformation~\cite{temme2011quantum,chowdhury2016quantum,gilyen2019quantum}. We briefly review what is known of the mixing time~\eqref{eq:tmix} of Lindbladian Gibbs samplers with fixed point $e^{\cL t}(\rho_\beta) = \rho_\beta$.

The mixing time can be analyzed by bounding the spectral gap associated with the map $\cL$. For example, lower bounds on the spectral gap using correlation decay and Poincar\'e inequalities can show upper bounds on the mixing time of commuting Hamiltonians~\cite{kastoryano2016quantum,temme2017thermalization}. (In comparison, our results concern \emph{lower} bounds on the mixing time.) For specific commuting models related to error correction, a notable line of work shows results related to the Arrhenius law~\cite{bravyi2009no,brown2016quantum}: heuristically, the Arrhenius law anticipates that the mixing time is exponential in the energy barrier given by the maximal energy of intermediate states in a sequence of local transformations that take a ground state to an orthogonal ground state, minimized over all such sequences. However, the connection between the energy barrier and the thermalization time is not fully rigorous. Instead, the Arrhenius law only provides an \emph{upper} bound on the mixing time~\cite{bravyi2013quantum,michnicki20143d,yoshida2014violation}. A more careful characterization of the mixing time requires analyzing \emph{free} energy, which includes the entropic contribution.

In the self-correcting quantum memory literature, related upper bounds on the mixing times of error correction codes and stabilizer Hamiltonians are found in~\cite{alicki2009thermalization,komar2016necessity,lucia2023thermalization,temme2015fast}. Several lower bounds on mixing time (with respect to the Davies generator) are also known. The most important of these works are the 4D toric code lower bounds~\cite{dennis2002topological,alicki2010thermal}, which use a Peierls-like argument to show slow mixing (and requires a Lindbladian that satisfies detailed balance). Work related to codes under weak local perturbation includes~\cite{hastings2011topological}, which rules out topological order from local commuting Hamiltonians on the 2D lattice. Perhaps the most relevant result in this line of work is~\cite{chesi2010thermodynamic}, which considers non-commuting Hamiltonians due to subsystem codes and shows a polynomial mixing time lower bound (as opposed to the superpolynomial lower bounds shown here, which we refer to as \emph{slow} mixing) via the structure provided by logical operators in the code. We note that related results to ours appeared simultaneously on the arXiv~\cite{rakovszky2026bottlenecks}, showing slow mixing of perturbed good codes via a similar bottleneck lemma, under the assumption that the locality of the Lindbladian Gibbs sampler is smaller than the code distance.

Another line of work on mixing time bounds depends on log-Sobolev (and modified log-Sobolev) inequalities~\cite{kastoryano2013quantum,bardet2024entropy,bardet2023rapid,capel2020modified,bardet2022approximate,bardet2021modified,gao2022complete,kochanowski2024rapid} to show upper bounds on mixing time. These results are largely restricted to local commuting Hamiltonians, which imply locality of the Davies generator. We note the work of~\cite{chen2021fast} instead uses a modified log-Sobolev inequality under the eigenstate thermalization hypothesis. Generalizing these techniques to non-commuting Hamiltonians without such structure remains an open question~\cite{kochanowski2024rapid}. A distinct approach to upper-bounding mixing times is based on hypocoercivity~\cite{fang2024mixing,li2024quantum}. Although these techniques can provide mixing time upper bounds for non-commuting models such as the 1D transverse field Ising model, they only apply to Lindbladian Gibbs samplers with broken detailed balance (i.e., irreversible Markov processes), unlike the Gibbs samplers of~\cite{chen2023efficient,gilyen2024quantum,ding2024efficient}.

Recent results have also analyzed hardness of Gibbs sampling of instantaneous quantum polynomial time (IQP) circuits\footnote{$n$-qubit unitaries for these circuits take the form $U=H^{\otimes n} D H^{\otimes n}$ where $D$ is a constant depth diagonal unitary and $H$ are Hadamard gates.} at constant temperature~\cite{bergamaschi2024quantum,rajakumar2024gibbs}. These works construct a family of commuting Hamiltonians whose Gibbs states can be efficiently sampled quantumly, but, assuming the polynomial hierarchy does not collapse, cannot be efficiently sampled classically. Given block encoding access to a Hamiltonian, the work of~\cite{wang2023quantum} also shows a query lower bound of $\Omega(\beta)$ for Gibbs state preparation.
Finally, the recent work of~\cite{hong2024quantum} shows a mixing time lower bound of $e^{\Omega(\sqrt n)}$ for a specific class of quantum stabilizer codes with distance $d=\Theta(\sqrt{n})$ at low enough constant temperature; they comment that their techniques should extend to show a lower bound of $e^{\Omega(n)}$.

We note that exponential mixing time lower bounds for local Gibbs samplers --- much like mixing time lower bounds on classical Glauber dynamics and other local classical algorithms --- do not preclude the possibility of efficient algorithms that lack such structure. A notable example is the work of~\cite{bravyi2014monte,bravyi2017polynomial}, which provide classically efficient algorithms to estimate the partition function of families of quantum spin glasses such as the ferromagnetic transverse field Ising model, analogous to the seminal algorithm of~\cite{jerrum1993polynomial} for the classical Ising model. However, as it stands now, there is no known polynomial time equivalence between efficient estimation of partition functions and efficient preparation of Gibbs states in the quantum setting. For example, the algorithm of \cite{bravyi2017polynomial} does not directly allow for the preparation of a Gibbs state or the efficient estimation of the thermal expectation of arbitrary local observables.

\subsection{Open questions}

Our work establishes superpolynomial mixing time lower bounds for various quantum Gibbs samplers in the low-temperature regime. Key to our results was a fundamental notion of clustering in the Gibbs state; we expect this may be broadly useful. Several open questions remain, which we outline below.

\paragraph{Lower bounds for non-commuting random quantum Hamiltonians.}

An important direction for future research is extending our mixing time lower bounds to commonly studied \emph{random} quantum Hamiltonians that are non-commuting, such as quantum spin glasses \cite{anschuetz2024bounds,erdHos2014phase}. Classically, slow mixing has been analyzed for random structures like random constraint satisfaction problems and spin glasses using techniques based on the overlap gap property (OGP) and free energy barriers~\cite{gamarnik2021overlap,gamarnik2023disordered,arous2020free}. In the quantum setting, randomness and non-commutativity introduce significant challenges. Since the eigenbasis for these Hamiltonians is not easy to describe, expansion of jump operators in the Bohr frequencies is not directly possible for these systems. Extending our bottleneck arguments to these systems (lacking additional structure such as a transverse field) would involve dealing with the complexity of the quantum state space and addressing the lack of an analogous notion of distance between configurations used to define an OGP (for classical Hamiltonians, the Hamming distance is used to define OGPs). 

\paragraph{Lower bounds for dense Ising models.}

The locality of a Gibbs sampler can often be bounded for models which commute or admit a strong enough Lieb-Robinson bound. \emph{Dense} models, such as the all-to-all connected Curie-Weiss model in a transverse field, can in general produce highly non-local Lindblad operators whose locality cannot be truncated approximately with direct application of Lieb-Robinson bounds. Here, for example, the Hamiltonian takes the form
\begin{equation}
    H = -n^{-1}\sum_{1 \leq i < j \leq n} Z_i Z_j - h \sum_{i=1}^n X_i.
\end{equation}
Without the transverse field ($h=0$), slow mixing is well-understood for these models due to the presence of a free energy barrier that scales linearly with the system size~\cite{levin2017markov}. However, in the quantum case with a transverse field, we currently cannot bound the locality of the Gibbs sampler.

\paragraph{Initialization-dependent mixing time lower bounds.}

Our results here do not directly give mixing time lower bounds for \emph{specific} initializations, such as the completely mixed state $\rho = 2^{-n} I$. Even classically, showing slow mixing from specific initial states is challenging because it often requires precise control over the dynamics starting from that state. Only recently have such slow mixing results been proven for specific initializations in well-studied classical models~\cite{gamarnik2024hardness,sellke2024threshold}. Addressing this question could provide a more nuanced understanding of the thermalization process in quantum systems and might reveal initialization-dependent barriers to efficient Gibbs sampling, potentially giving insights into algorithms that mix faster than predicted by worst-case analysis.

\paragraph{Quantum algorithmic hardness from strong clustering.}

Finally, our work raises the question of whether the techniques developed here can be extended to show lower bounds for quantum algorithms beyond Gibbs sampling. In the classical realm, OGP-based arguments have been used to establish lower bounds against a wide class of algorithms, including those that are robust to certain types of perturbations or noise~\cite{gamarnik2021overlap,huang2022tight}. Given that our work represents a quantum analogue of these arguments, it is natural to ask whether similar lower bounds can be proven for other quantum algorithms, particularly those aimed at solving optimization or sampling problems. We expect that the generic geometric description provided by strong clustering may be useful in this line of work.

\subsection{Organization}

The paper is organized as follows. In \Cref{sec:prelim}, we introduce notation, Lindbladian Gibbs samplers, and the jump operators associated with the Davies generator. In \Cref{sec:cluster}, we define distance between regions of the Hilbert space based on the number of jump operators (\Cref{sec:distdavies}) and number of local operators (\Cref{sec:distlocal}) required to move from one region to another. We then use these notions of distance to formally define strong clustering (\Cref{sec:strong_cluster}). These allow the bottleneck lemma to be formalized in \Cref{sec:slow}, showing sufficient conditions for slow mixing. We present examples of systems with slow mixing in \Cref{sec:systems}, including classical Hamiltonians (\Cref{sec:classical}), commuting quantum Hamiltonians (\Cref{sec:codes}), and finally a more general approach to non-commuting quantum Hamiltonians including the 2D transverse field Ising model (\Cref{sec:localtfim}).

\section{Preliminaries}
\label{sec:prelim}
\subsection{Notation}
We denote the set of $n$ elements as $[n] := \{1, \dots, n\}$. 
For bitstrings $x,y \in \{0,1\}^n$, we denote their Hamming distance as $|x-y|_H$ which is equal to the number of entries at which the values in $x$ and $y$ differ. Given a set of bitstrings $S_1$ and $S_2$, we similarly define their Hamming distance as the minimum distance between bitstrings in $S_1$ and $S_2$.

We denote the Hilbert space of $n$ qubits as $\bH_n = \mathbb{C}^{2^n}$, and the set of states on $n$ qubits as $\cS_n$. We denote the set of bounded operators on $\bH_n$ as $\mathcal{O}_{\bH_n}$.
For $A,B \in \mathcal{O}_{\bH_n}$, we say $A\succeq B$ if $A-B$ is positive semi-definite.
A Hamiltonian is a bounded Hermitian operator acting on $\bH_n$.
Given a matrix $M \in \mathcal{O}_{\bH_n}$, we denote its $p$-Schatten norm as $\|M\|_p$. 
When subscript is not explicitly stated, $\|M\|$ is assumed to be the operator norm $\|M\|_\infty$.
A Gibbs state of a Hamiltonian $H$ at inverse temperature $\beta$ is given by
\begin{equation}
    \rho_\beta := \frac{e^{-\beta H}}{Z(\beta)}, \quad Z(\beta) := \Tr(e^{-\beta H}).
\end{equation}
For a Hamiltonian $H$, we write its eigenbasis as
\begin{equation}
    H = \sum_{i=1}^{2^n} E_i \ket{\psi_i}\bra{\psi_i}.
\end{equation}
We use $\operatorname{Spec}(H) = \{E_i\}$ to denote the spectrum or set of energy eigenvalues for a Hamiltonian $H$. The set of possible differences in energies of eigenstates (commonly termed the Bohr spectrum) is denoted by 
\begin{equation}
    B(H) := \{E_i - E_j: i,j \in [2^n]\}.
\end{equation}
The Pauli operators are denoted by
\begin{equation}
    I = \begin{pmatrix} 1 & 0 \\ 0 & 1 \end{pmatrix}, \quad X = \begin{pmatrix} 0 & 1 \\ 1 & 0 \end{pmatrix}, \quad Y = \begin{pmatrix} 0 & -i \\ i & 0 \end{pmatrix}, \quad Z = \begin{pmatrix} 1 & 0 \\ 0 & -1 \end{pmatrix}.
\end{equation}
As operators acting on $n$ qubits, the Pauli operator $P^a$ which acts on qubit $a$ is
\begin{equation}
    P^a = \underbrace{I \otimes \cdots \otimes I}_{a-1} \otimes P \otimes \underbrace{I \otimes \cdots \otimes I}_{n-a},
\end{equation}
where  $P \in \{X,Y,Z\}$. The set of single qubit Pauli operators is thus $\{P^a: P \in \{X,Y,Z\}, a \in [n]\}$. We denote the $n$-qubit Pauli group by $\cP_n = \{s E_n \otimes \cdots \otimes E_1 | s\in\{-1,1,-i,i\}, E_i \in \{I,X,Y,Z\}\}$.

We formally define a quantum channel. Let $\bH$ be a Hilbert space and $L(\bH)$ be the set of linear operators on $\bH$. A map $\cE: L(\bH) \to L(\bH)$ is a quantum channel if it is a completely positive trace-preserving (CPTP) map; i.e., if for any auxiliary Hilbert space $\bH_n$ of arbitrary dimension $2^n$, $\cE$ satisfies
\begin{equation}
    (\cE \otimes I_{2^n})(X) \geq 0 \text{ for any positive operator }X\text{,\; and } \Tr(\cE(Y)) = \Tr(Y) \text{ for any } Y \in L(\bH).
\end{equation}
Any CPTP map $\cE$ can be represented in terms of Kraus operators $\{K_i\}$, which are a set of linear operators on $\bH$ such that the action of $\cE$ on $\rho$ is given by
\begin{equation}
    \cE(\rho) = \sum_i K_i \rho K_i^\dagger \;\text{ such that } \;\sum_i K_i^\dagger K_i = I,
\end{equation}
where the index $i$ runs over a finite or countably infinite set. We will use $\cL$ to denote a linear map that corresponds to a Lindbladian, which we now define in the context of Gibbs sampling.
We use $I$ for the identity matrix and $\operatorname{Id}$ for the identity CPTP map.

\subsection{Gibbs sampling}
We refer to a linear map $\cL$ as a Lindbladian if it acts on a density matrix $\rho$ as
\begin{equation}
    \cL(\rho) := -i\Big[G, \rho\Big] + \sum_a \left(L^a \rho (L^a)^\dagger - \frac{1}{2}\{(L^a)^\dagger L^a, \rho\}\right),
\end{equation}
where $G$ is a Hermitian operator and $\{L^a\}$ are arbitrary linear operators. We will further refer to a Lindbladian Gibbs sampler as a Lindbladian that satisfies $\cL(\rho_\beta) = 0$ for Gibbs state $\rho_\beta$; this implies $e^{\cL t}(\rho_\beta) = \rho_\beta$ for all $t$. This is analogous to the stationary distribution classically: a Markov chain defined by transition matrix $P$ satisfies detailed balance with respect to distribution $\pi$ if $\pi_i P_{ij} = \pi_j P_{ji}$, implying that $\pi$ is a stationary state of the Markov chain. We note that $\rho_\beta$ being a fixed point of $\cL$ is implied by detailed balance, but our results never require detailed balance.

The \emph{time evolution} under the Lindbladian $\cL$ refers to the application of the quantum channel $e^{\cL t}$ for some $t > 0$. 
Unless stated otherwise, we assume $\|\cL\|_\diamond=n^{O(1)}$.
In proofs, it is often also convenient to refer to the discrete Gibbs sampler $\cT$ evolved for integer time $t$, giving the quantum channel $\cT^t$. We formally introduce the notion of a mixing time $\tmix(\cdot)$ for both the discrete and continuous settings; if the input to $\tmix(\cdot)$ is obvious from context, we will at times simply denote the mixing time as $\tmix$.

\begin{definition}[Mixing time (discrete sampler)]
Given a quantum channel $\cT$ and constant $0<\theta \leq \frac{1}{2}$, let $\tmix(\cT)$ be the smallest time for which, for any states $\rho$ and $\sigma$,
\begin{equation}
    \norm{\cT^{\tmix}[\rho-\sigma]}_1 \leq \theta\norm{\rho - \sigma}_1.
\end{equation}
\end{definition}

\begin{definition}[Mixing time (Lindbladian)]
Given a Lindbladian $\cL$ and constant $0<\theta \leq \frac{1}{2}$, let $\tmix(\cL)$ be the smallest time for which, for any states $\rho$ and $\sigma$,
\begin{equation}
    \norm{e^{\cL \tmix}[\rho-\sigma]}_1 \leq \theta\norm{\rho - \sigma}_1.
\end{equation}
\end{definition}

The mixing time $\tmix(\cL)=\infty$ if the  threshold at $\theta$ is never reached.
We remark that as long as the threshold $\theta$ is small enough, varying the constant $\theta$ (independently of other variables such as the number of qubits $n$) can only change the mixing time by a constant factor.
The mixing time defined above also relates directly to the distance to stationarity.
Indeed, let $\Phi_t$ denote the channel after time $t$, with stationary state $\rho_\beta$, and compare
\begin{equation}
    d(t)=\sup_\sigma\frac12\norm{\Phi_t(\sigma)-\rho_\beta}_1,
    \qquad
    \eta(t)=\sup_{\sigma\ne\tau}
       \frac{\norm{\Phi_t(\sigma-\tau)}_1}{\norm{\sigma-\tau}_1}.
\end{equation}
Then $d(t)\leq\eta(t)\leq2d(t)$. The first inequality follows by taking $\tau=\rho_\beta$. For the second inequality, decompose the positive and negative parts of any traceless Hermitian matrix as multiples of two states, and apply the triangle inequality relative to $\rho_\beta$.
The following lemma compares continuous-time evolution with its first-order approximation.


\begin{lemma}
\label{lem:discrete_convert_to_continuous}
Let $\cL$ be a Lindbladian and let $\widetilde\cL$ be a
Hermiticity-preserving map with
$\|\cL-\widetilde\cL\|_\diamond\le\epsilon$.
For every $\delta\ge0$,
\begin{equation}
 \big\|e^{\delta\cL}-(\operatorname{Id}+\delta\widetilde\cL)\big\|_\diamond
 \le \delta\epsilon+\frac{\delta^2}{2}\|\cL\|_\diamond^2.
\end{equation}
\end{lemma}
\begin{proof}
Taylor's remainder theorem gives
\begin{equation}
 e^{\delta\cL}-\operatorname{Id}-\delta\cL
 =\int_0^\delta(\delta-s)e^{s\cL}\cL^2\,ds.
\end{equation}
Since $e^{s\cL}$ is a channel, its diamond norm $\|e^{s\cL}\|_\diamond=1$. Thus taking diamond norm and using submultiplicativity of diamond norm, $\int_0^\delta(\delta-s)\|\cL^2 \|_\diamond\,ds \le \frac{\delta^2}{2}\|\cL^2 \|_\diamond$. 
Finally use triangle inequality $\|e^{\delta\cL}-\operatorname{Id}-\delta\widetilde\cL\|_\diamond \le \|e^{\delta\cL}-\operatorname{Id}-\delta \cL\|_\diamond + \delta \epsilon$ to conclude.
\end{proof}


\subsection{Jump operators}
\label{sec:jump}
We briefly review the physics background associated with the Lindbladians in the study of open system dynamics. The Lindblad master equation $\dot \rho = \cL(\rho)$ describes the evolution of state $\rho$ in the same Hilbert space as a system Hamiltonian $H$ when the system is coupled to a large bath. The operator $G$ produces coherent time evolution, and the operators $\{L^a\}$ produce dissipative dynamics. In this physical model, the system and bath evolve under total Hamiltonian $H_\mathrm{tot} = H + H_B + H_\mathrm{int}$, where $H_\mathrm{int}$ couples the system Hamiltonian ($H$) to the bath Hamiltonian ($H_B$) in the form $H_\mathrm{int} = \sum_a A^a \otimes B^a$ for \emph{jump operators} $A^a$ on the system and $B^a$ on the bath.

Each jump operator $A^a$ is chosen as a local operator, such as a local Pauli operator. In a classical Markov chain describing the evolution of a probability distribution over bitstring configurations, each $A^a$ would correspond to flipping the $a$th bit.

Given the Bohr spectrum $B(H) := \{E_i - E_j: i,j \in [2^n]\}$ and a jump operator $A^a$, we can write $A^a$ in the eigenbasis of $H$ as
\begin{equation} \label{eq:bohr_decomposition_A}
    A^a = \left(\sum_{i} \ket{\psi_i}\bra{\psi_i} \right) A^a \left(\sum_{i} \ket{\psi_i}\bra{\psi_i} \right) = \sum_{\nu \in B(H)} \sum_{i,j: E_i - E_j = \nu} \bra{\psi_i} A^a \ket{\psi_j} \ket{\psi_i}\bra{\psi_j} = \sum_{\nu \in B(H)} A_\nu^a,
\end{equation}
where 
\begin{equation}
\label{eq:anua}
    A_\nu^a := \sum_{i,j: E_i - E_j = \nu} \bra{\psi_i} A^a \ket{\psi_j} \ket{\psi_i}\bra{\psi_j} = \sum_{i,j: E_i - E_j = \nu} A_{j\to i}^a.
\end{equation}
$A_\nu^a$ governs transitions $\ket{\psi_i}\bra{\psi_j}$ from eigenstates $j$ to $i$ separated by energy $\nu$.
We can decompose Lindbladian Gibbs sampler in terms of these operators $A_\nu^a$ or $A_{j \to i}^a$. The simplest of these is the \emph{Davies generator}, which emerges in the infinite-time limit of a weakly coupled system to a bath. The Davies generator has Lindblad operators of the form $\sqrt{\gamma(\nu)} A_\nu^a$ for a function $\gamma:\R \to \R$ chosen to satisfy detailed balance. However, these Lindblad operators are not necessarily local, making it difficult to implement efficiently. 

References~\cite{gilyen2024quantum} and~\cite{ding2024efficient} overcome this issue by using Lindblad operators of the form
\begin{equation}
    L^a = \sum_{\nu \in B(H)} \sqrt{\gamma(\nu)} A_\nu^a
\end{equation}
for $\gamma$ chosen both to satisfy detailed balance and to be sufficiently smooth for efficient implementation. The earlier approach of~\cite{chen2023efficient} used the Fourier-transformed jump operator (for $\omega \in \R$)
\begin{equation}
    L^a = \frac{1}{\sqrt{2\pi}} \int_{-\infty}^\infty f(t) e^{iHt}A^a e^{-iHt} e^{-i\omega t} dt = \sum_{\nu \in B(H)} \hat f(\omega - \nu) A^a_\nu
\end{equation}
for Gaussian $f$ and its Fourier transform $\hat f$. Our lower bounds will apply to all these constructions.

\section{Strong Clustering}
\label{sec:cluster}

In this section, we formally introduce weak and strong clustering and their relationship for various families of Hamiltonians. We begin by defining two measures of distance: one based on the Hamiltonian eigenbasis, and another based on locality. We then define clustering relative an arbitrary distance measure. Finally, we specialize to the locality distance measure and show sufficient conditions for weak clustering to imply strong clustering.

\subsection{Distance in the Hamiltonian eigenbasis}
\label{sec:distdavies}
In classical Hamiltonians, all eigenstates are diagonal in the computational basis. Hence, a natural generalization of Hamming distance is given by the eigenbasis of the Hamiltonian whose thermal state is being prepared.
Recall from \Cref{sec:jump} that the jump operator $A^a$ can be decomposed in terms of the Bohr spectrum $B(H) := \{E_i - E_j: i,j \in [2^n]\}$ as
\begin{equation}
    A_\nu^a := \sum_{i,j: E_i - E_j = \nu} \bra{\psi_i} A^a \ket{\psi_j} \ket{\psi_i}\bra{\psi_j} = \sum_{i,j: E_i - E_j = \nu} A_{j\to i}^a,
\end{equation}
where $A_\nu^a$ is only supported on transitions where $\bra{\psi_i} A^a \ket{\psi_j} \neq 0$. We also use the notation $A_{j\to i}^a = \bra{\psi_i} A^a \ket{\psi_j} \ket{\psi_i}\bra{\psi_j}$ to denote transitions from eigenstate $j$ to $i$. The Lindbladian for the Davies generator in particular can be written in terms of the operators $A_\nu^a$.

For convenience, we will define the set $\Tset(H, \{A^a\})$ of nonzero $A_{j \to i}^a$ as the set of possible transitions.

\begin{definition}[Set of transitions $\Tset(H, \{A^a\})$]\label{def:set_of_transitions}
    Given a Hamiltonian $H$ with eigenbasis $\{\ket{\psi_i}\}$, including a fixed choice within each degenerate subspace, and a set of jump operators $\{A^a\}$ which is closed under adjoints (if $A^a$ in the set, then $(A^a)^\dagger$ is also in the set), the set of transitions $\Tset(H, \{A^a\})$ consists of all possible nonzero $A_{j \to i}^a$ generated by $\{A^a\}$, i.e. 
    \begin{equation}
        \Tset(H, \{A^a\}) = \left\{ A_{j \to i}^a: \; A^a \in \{A^a\}, \ket{\psi_i}, \ket{\psi_j} \in \{\ket{\psi_k}\}, \bra{\psi_i} A^a \ket{\psi_j} \neq 0  \right\}.
    \end{equation}
\end{definition}

When the inputs to $\Tset(H, \{A^a\})$ are obvious from context, we will at times simply denote the set $\Tset(H, \{A^a\})$ as $\Tset$. The above definition recovers the standard support of classical Markov chain samplers as detailed below.

\begin{example}[Classical Markov chain sampler]
    For classical Markov chains, the eigenbasis of a Hamiltonian is the computational basis and the jump operators $\{A^a\}_{a=1}^n$ are given by the set of single qubit Pauli $X$ operators where $A^a = X^a$.
    In this case, $A_{y \to x}^a = \bra{x}X^a\ket{y}\ket{x}\bra{y}$ and supports transitions between computational basis states $\ket{x}$ and $\ket{y}$ which differ only in the $a$-th bit. Therefore, 
    \begin{equation}
        \Tset(H, \{A^a\}) = \left\{ \ket{x}\bra{y}: \; x,y \in \{0,1\}^n , |x-y|_H = 1 \right\}.
    \end{equation}
\end{example}

The Gibbs samplers described in \Cref{sec:jump} all have a common feature: each Lindblad operator $L^a$ is expressible as a sum of transition operators $A_{j \to i}^a$. Here, we define a notion of range valid for channels and any Hermiticity-preserving map in terms of its degree in transition operators. This jump operator range will capture the number of hops achievable by a single application of Lindbladian $\cL$.

\begin{definition}[Jump operator range]
\label{def:jumpdist}
    Given a quantum channel $\cM$, the jump operator range $\trange(\cM)$ is equal to the minimum positive integer $k$ such that there exist Kraus operators $K_1, \dots, K_D$ which can be written as
    \begin{equation}
        K_i = D_i+\sum_{r=1}^k \sum_{A_1, \dots, A_r \in \Tset} c_{i,r}(A_1, A_2, \dots, A_r) A_1 A_2 \cdots A_r
    \end{equation}
    where $D_i$ is diagonal in the chosen eigenbasis and functions $c_{i,r}:\Tset^r \to \mathbb{C}$ such that $\cM(\rho):=\sum_{i=1}^D K_i \rho K_i^\dagger$. We say that $\cM$ is a \emph{degree-}$k$ channel with respect to $\Tset$.
    For a Hermiticity-preserving map, the same definition applies to the decomposition allowing for signs: $\cM(\rho)=\sum_i s_i K_i\rho K_i^\dagger$ with $s_i\in\{\pm1\}$.
\end{definition}
In words, a quantum operator $\cM$ has range $k$ if it can be written as a degree at most $k$ operator in the set $\Tset$. The range of a discrete time sampler $\cT$ or continuous time Lindbladian $\cL$ can be restricted as shown in the previous section by analyzing the set $\Tset(H,\{A^a\})$ (see \Cref{def:set_of_transitions}). Jump operator range is immediately applicable to the Gibbs samplers described explicitly in \Cref{sec:jump}; we show here an example of its application to a Gibbs sampler with a less explicit relation to the Davies generator.

\begin{example}[Range of sampler from \cite{jiang2024quantum}]
    The Gibbs sampler $\cM$ of \cite{jiang2024quantum} is a degree $2$ operator in the transitions $A_{j \to i}^a$.  Each iteration of this sampler proceeds in three steps which can be summarized as: 
    \begin{enumerate}
        \item A unitary $U_C= \operatorname{QPE}_{1,3} \circ A^a \circ \operatorname{QPE}_{1,2}$ is applied. Here, given a starting state $\ket{\phi}$, $\operatorname{QPE}_{1,2}$ estimates the energy of $\ket{\phi}$ using quantum phase estimation outputting this estimate in register $2$. Then it applies a jump operator $A^a$ and finally applies quantum phase estimation again to estimate the new energy storing the result in register $3$.
        \item A sequence of measurement is applied to the ancillary registers. 
        \item The inverse unitary $U_C^\dagger$ is then potentially applied depending on the outcome of the measurements.
    \end{enumerate} 
    Given a set of jump operators $\{A^a\}$, the range of this sampler $\cM$ is $\trange(\cM)\le2$ since the unitary $U_C$ is degree $1$ in jumps $A^a_{j \to i}$ and this operation is applied at most twice (measurements on ancillary measurements cannot add further jumps between eigenbases). See further details in \Cref{app:ex_operator_range_bound}.
\end{example}

The jump operator range of a discrete sampler or Lindbladian will be compared to a distance between subspaces of the Hilbert space: the bottleneck will be constructed such that if two regions $A$ and $C$ are further apart than the range of the operator, then the operator will necessarily pass through the bottleneck. To make this formal, we define a notion of distance between subspaces.

\begin{definition}[Jump operator distance]
    Given a Hamiltonian on $n$ qubits with eigenbasis 
    \begin{equation}
        H = \sum_{i=1}^{2^n} E_i \ket{\psi_i}\bra{\psi_i},
    \end{equation}
    let $B,C \subset \bH_n$ be any two subspaces. Given a set of jump operators $\{A^a\}_a$, let $\Tset(H, \{A^a\})$ denote the set of transitions defined in \Cref{def:set_of_transitions}. 
    The jump operator distance $\dist_{\Tset}(B, C)$ is defined as
    \begin{equation}
        \dist_{\Tset}(B, C) = \min \bigl\{d \ge 0: \Pi_CK\Pi_B\neq 0 \text{ for an operator $K$ of degree at most $d$} \bigr\}.
    \end{equation}
    The degree above is in terms of the jump operator range defined in \Cref{def:jumpdist}.
\end{definition}

In words, the jump operator distance $\dist_{\Tset}(B, C)$ quantifies the minimum number of jump operators that one needs to apply to move from a state in $B$ to a state in $C$.

\subsection{Distance from locality}
\label{sec:distlocal}

When the eigenbasis of the Hamiltonian is difficult to control (e.g., for non-commuting Hamiltonians), it is often easier to use a measure of distance analogous to that of error correction codes: the distance between two states is the number of local unitaries required to achieve nonzero overlap between them. We will find that this is useful even for non-code Hamiltonians, and hence we refer to it by a more generic name of \emph{quantum Hamming distance}.
\begin{definition}[Quantum Hamming distance]
\label{def:distlocal}
    Given states $\ket{\phi},\ket{\psi}$, define their quantum Hamming distance as 
    \begin{equation}
        d_H(\phi,\psi) := \min_k \{k: \exists k\text{-local unitary } U \text{ s.t. }  \bra{\psi}U\ket{\phi} \neq 0 \}.
    \end{equation}
\end{definition}

The distance between subspaces of the Hilbert space is simply defined as the minimum distance between states in different subspaces. It will be convenient to use the following fact.
\begin{remark}
    Let $A, B \subset \bH_n$ be two subspaces with corresponding projectors $\Pi_A, \Pi_B$, and define
    \begin{equation}
        S_1 = \{z \in \{0,1\}^n : \bra{z}\Pi_A\ket{z} \neq 0\}, \quad S_2 = \{z \in \{0,1\}^n : \bra{z}\Pi_B\ket{z} \neq 0\}.
    \end{equation}
    The quantum Hamming distance $\dist_H(A, B)$ then satisfies
    \begin{equation}
        \dist_H(A, B) \geq \min_{z_1 \in S_1, z_2 \in S_2} |z_1 - z_2|_H.
    \end{equation}
\end{remark}
This will be useful for Hamiltonians formed by a classical Hamiltonian in a transverse field, i.e., of the form
\begin{equation}
    H = H_0 + h\sum_{i=1}^n X_i
\end{equation}
for $H_0$ diagonal in the $Z$ basis. Due to the presence of $H_0$, these Hamiltonians are naturally analyzable in the computational basis when $h$ is sufficiently small. For example, if two projectors $\Pi_A, \Pi_C$ are supported only on bitstrings far apart in Hamming distance --- as is often the case for regions of $H_0$ separated by a free energy barrier --- then the application of a local Lindbladian operator will necessarily traverse the bottleneck region between $A$ and $C$. Local Lindbladians arise naturally in the context of Gibbs sampling, both due to the quasi-locality of time-evolved operators used in Lindbladian constructions for geometrically local Hamiltonians (e.g., see~\cite{chen2023efficient,gilyen2024quantum} and  \Cref{sec:localtfim}), and due to the practical feasibility of implementing local operations compared to global unitaries. Formally, we define a $k$-local quantum channel as follows.

\begin{definition}[Local maps]
\label{def:local}
A channel has $\range_Z(\cM)\le k$ if its Kraus operators are linear combinations of Pauli strings of weight at most $k$.
For a Hermiticity-preserving map $\cM$ the same definition applies if its signed decomposition $\cM(\rho) = \sum_i s_i K_i \rho K_i^\dagger$ for $s_i \in \{\pm 1\}$ has operators $K_i$ which are combinations of Pauli strings of weight at most $k$. 
\end{definition}
When geometric locality is required, we impose the stronger condition that each Kraus operator is supported in a specified neighborhood. If each jump $L^a$ of a generator has jump degree at most one and the coherent operator has degree at most two, the full generator has degree at most two. 

\subsection{Strong clustering}
\label{sec:strong_cluster}
A major technical hurdle is to give a clear geometric picture of clustering for non-commuting Hamiltonians, where the eigenbasis and cluster basis do not necessarily coincide. In this section, we will introduce \emph{strong clustering} and understand its relation to the classical notion of \emph{weak clustering}. In the following section, we will see that strong clustering implies slow mixing.

We begin by recalling the definitions of weak and strong clustering from the introduction. Informally, weak clustering is when the support of the Gibbs state is clustered; strong clustering is when the Gibbs state is additionally block diagonal in the cluster basis.

\begin{definition}[Weak clustering of Gibbs states]
    Define orthogonal projectors $\Pi_1, \dots, \Pi_M$ on a $2^n$-dimensional Hilbert space such that $d(\Pi_i, \Pi_j) \geq d_0$ for all $i \neq j$ under some measure of distance $d$. Define $\Pi_\perp := I - \sum_{i=1}^M \Pi_i$. Let $\rho_\beta \propto e^{-\beta H}$ of $n$-qubit Hamiltonian $H$ be a Gibbs state. If there exists constants $\gamma_\perp, \gamma_0 > 0$ such that
    \begin{equation}
        \Tr[\Pi_\perp \rho_\beta] \leq \exp[-n^{\gamma_\perp}] \quad \text{and} \quad \Tr[\Pi_i \rho_\beta] \geq 1/n^{\gamma_0} \, \forall \, i \in [M],
    \end{equation}
    then we refer to $\rho_\beta$ as \emph{weakly clustered} under distance measure $d$. We refer to each $\Pi_i$ as a \emph{cluster} and to $\Pi_\perp$ as the \emph{bottleneck}.
\end{definition}

\begin{definition}[Strong clustering of Gibbs state]
    Let $\rho_\beta$ be a Gibbs state that is weakly clustered under distance measure $d$ with clusters $\Pi_1, \dots, \Pi_M$ and bottleneck $\Pi_\perp$. If the Gibbs state further satisfies
    \begin{equation}
        \frac{1}{2}\norm{\rho_\beta - \tilde \rho_\beta}_1 \leq \exp[-n^\alpha],
    \end{equation}
    for constant $\alpha > 0$ and Hermitian matrix
    \begin{equation}
        \tilde \rho_\beta := \Pi_\perp \rho_\beta \Pi_\perp + \sum_{i=1}^M \Pi_i \rho_\beta \Pi_i,
    \end{equation}
    then we refer to the Gibbs state as \emph{strongly clustered} under distance measure $d$.
\end{definition}

If the clusters are in the Hamiltonian eigenbasis, then strong clustering follows trivially from weak clustering. For this reason, for non-commuting Hamiltonians where the eigenbasis is unknown, we define clusters according to quantum Hamming distance $d_H$. Below, we construct a trivial example of a Hamiltonian that exhibits weak clustering but not strong clustering when measured under $d_H$.

\begin{example}\label{ex:weak}
    Define projectors
    \begin{equation}
        \Pi_1 = \ketbra{0^n}, \quad \Pi_2 = \ketbra{1^n}, \quad \Pi_\perp = I - \Pi_1 - \Pi_2
    \end{equation}
    and Hamiltonian
    \begin{equation}
        H = -(\Pi_1 + \Pi_2) - \delta\lr{\ketbra{0^n}{1^n} + \ketbra{1^n}{0^n}} + n\Pi_\perp
    \end{equation}
    for constant $\delta > 0$.
    For all $\beta > \log 2$, the Gibbs state satisfies weak clustering, i.e.,
    \begin{equation}
        d_H(\Pi_1, \Pi_2) = n, \qquad \Tr[\Pi_\perp \rho_\beta] = \exp[-\Omega(n)], \qquad \Tr[\Pi_1 \rho_\beta] = \Tr[\Pi_2 \rho_\beta] = \Omega(1).
    \end{equation}
     Moreover, the Gibbs state does not satisfy strong clustering:
     \begin{equation}
         \frac{1}{2}\norm{\rho_\beta - \tilde \rho_\beta}_1 = \Omega(1).
     \end{equation}
\end{example}
\begin{proof}
    The Hamiltonian acts block-diagonally with respect to subspaces $\Pi_{||} := \Pi_1 + \Pi_2$ and $\Pi_\perp$. In the bottleneck subspace, the Hamiltonian is $H_\perp = nI_{2^n-2}$ and thus satisfies $\Tr[e^{-\beta H_\perp}] = (2^n-2)e^{-\beta n}$. We thus have $\Tr[\Pi_\perp \rho_\beta] \leq \exp[-(\beta - \log 2)n]$, i.e., weak clustering for $\beta > \log 2$. $\Tr[\Pi_1 \rho_\beta] = \Tr[\Pi_2 \rho_\beta] = \Omega(1)$ follows by symmetry.
    
    In the 2-dimensional cluster subspace, $H_{||} = -I - \delta X$ for Pauli matrix $X$. Hence, $e^{-\beta H_{||}}/Z_{||} = \frac{1}{2}\lr{I + \tanh(\beta \delta)X}$. This gives
    \begin{equation}
        \frac{1}{2}\norm{\rho_\beta - \tilde \rho_\beta}_1 = \frac{\Tr[(\Pi_1+\Pi_2)\rho_\beta]}{2} \, |\tanh(\beta\delta)| = \Omega(1),
    \end{equation}
    violating strong clustering.
\end{proof}

Much of the technical work in this paper is dedicated to proving strong clustering. This generally requires two steps: (1) showing weak clustering, and (2) showing the additional block-diagonal structure of strong clustering. The remainder of this section is dedicated to developing tools for step (2). We will use these tools later to show strong clustering of systems such as lattices of dimension $D \geq 3$. These proofs will use the following result, which hinges upon Lemma~\ref{lem:transition} that we will discuss next.
\begin{proposition}[Strong clustering]
\label{prop:strong}
    Given Hamiltonian $H$ with $\Tr[H]=0$, and $M\ge 2$ orthogonal subspaces $\Pi_1, \dots, \Pi_M$ and $\Pi_\perp := I - \sum_{i=1}^M \Pi_i$, then for all integers $K \geq 1$,
    \begin{equation}
        \frac{1}{2}\norm{\rho_\beta - \tilde \rho_\beta}_1 \leq \sqrt{\Tr[\Pi_\perp\rho_\beta ]} + 2^{K-1} \sum_{i=1}^M \lr{\norm{\Pi_i e^{-\beta H/2^K} (I-\Pi_\perp-\Pi_i)} + \lr{\Tr[\Pi_\perp \rho_\beta]}^{2/2^K}},
    \end{equation}
    where
    \begin{equation}
        \tilde \rho_\beta := \Pi_\perp \rho_\beta \Pi_\perp + \sum_{i=1}^M \Pi_i \rho_\beta \Pi_i.
    \end{equation}
\end{proposition}
\begin{proof}
    By the triangle inequality, we have
    \begin{equation}
    \label{eq:stmain}
        \begin{split}
            \left\| \widetilde{\rho}_\beta - \rho_\beta \right\|_1 
            &= \left\|  \sum_{\substack{i,j \in \{\perp, 1, \dots, M\} \\ j \neq i} }  \Pi_i \rho_\beta \Pi_j  \right\|_1 \\
            &\leq 
            2\left\|   \Pi_\perp \rho_\beta \sum_{\substack{j \in \{1, \dots, M\} } } \Pi_j  \right\|_1 +  \sum_{\substack{j \in \{1, \dots, M\} } } 
            \left\|   \sum_{\substack{i \in \{1, \dots, M\} \\ i \neq j} } \Pi_i \rho_\beta \Pi_j  \right\|_1.
        \end{split}
    \end{equation}
    To bound the terms in \eqref{eq:stmain} involving $\Pi_\perp$, we use H\"older's inequality to obtain
    \begin{equation}
        \left\|\Pi_\perp \rho_\beta\sum_{\substack{j \in \{1, \dots, M\} } } \Pi_j\right\|_1 
        \leq 
        \left\|\sum_{\substack{j \in \{1, \dots, M\} } }\Pi_j \right\|_\infty \|\Pi_\perp \rho_\beta \|_1 \leq \norm{\Pi_\perp\rho_\beta^{1/2}}_2 \norm{\rho_\beta^{1/2}}_2 = \sqrt{\Tr[\Pi_\perp\rho_\beta ]}.
    \end{equation}
    We collect the above inequalities and apply them to \eqref{eq:stmain}. For each $i$, group the other clusters by taking $\Pi_C=I-\Pi_\perp-\Pi_i$. Applying the triangle inequality gives
    \begin{equation}
        \frac12\norm{\rho_\beta-\tilde\rho_\beta}_1
        \le \sqrt{\Tr[\Pi_\perp\rho_\beta]}
        +\frac12\sum_{i=1}^M\norm{\Pi_i\rho_\beta(I-\Pi_\perp-\Pi_i)}_1.
    \end{equation}
    Lemma~\ref{lem:transition} proven below gives the necessary bound on $\norm{\Pi_i\rho_\beta(I-\Pi_\perp-\Pi_i)}_1$ to conclude.
\end{proof}

\begin{remark}
    Proposition~\ref{prop:strong} has no constraint on the distance measure between clusters. In the following subsection, we will use it with the quantum Hamming distance $d_H$, but it may be useful more generally under other distance measures.
\end{remark}

The main bound used to show strong clustering in Proposition~\ref{prop:strong} is given by Lemma~\ref{lem:transition}, which controls the size of $\norm{\Pi_i \rho_\beta \Pi_j}_1$ across clusters $i \neq j$ in terms of the quantity $\norm{\Pi_i e^{-\beta H/2^K} (I-\Pi_\perp-\Pi_i)}$ for some integer $K \geq 1$. To gain intuition for this bound, one can consider Trotterizing the imaginary time evolution of $\rho_\beta$ into $2^K$ steps. For a local Hamiltonian, the exponential of $H$ can only move a small quantum Hamming distance away from $\Pi_i$. If the clusters are far apart, then this will cause the imaginary time evolution to fall into the bottleneck $\Pi_\perp$, where it is suppressed by the Gibbs measure; this gives rise to a quantity $\Tr[\Pi_\perp \rho_\beta]$. However, the imaginary time evolution may also skip past the bottleneck and instead arrive directly at a cluster (possibly $\Pi_j$ or possibly some other cluster $\Pi_k$); in this case, we receive the quantity $\norm{\Pi_i \rho_\beta (I-\Pi_i-\Pi_\perp)}_\infty$.

Formally, we do not bound Trotter error due to the difficulty of controlling it for generic local Hamiltonians. (Indeed, we do not even assume any locality of the Hamiltonian in the bound below, contrary to the intuition provided by the Trotter argument.) We instead chiefly rely on the repeated application of H\"older's inequality on a telescoping decomposition of the Gibbs state. We also use the tensor power trick and Jensen's inequality to prove two facts used on the proof (which are themselves proven after the lemma). The value of $K$ in the lemma below will be set large enough so that the two terms in the sum in \Cref{eq:transition_cluster_K} decay with $n$. We give various examples in the next section (\Cref{sec:examples_strong_clustering}).

\begin{lemma}[Transition across clusters]
\label{lem:transition}
    Given Hamiltonian $H$ with $\Tr[H]=0$, and orthogonal subspaces $\Pi_i,\Pi_j,\Pi_\perp$, then for all integers $K \geq 1$,
    \begin{equation}\label{eq:transition_cluster_K}
        \norm{\Pi_i \rho_\beta \Pi_j}_1 \leq 2^K\lr{\norm{\Pi_i e^{-\beta H/2^K} \lr{I-\Pi_\perp-\Pi_i}}_\infty + \lr{\Tr[\Pi_\perp \rho_\beta]}^{2/2^K}}.
    \end{equation}
\end{lemma}
\begin{proof}
    Let $K$ be a positive integer and define quantities
    \begin{equation}
        Z := \Tr[e^{-\beta H}], \qquad T_r := \lr{e^{-\beta H/2^K}}^r \Pi_\perp e^{-\beta H/2^K} \left[\lr{I-\Pi_\perp}e^{-\beta H/2^K}\right]^{2^K-r-1}.
    \end{equation}
    Inserting the identity, we have for any integer $0 \leq q \leq 2^K-1$ that
    \begin{equation}
    \begin{split}
        \lr{e^{-\beta H/2^K}}^{2^K-q}\left[(I-\Pi_\perp)e^{-\beta H/2^K}\right]^q &= \lr{e^{-\beta H/2^K}}^{2^K-q-1}(I-\Pi_\perp+\Pi_\perp)e^{-\beta H/2^K}\left[(I-\Pi_\perp)e^{-\beta H/2^K}\right]^q\\
        &= \lr{e^{-\beta H/2^K}}^{2^K-q-1}\left[(I-\Pi_\perp)e^{-\beta H/2^K}\right]^{q+1} + T_{2^K-q-1}.
    \end{split}
    \end{equation}
    Consequently, we can apply the decomposition
    \begin{equation}
        e^{-\beta H} = e^{-\beta H/2^K} \left[(I-\Pi_\perp)e^{-\beta H/2^K}\right]^{2^K-1} + \sum_{r=1}^{2^K-1} T_r
    \end{equation}
    to obtain via the triangle inequality
    \begin{equation}
        \norm{\Pi_i \rho_\beta \Pi_j}_1 \leq \frac{1}{Z}\lr{\norm{\Pi_i e^{-\beta H/2^K} \left[(I-\Pi_\perp)e^{-\beta H/2^K}\right]^{2^K-1} \Pi_j}_1 + \sum_{r=1}^{2^K-1}\norm{\Pi_i T_r \Pi_j}_1}.
    \end{equation}
    We repeat a similar decomposition.
    Letting $\Pi_C = I - \Pi_i - \Pi_\perp$, we have
    \begin{equation}
        \lr{(\Pi_i+\Pi_C)e^{-\beta H/2^K}}^{2^K-1} = \lr{\Pi_i e^{-\beta H/2^K}}^{2^K-1} + \sum_{r=1}^{2^K-1} U_r
    \end{equation}
    for
    \begin{equation}
        U_r := \lr{\Pi_i e^{-\beta H/2^K}}^{r-1} \Pi_C e^{-\beta H/2^K} \lr{(\Pi_i+\Pi_C)e^{-\beta H/2^K}}^{2^K-r-1}.
    \end{equation}
    By the triangle inequality, this implies
    \begin{equation}
        \norm{\Pi_i e^{-\beta H/2^K} \left[(I-\Pi_\perp)e^{-\beta H/2^K}\right]^{2^K-1} \Pi_j}_1 \leq \norm{\lr{\Pi_i e^{-\beta H/2^K}}^{2^K} \Pi_j}_1 + \sum_{r=1}^{2^K-1} \norm{\Pi_i e^{-\beta H/2^K} U_r \Pi_j}_1.
    \end{equation}
    It now remains to bound the terms in
    \begin{equation}
        \norm{\Pi_i \rho_\beta \Pi_j}_1 \leq \frac{1}{Z}\lr{\norm{\lr{\Pi_i e^{-\beta H/2^K}}^{2^K} \Pi_j}_1 + \sum_{r=1}^{2^K-1} \norm{\Pi_i e^{-\beta H/2^K} U_r \Pi_j}_1 + \sum_{r=1}^{2^K-1}\norm{\Pi_i T_r \Pi_j}_1}.
    \end{equation}
    Applying H\"older's inequality with $\norm{\Pi}_\infty = 1$ for all projectors, the first term is
    \begin{equation}
    \begin{split}
        \norm{\lr{\Pi_i e^{-\beta H/2^K}}^{2^K} \Pi_j}_1 &\leq \norm{\Pi_i e^{-\beta H/2^K} \Pi_j}_{2^K}\norm{e^{-\beta H/2^K}}_{2^K}^{2^K-1}\\
        &\leq \norm{\Pi_i e^{-\beta H/2^K} \Pi_C}_{2^K}\norm{e^{-\beta H/2^K}}_{2^K}^{2^K-1},
    \end{split}
    \end{equation}
    and the second term is
    \begin{equation}
    \begin{split}
        \norm{\Pi_i e^{-\beta H/2^K} U_r \Pi_j}_1 &= \norm{\lr{\Pi_i e^{-\beta H/2^K}}^r \Pi_C e^{-\beta H/2^K} \lr{(\Pi_i+\Pi_C)e^{-\beta H/2^K}}^{2^K-r-1}\Pi_j}_1\\
        &\leq \norm{\Pi_i e^{-\beta H/2^K} \Pi_C}_{2^K} \norm{e^{-\beta H/2^K}}_{2^K}^{2^K-1}.
    \end{split}
    \end{equation}
    We apply the definition of the Schatten norm, the inequality $\norm{M}_p \leq \dim(M)^{1/p} \norm{M}_\infty$, and Jensen's inequality $Z \geq 2^n$ (Fact~\ref{fact:jensen}) to obtain
    \begin{equation}
    \begin{split}
        \norm{\Pi_i e^{-\beta H/2^K} \Pi_C}_{2^K} \norm{e^{-\beta H/2^K}}_{2^K}^{2^K-1} &= \norm{\Pi_i e^{-\beta H/2^K} \Pi_C}_{2^K} Z^{1-1/2^K}\\
        &\leq 2^{n/2^K} \norm{\Pi_i e^{-\beta H/2^K} \Pi_C}_\infty Z^{1-1/2^K}\\
        &\leq Z \norm{\Pi_i e^{-\beta H/2^K} \Pi_C}_\infty.
    \end{split}
    \end{equation}
    Hence, the first two terms satisfy
    \begin{equation}
        \frac{1}{Z}\lr{\norm{\lr{\Pi_i e^{-\beta H/2^K}}^{2^K} \Pi_j}_1 + \sum_{r=1}^{2^K-1} \norm{\Pi_i e^{-\beta H/2^K} U_r \Pi_j}_1} \leq 2^K \norm{\Pi_i e^{-\beta H/2^K} \Pi_C}_\infty.
    \end{equation}
    We now turn to the norms involving $T_r$ following a similar argument. H\"older's inequality gives
    \begin{equation}
        \norm{\Pi_i T_r \Pi_j}_1 \leq \norm{e^{-\beta H/2^K}}_{2^K}^{2^K-2} \norm{e^{-\beta H/2^K}\Pi_\perp e^{-\beta H/2^K}}_{2^{K-1}}.
    \end{equation}
    By Fact~\ref{fact:tensor} and rearranging via cyclicity of trace, we have
    \begin{equation}
        \norm{\Pi_i T_r \Pi_j}_1 \leq \norm{e^{-\beta H/2^K}}_{2^K}^{2^K-2} \lr{\Tr[\Pi_\perp e^{-\beta H}]}^{2/2^K} = Z \lr{\Tr[\Pi_\perp \rho_\beta]}^{2/2^K}.
    \end{equation}
    Combining everything, we have
    \begin{equation}
    \begin{split}
        \norm{\Pi_i \rho_\beta \Pi_j}_1 &\leq 2^K \norm{\Pi_i e^{-\beta H/2^K} \Pi_C}_\infty + (2^K-1)\lr{\Tr[\Pi_\perp \rho_\beta]}^{2/2^K},
    \end{split}
    \end{equation}
    where we again used $Z \geq 2^n$ by Jensen's inequality (Fact~\ref{fact:jensen}). We weaken the inequality slightly in the lemma statement by adding $\lr{\Tr[\Pi_\perp \rho_\beta]}^{2/2^K}$.
\end{proof}

\begin{fact}\label{fact:tensor}
    Given Hermitian matrices $A,B$ and any integer $p=2^k$ which is a power of $2$ (i.e. $k>0$ is a positive integer), 
    \begin{equation}
        \Tr[(AB)^p] \leq \Tr[A^p B^p].
    \end{equation}
\end{fact}
\begin{proof}
    This is a consequence of the so-called tensor power trick. Since $p$ is a positive integer power, by Hölder's inequality
    \begin{equation}
        \Tr[(AB)^p]  \leq \|AB\|_{p}^{p},
    \end{equation}
    where in the above $\|\cdot\|_p$ denotes the Schatten $p$-norm. We have by applying the definition of the Schatten norm and using cyclicity of trace:
    \begin{equation}
        \Tr[(AB)^p] \leq \|AB\|_{p}^{p} = \Tr[(AB (AB)^\dagger)^{p/2}] = \Tr[(A^2B^2)^{p/2}].
    \end{equation}
    Iterating this inequality we arrive at the stated result.
\end{proof}

\begin{fact} \label{fact:jensen}
Let $H$ be a Hamiltonian on an $N$-dimensional Hilbert space satisfying $\Tr(H)=0$. Then, for any $\beta > 0$,
\begin{equation}
\Tr\left(e^{\beta H}\right) \geq N.
\end{equation}
\end{fact}
\begin{proof}
    Let $\lambda_1, \lambda_2, \dots, \lambda_N$ be the eigenvalues of $H$.
    Consider the convex function $f(x)=e^{\beta x}$ for $\beta > 0$. By Jensen's inequality, using the uniform probability distribution $p_i = \frac{1}{N}$ for each eigenvalue, we obtain:
    \begin{equation}
    \frac{1}{N}\sum_{i=1}^{N} e^{\beta \lambda_i} \geq e^{\beta \left(\frac{1}{N}\sum_{i=1}^{N} \lambda_i\right)} = e^{\beta \cdot 0} = 1.
    \end{equation}
    Multiplying both sides of the inequality by $N$ gives:
    \begin{equation}
    \sum_{i=1}^{N} e^{\beta \lambda_i} \geq N.
    \end{equation}
    Since the trace of $e^{\beta H}$ is the sum of $e^{\beta \lambda_i}$ over all eigenvalues, it follows that:
    \begin{equation}
    \Tr\left(e^{\beta H}\right) \geq N.
    \end{equation}
\end{proof}

\subsection{Examples of strong clustering} \label{sec:examples_strong_clustering}
We now show how to apply Proposition~\ref{prop:strong} to obtain sufficient conditions for weak clustering to imply strong clustering. We will provide a general bound that applies to several common Hamiltonians, but which excludes 2D Hamiltonians.
Throughout this section, we will assume Hamiltonians are centered so $\Tr[H]=0$.

\begin{theorem}\label{thm:strong}
    Assume Hamiltonian $H$ is $\ell$-local and its Gibbs state $\rho_\beta$ is weakly clustered with Hamming distance $d_0$ between clusters satisfying $d_0/\ell\ge n^\eta$ for some constant $\eta>0$ and a bottleneck
    \begin{equation}
        \Tr[\Pi_\perp \rho_\beta] \leq \exp[-\Omega(n^\gamma)].
    \end{equation}
    If there exists positive integer $K \geq 1$ and such that for some constant $\delta > 0$,
    \begin{align}
        \frac{10\ell \beta \norm{H}}{d_0} < 2^K < \min\left\{n^{\gamma-\delta}, \exp\left[\frac{d_0}{2\ell}\right]\right\},
    \end{align}
    then $\rho_\beta$ is strongly clustered.
\end{theorem}
\begin{proof}
    We apply Proposition~\ref{prop:strong}. Take a cluster $\Pi_i$ and the remaining clusters $\Pi_C := I-\Pi_\perp-\Pi_i$, and set $m:=\lceil d_0/\ell\rceil$.
    Since $\Pi_i$ and $\Pi_C$ are at least $d_0$ apart in Hamming distance and $H$ is $\ell$-local, we have $\Pi_i H^a\Pi_C=0$ for every integer $a<m$.
    By the triangle inequality and submultiplicativity,
    \begin{equation}
        \begin{split}
            \norm{\Pi_i e^{-\beta H/2^K}\Pi_C}_\infty
            &=
            \norm{
                \sum_{a\geq m}
                \frac{(-\beta/2^K)^a}{a!}
                \Pi_i H^a\Pi_C
            }_\infty\\
            &\leq
            \sum_{a\geq m}
            \frac{(\beta\norm{H}/2^K)^a}{a!}.
        \end{split}
    \end{equation}
    We apply the standard bound
    $\sum_{a\geq m}x^a/a!\leq e^x x^m/m!$ for $x\geq0$
    and the Stirling inequality $m!\geq(m/e)^m$ to obtain
    \begin{equation}
        \begin{split}
            \norm{\Pi_i e^{-\beta H/2^K}\Pi_C}_\infty
            &\leq
            \frac{(\beta\norm{H}/2^K)^m}{m!}
            \exp\left[\frac{\beta\norm{H}}{2^K}\right]\\
            &\leq
            \exp\left[
                \frac{\beta\norm{H}}{2^K}
                -m\log\lr{
                    \frac{2^K m}{e\beta\norm{H}}
                }
            \right].
        \end{split}
    \end{equation}
    Choose a constant $c_\perp>0$ such that $\Tr[\Pi_\perp\rho_\beta]\leq
    \exp[-c_\perp n^\gamma]$ for all sufficiently large $n$.
    Applying proposition~\ref{prop:strong} gives
    \begin{equation}
        \begin{split}
            \frac12\norm{\rho_\beta-\tilde\rho_\beta}_1
            &\leq
            \exp\left[-\frac{c_\perp n^\gamma}{2}\right]\\
            &\quad+
            M2^K\left(
                \exp\left[
                    \frac{\beta\norm{H}}{2^K}
                    -m\log\lr{
                        \frac{2^K m}{e\beta\norm{H}}
                    }
                \right]
                +
                \exp\left[
                    -\frac{2c_\perp n^\gamma}{2^K}
                \right]
            \right).
        \end{split}
    \end{equation}
    It remains to show that the right hand side is exponentially small. The lower bound on $2^K$ gives
    \begin{equation}
        \frac{\beta\norm{H}}{2^K}
        <\frac{d_0}{10\ell}
        \leq\frac{m}{10},
    \end{equation}
    and therefore
    \begin{equation}
        \frac{\beta\norm{H}}{2^K}
        -m\log\lr{\frac{2^K m}{e\beta\norm{H}}}
        \leq
        -\left(\log10-\frac{11}{10}\right)m.
    \end{equation}
    Also, $2^K<n^{\gamma-\delta}$ implies
    \begin{equation}
        \exp\left[-\frac{2c_\perp n^\gamma}{2^K}\right]
        \leq \exp[-2c_\perp n^\delta].
    \end{equation}
    By weak clustering, there is a constant $\gamma_0>0$
    such that $\Tr[\Pi_i\rho_\beta]\geq n^{-\gamma_0}$
    for every $i$. Summing these weights gives
    $M\leq n^{\gamma_0}$.
    Consequently,
    \begin{equation}
        \begin{split}
            \frac12\norm{\rho_\beta-\tilde\rho_\beta}_1
            &\leq
            \exp\left[-\frac{c_\perp n^\gamma}{2}\right]\\
            &\quad+
            n^{\gamma_0+\gamma-\delta}\left(
                \exp\left[
                    -\left(\log10-\frac{11}{10}\right)n^\eta
                \right]
                +\exp[-2c_\perp n^\delta]
            \right)\\
            &=
            \exp[-\Omega(n^{\min\{\eta,\delta\}})].
        \end{split}
    \end{equation}
    For $\alpha < \min\{\eta,\delta\}$, this is $\exp[-\Omega(n^\alpha)]$ for all sufficiently large $n$, proving strong clustering.
\end{proof}

\begin{corollary}[Strong clustering for linear free energy barriers]
    Assume
    \begin{equation}
        d_0 = dn, \qquad \Tr[\Pi_\perp \rho_\beta] \leq \exp[-cn], \qquad \norm{H} = hn.
    \end{equation}
    Then weak clustering implies strong clustering by Theorem~\ref{thm:strong}, as seen by choosing $K = \max\left\{1,\left\lceil \log_2 \frac{10\ell\beta h}{d} \right\rceil\right\} = O(1)$.
\end{corollary}

\begin{corollary}[Strong clustering for lattices]
    Assume
    \begin{equation}
        d_0 = d n^{1-1/D}, \qquad \Tr[\Pi_\perp \rho_\beta] \leq \exp[-c n^{1-1/D}], \qquad \norm{H} = hn.
    \end{equation}
    Then for $D > 2$, weak clustering implies strong clustering by Theorem~\ref{thm:strong}, as seen by choosing $K = \max\left\{1,\left\lceil \log_2 \frac{10\ell\beta h n^{1/D}}{d} \right\rceil\right\} = \Theta(\log n^{1/D})$.
\end{corollary}

The first corollary addresses sparse Hamiltonians such as LDPC codes with a non-commuting $O(1)$ perturbation. The second corollary addresses lattice Hamiltonians in dimension $D \geq 3$. Classically, proofs of slow mixing usually rely on fault line arguments that result in a sublinear free energy barrier given by $n^{1-1/D}$. However, Theorem~\ref{thm:strong} does not address $D=2$, since the inequality $n^{1/D} < 2^K < n^{1-1/D-\delta}$ cannot be satisfied for $D = 2$. Showing that weak clustering generically implies strong clustering for $D=2$ remains an open question.

\section{Slow mixing}

\label{sec:slow}
\begin{figure}
    \centering
    \includegraphics[width=0.6\textwidth]{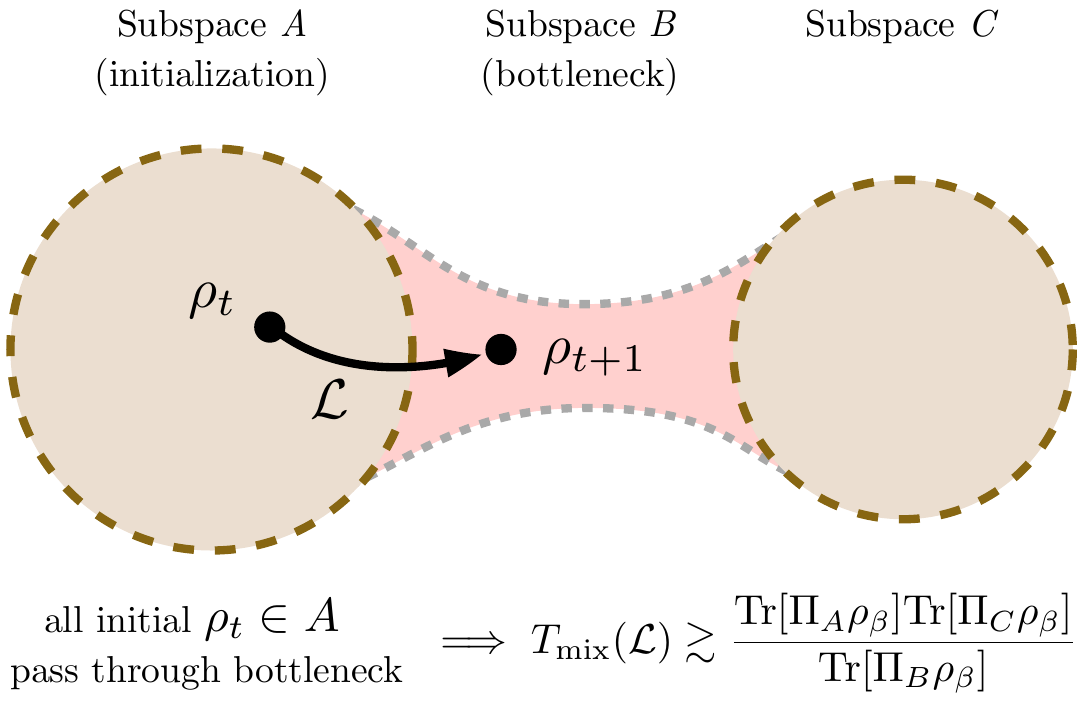}
    \caption{Proofs guaranteeing slow mixing of a given Gibbs sampling algorithm require two ingredients. First, one must find a region $A$ where the state is initialized and a bottleneck $B$ which the sampler must pass through to sample from a region $C$. Second, one must show that the bottleneck region $B$ has vanishing probability. Namely, if the probability of bottleneck region $B$ is exponentially smaller than that of $A$ and $C$, then the mixing time is also exponentially long.}
    \label{fig:bottleneck_illustrative}
\end{figure}

Here, we introduce a generic argument used to show lower bounds on mixing times, schematically illustrated in \Cref{fig:bottleneck_illustrative}. Roughly speaking, the argument consists of two steps. First, we identify a \emph{bottleneck} region $B$ between two regions ($A$ and $C$) that each have large support in the thermal state. The bottleneck must be chosen such that $\Tr(\Pi_B \rho_\beta)$ is exponentially (or superpolynomially) small. Second, we define a notion of distance such that a Gibbs sampler initialized at a state in region $A$ must pass through $B$ to arrive at region $C$. In the classical case, this distance is often taken to be Hamming distance: a classical sampler typically flips a single bit (or constant number of bits) in each step. The classical proof of slow mixing relies on initializing in a state confined to cluster $A$ then showing that the state is suppressed by the Gibbs measure once the sampler enters the bottleneck (see, e.g., \cite{levin2017markov} for a pedagogical introduction).

This classical argument relies on a fairly geometric interpretation of the sampling dynamics. In the quantum setting, we shall see that this translates to strong clustering, where the block-diagonal form of the Gibbs state ensures that the state is well-controlled as being in $\Pi_A, \Pi_B$ or $\Pi_C$ throughout the argument.

To require that the sampler passes through region $\Pi_B$, one can measure the range of the sampler and compare it to the distance between clusters. More generally, for clusters $\Pi_A$ and $\Pi_C$, it suffices to require that the Gibbs sampler $\cT$ satisfies $\cT^\dagger(\Pi_C) \preccurlyeq \Pi_B + \Pi_C$. We will state our main bottleneck lemma in terms of this more general condition; first, we formally relate it to our preestablished notions of distance and range.

\begin{lemma}[Necessary criteria for sampler to traverse bottleneck]
    Let $A, C \subseteq \mathbb{H}_n$ denote orthogonal subspaces with projectors $\Pi_A, \Pi_C \in \mathcal{O}_{\bH_n}$. Let $\cT$ be a discrete time Gibbs sampler where
    \begin{equation}
        \begin{split}
            \dist_{\Tset}(A, C) > \trange(\cT) \quad \text{or} \quad \dist_H(A, C) > \range_Z(\cT).
        \end{split}
    \end{equation}
    Set $\Pi_B = I - \Pi_A - \Pi_C$. Then, it holds that $\cT^\dagger(\Pi_C) \preccurlyeq \Pi_B + \Pi_C $.
\end{lemma}
\begin{proof}
    We have that for any state $\rho$:
    \begin{equation}
        \begin{split}
            \Tr(\cT^\dagger(\Pi_C)\rho) & = \Tr(\Pi_C \cT(\rho)) \\
            &= \Tr( \Pi_C \cT((\Pi_C + \Pi_B + \Pi_A )\rho(\Pi_C + \Pi_B + \Pi_A ))) \\
            &= \Tr( \Pi_C \cT((\Pi_C + \Pi_B  )\rho(\Pi_C + \Pi_B  ))) \\
            &\leq \Tr( \cT((\Pi_C + \Pi_B  )\rho(\Pi_C + \Pi_B  ))) \\
            &= \Tr( (\Pi_C + \Pi_B  )\rho).
        \end{split}
    \end{equation}
    In the second line, we use the fact that $\Pi_C + \Pi_B + \Pi_A = I$. The third line follows from either property, $\dist_{\Tset}(A, C) > \trange(\cT)$ or $\dist_H(A, C) > \range_Z(\cT)$. In the fourth and fifth lines, we use the completely positive and trace preserving properties of the map $\cT$.
\end{proof}
Note that for a Hermiticity-preserving linear operator $\cM$ with range smaller than the distance $\dist_{\Tset}(A, C)$, we have the identity
\begin{equation}
    \cM^\dagger(\Pi_C) = (\Pi_B + \Pi_C) \cM^\dagger(\Pi_C) (\Pi_B + \Pi_C).
\end{equation}

To produce a mixing time lower bound, we choose a strategic initial state such that the sampler is forced to traverse the bottleneck. Classically, one can simply start in region $A$, and the sampler will traverse region $B$ to arrive at $C$. We will choose an initial state
\begin{equation}
    \sigma_0 := \frac{e^{-\beta H/2}\Pi_A e^{-\beta H/2}}{\Tr[\Pi_A e^{-\beta H}]}.
\end{equation}
For the bottleneck $\Pi_B$ to dominate the mixing time lower bound, we require a decoupling property that ensures that $\sigma_0$ is not well-supported in region $\Pi_C$:
\begin{equation}
\label{eq:cacond}
    \Tr[\Pi_C \sigma_0] = \Tr[\Pi_C \frac{e^{-\beta H/2} \Pi_A e^{-\beta H/2}}{\Tr[\Pi_A e^{-\beta H}]}] = o(\Tr[\Pi_C\rho_\beta]).
\end{equation}
Whenever $[\Pi_A, H]=0$, this condition is trivially satisfied. For non-commuting Hamiltonians, we show that strong clustering implies \eqref{eq:cacond}. The proof uses a matrix inequality that we state below.

\begin{fact}[Powers-St{\o}rmer \cite{powers1970free}]\label{fact:ps}
    Any two PSD matrices $A, B$ satisfy $\norm{\sqrt A-\sqrt B}_2^2 \leq \norm{A-B}_1$.
\end{fact}

\begin{lemma}[Strong clustering implies decoupling]\label{lem:decouple}
    Assume $\rho_\beta$ is a strongly clustered Gibbs state. Then there exists orthogonal projectors $\Pi_A, \Pi_B, \Pi_C$ such that $\Pi_A + \Pi_B + \Pi_C = I$ and a constant $\alpha > 0$ such that
    \begin{equation}
        \frac{\Tr[\Pi_C e^{-\beta H/2} \Pi_A e^{-\beta H/2}]}{\Tr[e^{-\beta H/2} \Pi_A e^{-\beta H/2}]} \leq \exp[-n^\alpha+O(\log(n))].
    \end{equation}
\end{lemma}
\begin{proof}
    Let $\rho_\beta = e^{-\beta H}/Z$ and $\tilde \rho_\beta = \Pi_\perp \rho_\beta \Pi_\perp + \sum_{i=1}^M \Pi_i \rho_\beta \Pi_i$.
    We use the strong clustering property
    \begin{equation}
        \norm{\tilde \rho_\beta - \rho_\beta}_1 \leq 2\exp[-n^\alpha]
    \end{equation}
    and the observation that
    \begin{equation}
        \sqrt{\tilde \rho_\beta} = (\Pi_\perp  \rho_\beta \Pi_\perp)^{1/2} + \sum_{i=1}^M (\Pi_i \rho_\beta \Pi_i)^{1/2}
    \end{equation}
    due to its block-diagonal form. Let $\Pi_A = \sum_{i=1}^{M-1} \Pi_i$ and $\Pi_C = \Pi_M$, and define
    \begin{equation}
        X = \Pi_A e^{-\beta H/2} \Pi_C.
    \end{equation}
    We begin by showing an upper bound on
    \begin{equation}
        \Tr[\Pi_C e^{-\beta H/2} \Pi_A e^{-\beta H/2}] = \Tr[X^\dagger X] = \norm{X}_2^2.
    \end{equation}
    For convenience, define matrices $Y = \sqrt{Z \rho_\beta}$ and $\tilde Y = \sqrt{Z \tilde \rho_\beta}$.
    To bound $\norm{X}_2$, we note that since $\Pi_A \tilde Y \Pi_C = 0$ and using Hölder's inequality,
    \begin{equation}
        \norm{X}_2^2 = \norm{\Pi_A (Y - \tilde Y) \Pi_C}_2^2 \leq \norm{Y - \tilde Y}_2^2.
    \end{equation}
    The Powers-St{\o}rmer inequality (Fact~\ref{fact:ps}) states that
    \begin{equation}
        \norm{Y - \tilde Y}_2^2 \leq \norm{Y^2 - \tilde Y^2}_1 = Z\norm{\rho_\beta - \tilde \rho_\beta}_1 \leq 2Z \exp[-n^\alpha].
    \end{equation}
    By the assumption $\Tr[\Pi_A \rho_\beta] \geq 1/n^{\gamma_0}$, we conclude that
    \begin{equation}
        \frac{\Tr[\Pi_C e^{-\beta H/2} \Pi_A e^{-\beta H/2}]}{\Tr[e^{-\beta H/2} \Pi_A e^{-\beta H/2}]} \leq \frac{2Z \exp[-n^\alpha]}{Z}n^{\gamma_0} = \exp[-n^{\alpha}+O(\log(n))].
    \end{equation}
\end{proof}

We are now prepared to show slow mixing as a consequence of strong clustering.

\begin{proposition}[Mixing time from bottleneck (discrete Gibbs sampler)]\label{prop:mix_time_discrete}
    Denote $H$ as a Hamiltonian acting on $n$ qubits with Gibbs state $\rho_\beta$ at inverse temperature $\beta$. 
    Let $\cT$ denote a discrete time Gibbs sampler with fixed point $\cT(\rho_\beta)=\rho_\beta$ applied to initial state $\sigma_0$. 
    Let $A,B,C \subseteq \mathbb{H}_n$ be orthogonal linear subspaces of $\bH_n$ with corresponding projectors $\Pi_A,\Pi_B,\Pi_C \in \mathcal{O}_{\bH_n}$ such that $\Pi_A+\Pi_B+\Pi_C=I$ and $\Tr[\Pi_C\sigma_0]=o(\Tr[\Pi_C\rho_\beta])$, e.g., due to strong clustering with clusters $\Pi_A, \Pi_C$ and bottleneck $\Pi_B$.
    Assume furthermore that there exists a Hermiticity-preserving linear map $\widetilde{\cT}$ (not necessarily positive or trace preserving) which satisfies
    \begin{equation}
        \norm{\cT - \widetilde{\cT}}_\diamond \leq \epsilon, \quad \widetilde{\cT}^\dagger(\Pi_C) \preccurlyeq \Pi_B + \Pi_C.
    \end{equation}
    Then the mixing time of $\cT$ is bounded below as
    \begin{equation}
        \tmix(\cT) = \Omega\lr{\min\left[\Tr[\Pi_A \rho_\beta]\Tr[\Pi_C \rho_\beta]\Tr[\Pi_B \rho_\beta]^{-1}, \; \frac{ \Tr[\Pi_C \rho_\beta]}{\epsilon}\right]}
    \end{equation}
\end{proposition}
\begin{proof}
    Consider initial state
    \begin{equation}
    \sigma_0 := \frac{e^{-\beta H / 2} \Pi_A e^{-\beta H / 2}}{\Tr[\Pi_A e^{-\beta H }]}.
    \end{equation}
    Letting $\sigma_t = \cT^t(\sigma_0)$, we have
    \begin{equation}
    \begin{split}
        \Tr[\Pi_B  \sigma_{t-1}] &= \Tr[\Pi_B \cT^{t-1}(\sigma_0) ] \\
        &= \Tr[e^{-\beta H } \Pi_A]^{-1} \Tr[\Pi_B \cT^{t-1}(e^{-\beta H / 2} \Pi_A e^{-\beta H / 2}) ] \\
        &\leq \Tr[e^{-\beta H } \Pi_A]^{-1} \Tr[\Pi_B \cT^{t-1}(e^{-\beta H}) ]  \\
        &= \Tr[e^{-\beta H } \Pi_A]^{-1} \Tr[\Pi_B e^{-\beta H} ] \\
        &= \Tr[\Pi_A \rho_\beta]^{-1} \Tr[\Pi_B \rho_\beta],
    \end{split}
    \end{equation}
    where the inequality follows from $\cT$ being a CPTP linear map.
    Since $\widetilde \cT^\dagger(\Pi_C) \preccurlyeq \Pi_B + \Pi_C$, we have 
    \begin{equation} \label{eq:pre_iterate_pi_B}
    \begin{split}
        \Tr[\Pi_C \sigma_t] &= \Tr[\Pi_C \cT(\sigma_{t-1})] \\
        &= \Tr[\cT^\dagger(\Pi_C) \sigma_{t-1}] \\
        &\leq \Tr[\widetilde\cT^\dagger(\Pi_C) \sigma_{t-1}] + \epsilon \norm{\Pi_C}\\
        &\leq \Tr[\Pi_C  \sigma_{t-1}] + \Tr[\Pi_B  \sigma_{t-1}] + \epsilon.
    \end{split}
    \end{equation}
    To simplify what follows let $a=\Tr[\Pi_A\rho_\beta]$, $b=\Tr[\Pi_B\rho_\beta]$, $c=\Tr[\Pi_C\rho_\beta]$ and $q=\Tr[\Pi_C\sigma_0]$.
    Iterating \eqref{eq:pre_iterate_pi_B} gives
    \begin{equation}
        \Tr[\Pi_C\sigma_t]\le q+t(b/a+\epsilon).
    \end{equation}
    Since $\rho_\beta=a\sigma_0+(1-a)\sigma_{\mathrm{rest}}$ for a state $\sigma_{\mathrm{rest}}$, we have $\frac12\|\sigma_0-\rho_\beta\|_1\le b+c$.
    At a time satisfying the mixing time cutoff $\theta\le1/2$,
    \begin{equation}
     c-q-t(b/a+\epsilon)
     \le\frac12\|\sigma_t-\rho_\beta\|_1
     \le\theta(b+c).
    \end{equation}
    Consequently
    \begin{equation}
     t\ge\frac{(1-\theta)c-\theta b-q}{b/a+\epsilon}
     =\Omega\left(\min\left\{\frac{ac}{b},\frac c\epsilon\right\}\right),
    \end{equation}
    where the final step uses $b,q=o(c)$.
\end{proof}

\begin{lemma}[Mixing time from bottleneck (continuous time Lindbladian)]
\label{lem:slow}
    Denote $H$ as a Hamiltonian acting on $n$ qubits with Gibbs state $\rho_\beta$ at inverse temperature $\beta$. Let $A,B,C \subseteq \mathbb{H}_n$ be orthogonal linear subspaces of $\bH_n$ with corresponding projectors $\Pi_A,\Pi_B,\Pi_C \in \mathcal{O}_{\bH_n}$ such that $\Pi_A+\Pi_B+\Pi_C=I$ and $\Tr[\Pi_C\sigma_0]=o(\Tr[\Pi_C\rho_\beta])$, e.g., due to strong clustering of $\rho_\beta$ with clusters $\Pi_A, \Pi_C$ and bottleneck $\Pi_B$.
    Let $\cL$ be a Lindbladian Gibbs sampler with fixed point $\cL(\rho_\beta) = 0$
    and $\|\cL\|_\diamond= O(\poly{n})$. Assume there exists Hermiticity-preserving linear map $\widetilde\cL$ satisfying\footnote{As in \Cref{prop:mix_time_discrete}, the criteria $\dist_{\Tset}(A, C) > \trange(\widetilde\cL)$ or $\dist_H(A, C) > \range_Z(\widetilde\cL)$ suffice to obtain such a $\Pi_B$.}
    \begin{equation}
        \norm{\cL - \widetilde\cL}_\diamond \leq \epsilon, \quad \widetilde{\cL}^\dagger(\Pi_C) = (\Pi_C+\Pi_B)\widetilde{\cL}^\dagger(\Pi_C) (\Pi_C+\Pi_B).
    \end{equation}
    and assume there exist constants $c_1,\gamma,\gamma_0 > 0$ such that $\Tr[\Pi_C \rho_\beta]\ge n^{-\gamma_0}$ and
    \begin{equation}
        \exp\left[c_1 n^\gamma\right] \leq \min\left\{\Tr[\Pi_A \rho_\beta]\Tr[\Pi_C \rho_\beta]\Tr[\Pi_B \rho_\beta]^{-1}, \frac{\Tr[\Pi_C \rho_\beta]}{\epsilon}\right\}.
    \end{equation}
    Then $\cL$ satisfies
    \begin{equation}
        \tmix(\cL) = \exp\left[\Omega\lr{n^\gamma}\right].
    \end{equation}
\end{lemma}
\begin{proof}
    The proof in this case follows closely that of \Cref{prop:mix_time_discrete}.
    For $\delta>0$, set 
    \begin{equation}
        \cT_\delta=e^{\delta\cL},\qquad \widetilde\cT_\delta=\operatorname{Id}+\delta\widetilde\cL.
    \end{equation}
    The first map above is a channel and the second is its first-order approximation which is a Hermiticity-preserving map.  Lemma~\ref{lem:discrete_convert_to_continuous} gives
    \begin{equation}
         \|\cT_\delta-\widetilde\cT_\delta\|_\diamond
         \le\eta_\delta:=\delta\epsilon+\tfrac12\delta^2\|\cL\|_\diamond^2.
    \end{equation}

    Following \Cref{prop:mix_time_discrete}, choosing initial state $\sigma_0 := \tr(\Pi_A\rho_\beta)^{-1}\rho_\beta^{1/2} \Pi_A \rho_\beta^{1/2}$ and letting $\sigma_t = \cT_\delta^t(\sigma_0)$, we have
    \begin{equation}
        \begin{split}
            \Tr[\Pi_B  \sigma_{t-1}] &= \Tr[\Pi_B \cT_\delta^{t-1}(\sigma_0) ] \\
            &= \Tr(\Pi_A \rho_\beta)^{-1} \Tr[\Pi_B \cT_\delta^{t-1}(\rho_\beta^{1/2} \Pi_A \rho_\beta^{1/2}) ] \\
            &\leq \Tr(\Pi_A \rho_\beta)^{-1} \Tr[\Pi_B \cT_\delta^{t-1}(\rho_\beta) ] \\
            &= \Tr(\Pi_A \rho_\beta)^{-1} \Tr[\Pi_B \rho_\beta],
        \end{split}
    \end{equation}
    where the inequality follows from $\cT_\delta$ being a positive linear map with fixed point $\rho_\beta$. 

    Writing $a=\Tr[\Pi_A\rho_\beta]$, $b=\Tr[\Pi_B\rho_\beta]$, $c=\Tr[\Pi_C\rho_\beta]$ and $q=\Tr[\Pi_C\sigma_0]$ as in the previous proof, we have $\Tr[\Pi_B\sigma_t]\le b/a$.
    Considering $\Tr[\Pi_C\sigma_{j+1}]$, and letting $Q=\Pi_B+\Pi_C$, we have
    \begin{equation}\label{eq:pre_iterate_pi_B_continuous}
    \begin{split}
     \Tr[\Pi_C\sigma_{j+1}]
     &\le\Tr[\Pi_C\widetilde\cT_\delta(\sigma_j)]+\eta_\delta\\
     & =\Tr[\Pi_C\widetilde\cT_\delta(Q\sigma_jQ)]+\eta_\delta\\
     &\le\Tr[\Pi_C\cT_\delta(Q\sigma_jQ)]+2\eta_\delta\\
     &\le\Tr[\cT_\delta(Q\sigma_jQ)]+2\eta_\delta\\
     &=\Tr[\Pi_C\sigma_j]+\Tr[\Pi_B\sigma_j]+2\eta_\delta.
    \end{split}
    \end{equation}
    Here $Q\sigma_jQ$ is positive and has trace at most one. Thus both norm errors are at most $\eta_\delta$. The last two lines use positivity and trace preservation of $\cT_\delta$, not of $\widetilde\cT_\delta$.
    Iterating yields
    \begin{equation}
     \Tr[\Pi_C\sigma_j]
     \le q+j\left(b/a+2\delta\epsilon+ \delta^2 \|\cL\|_\diamond^2\right).
    \end{equation}
    The assumed exponential bound gives
    $b/a\le c/\exp[c_1n^\gamma]$ and $\epsilon\le c/\exp[c_1n^\gamma]$.
    Choose
    \begin{equation}
     \delta=\exp[c_1n^\gamma]^{-1/2}/\|\cL\|_\diamond,\qquad
     j=\left\lfloor c\exp[c_1n^\gamma]/32\right\rfloor.
    \end{equation}
    These choices give $\Tr[\Pi_C\sigma_j]\le q+c/8$.
    Since $q=o(c)$ and $b=o(c)$, the mixing time comparison at the end of Proposition~\ref{prop:mix_time_discrete} shows that this state has not mixed. The time is at least $j\delta$, and therefore
    \begin{equation}
     \tmix(\cL)\ge j\delta
     =\Omega\left(\frac{c}{\|\cL\|_\diamond}\exp[c_1n^\gamma/2]\right)
     =\exp[\Omega(n^\gamma)].
    \end{equation}
    The last step uses $c\ge n^{-\gamma_0}$ and $\|\cL\|_\diamond=n^{O(1)}$.
    
\end{proof}

\section{Systems with slow mixing}
\label{sec:systems}
We prove here the main results which show various models of Hamiltonians mix slowly at low enough temperatures. 

\subsection{Classical models}
\label{sec:classical}
Classical Hamiltonians are those for which the Hamiltonian $H$ is diagonal in the computational basis. Equivalently, these models can be written as polynomials in the Pauli operators which are only in the $Z$ basis. 
\begin{equation}\label{eq:classical_hamiltonian_form}
    \text{(classical Hamiltonian): } \quad H = \sum_{P \in \{E_1 \otimes \cdots \otimes E_n | E_i=I,Z\}} c_P P,
\end{equation}
where $c_P \in \mathbb{R}$ are real coefficients. 

For such Hamiltonians, the eigenbasis consists of bitstrings in the computational basis and bottlenecks constructed classically between bitstrings that are far apart in Hamming distance also apply in the quantum setting.

In classical MCMC, each step of the sampling algorithm flips at most a constant number of bits to produce a candidate state. Bottlenecks in classical models are thus typically formed by finding a subspace where states close to that subspace in Hamming distance have low probability under the Gibbs measure. These bottlenecks also lower bound mixing times for quantum Gibbs samplers.

\begin{proposition}[Bottlenecks in classical models]
\label{prop:slowclassical}
    Let $H$ be a classical Hamiltonian of the form in \Cref{eq:classical_hamiltonian_form}. Given jump operators $\{A^a\}_{a=1}^n$ that are at most $k$-local, let $\cT$ denote a quantum Gibbs sampler satisfying $\trange(\cT) \leq \ell$ (i.e. it is a degree $\ell$ operator in $\Tset(H, \{A^a\})$). Let $A,C \subseteq \mathbb{H}_n$ be orthogonal linear subspaces of $\bH_n$ with projectors $\Pi_A, \Pi_C \in \mathcal{O}_{\bH_n}$ whose basis consists of computational basis states where the Hamming distance satisfies $|A-C|_H>k\ell$. Denote the bottleneck subspace $B \subseteq \mathbb{H}_n$ as spanned by bitstrings within Hamming distance $k\ell$ from $C$
    \begin{equation}
        B = \operatorname{span}\{\ket{x}: 1 \leq |x-C|_H \leq k\ell \}.
    \end{equation}
    Then it holds that
    \begin{equation}
        \tmix(\cT) = \Omega\lr{\Tr[\Pi_A \rho_\beta]\Tr[\Pi_C \rho_\beta]\Tr[\Pi_B \rho_\beta]^{-1}}.
    \end{equation}
\end{proposition}
\begin{proof}
    When the jump operators $\{A^a\}_{a=1}^n$ are at most $k$-local operators (e.g., single qubit Pauli $X$ operators), transitions between computational basis states $\ket{x}$ and $\ket{y}$ will only differ in at most $k$ bits. In this case, the set of transitions $\Tset$ from \Cref{def:set_of_transitions} contains only transitions between eigenbases whose Hamming distance is at most $k$:
    \begin{equation}
        \Tset(H, \{A^a\}) \subseteq \left\{ w \, \ket{x}\bra{y}: \; w \in \mathbb C, \; x,y \in \{0,1\}^n , |x-y|_H \leq k \right\},
    \end{equation}
    and thus $\dist_{\Tset}(A, C) > \trange(\cT)$.
    Therefore, the assumptions of \Cref{prop:mix_time_discrete} apply, and we obtain the corresponding mixing time in \Cref{prop:mix_time_discrete}.
\end{proof}

As a reminder, one can use \Cref{lem:slow} to obtain a similar mixing time lower bound for continuous time Lindbladian Gibbs samplers.

\subsection{Commuting Hamiltonians and stabilizer codes}
\label{sec:codes}
Commuting Hamiltonians are instances of Hamiltonians $H=\sum_i h_i$ where each pair of terms $h_i,h_j$ commute with each other and $[h_i,h_j]=0$. These Hamiltonians are often easier to analyze as the eigenbasis of the Hamiltonian can be obtained from the fact that each of the $h_i$ are mutually diagonalizable by the same unitary transformation. Commuting Hamiltonians are the basis for the class of stabilizer codes, and \cite{hong2024quantum} proves slow mixing for quantum codes at constant temperature. Their proof uses a code that has distance scaling as $\sqrt{n}$ achieving a mixing time lower bound of $\exp(\Omega(\sqrt{n}))$. Our proof which applies to codes of distance $\Theta(n)$ achieves a mixing time lower bound of $\exp(\Omega(n))$. Nonetheless, our proof techniques are similar in that both use expansion properties of the quantum code to prove slow mixing.

\begin{definition}[Quantum code]
    An $[[n,k,d]]$ quantum code $\cC$ is a subspace $\cC \subseteq \C^{2^n}$ of dimension $\dim \cC = 2^k$ such that, for all $\ket{\psi_1}, \ket{\psi_2} \in \cC$ satisfying $\bra{\psi_1}\ket{\psi_2} = 0$, and for all Pauli operators $E \in \cP_n$ of weight less than $d$, $\bra{\psi_1}E \ket{\psi_2}=0$.
\end{definition}
In physical terms, an $[[n,k,d]]$ quantum code is one that encodes $k$ logical qubits using $n$ physical qubits and detects errors up to distance $d$. A code is linear distance if $d = \Theta(n)$. 
We will restrict our attention to stabilizer codes, which admit a convenient Hamiltonian description with an immediate diagonalization.

\begin{definition}[Stabilizer code]
    An $[[n,k,d]]$ quantum code $\cC$ is a stabilizer code if there exists a collection $\{C_i\}_{i=1}^{m}$ of commuting Hermitian Pauli operators $C_i \in \cP_n$ called check operators that generate group $\cS$ (by multiplication), such that $\cS$ satisfies
    \begin{equation}
        \cC = \{\ket{\psi}:M\ket{\psi} = \ket{\psi} \forall M \in \cS\}.
    \end{equation}
    We say $\cC$ is $\ell$-local if every generator $C_i$ is weight at most $\ell$.
\end{definition}

Given an $[[n,k,d]]$ quantum code with check operators $C_1, \dots, C_m$, the code Hamiltonian we will analyze is
\begin{equation}\label{eq:code_hamiltonian_with_redundancy}
    H = \sum_{i=1}^{m} \frac{I-C_i}{2}.
\end{equation}
Note that the ground state of the above Hamiltonian has energy $0$ and
when $m > n-k$, there are redundant checks.
We describe standard properties of the associated Hamiltonian to a stabilizer code.

\begin{fact}
\label{fact:eigenbasis}
    Given an $[[n,k,d]]$ quantum code, let $C_1, \dots, C_{n-k}$ be a set of $n-k$ independent stabilizer generators.
    For every $s \in \{0,1\}^{n-k}$, define subspace
    \begin{equation}
        D_s = \{\ket{\psi} : C_i\ket{\psi} = (-1)^{s_i}\ket{\psi}  \forall i \in [n-k]\}.
    \end{equation}
    The Hamiltonian restricted to these independent checks
    \begin{equation}\label{eq:checks_for_code}
        H_{\mathrm ind} = \sum_{i=1}^{n-k} \frac{I-C_i}{2}.
    \end{equation}
    is diagonalized by eigenvalues $\lambda = |s|$ and subspaces $\{D_s:|s|=\lambda\}$.
\end{fact}

To complete the proof of slow mixing in our setting, we will resort to expansion properties of the code. Recent constructions of so-called quantum low-density parity-check (LDPC) codes are instances of stabilizer codes that have nice expansion properties \cite{leverrier2022quantum,panteleev2022asymptotically}.

\begin{definition}[Code bipartite graph]\label{def:code_graph}
    Given a $[[n,k,d]]$ quantum code $\mathcal{C}$ with check matrices $C_1,\dots,C_{m}$, its bipartite graph $G(\mathcal{C})$ has $n$ nodes on the left corresponding to qubits and $m$ nodes on the right corresponding to checks $C_1,\dots,C_{m}$ as in \Cref{eq:code_hamiltonian_with_redundancy}. There is an edge between node $i$ on the left and node $j$ on the right if $C_j$ acts nontrivially on qubit $i$. 
\end{definition}

\begin{definition}[$(\gamma, \alpha)$ small-set expansion] \label{def:small_set_expander}
    Given a linear distance $[[n,k,d]]$ quantum code $\mathcal{C}$ with checks $C_1, \dots, C_{m}$ and its bipartite graph $G(\mathcal{C})$, let $a\in \{I,X,Y,Z\}^n$ be an assignment to the nodes on the left side of the bipartite graph. Let $P_a=\sigma^{a_1} \otimes \cdots \otimes \sigma^{a_n}$ denote the $n$-qubit Pauli corresponding to the assignment. The weight of the assignment $w(a)$ is equal to the number of terms in $a$ not equal to $I$:
    \begin{equation}
        w(a) = |\{i \in [n]: a_i \neq I\}|.
    \end{equation}
    For a given assignment $a$, we define $V(a)$ as the number of violations: 
    \begin{equation}
        V(a) = \left|\{ C_i: [P_a,C_i] \neq 0 \}\right|.
    \end{equation}
    Define also the stabilizer-reduced weight
    \begin{equation}
        w_{\mathrm{red}}(a) =\min_{S\in\cS}\operatorname{wt}(P_aS).
    \end{equation}
    The chosen checks satisfy $(\gamma,\alpha)$ small-set expansion if $0<\alpha<d/(2n)$ and
    \begin{equation}
         V(a)\ge\gamma w_{\mathrm{red}}(a)
         \quad\text{for every }a\text{ with }w_{\mathrm{red}}(a)\le\alpha n.
    \end{equation}
    
\end{definition}

As an example, suitable quantum Tanner codes of Leverrier and Z\'emor \cite{leverrier2022quantum} have this property. A two-sided small-set boundary expansion, established in \cite[Theorem~8.1]{hopkins2022explicit}, shows that this holds (see also expansion property in \cite{anshu2023nlts}).
For CSS check matrices $H_Z$ and $H_X$ specifying the $Z$ and $X$ stabilizers respectively, their result gives constants $\alpha_*,\gamma_*>0$ such that
\begin{equation}
 |H_Zx|\ge\gamma_*\dist(x,\operatorname{row}H_X),\qquad
 |H_Xz|\ge\gamma_*\dist(z,\operatorname{row}H_Z)
\end{equation}
whenever the respective reduced weight is at most $\alpha_*n$.
If $P_a=i^tX^xZ^z$, and the two reduced component weights are $u=w_{\mathrm{red}}(x),v=w_{\mathrm{red}}(z)$, then $u,v\le w_{\mathrm{red}}(a)\le u+v$. Therefore $V(a)=|H_Zx|+|H_Xz| \ge \gamma_*w_{\mathrm{red}}(a)$ and small-set expansion holds for $\alpha<\min\{\alpha_*,d/(2n)\}$. 

We now show slow mixing from small-set expansion.
\begin{proposition}
\label{prop:good_code_mixing}
    Let $\rho_\beta$ be the Gibbs state of the Hamiltonian associated with a $[[n, k, d]]$ stabilizer code (\Cref{eq:code_hamiltonian_with_redundancy}) with $k\ge1$ whose checks satisfy $(\gamma, \alpha)$ small-set expansion (\Cref{def:small_set_expander}). Assume jump operators act on at most $\ell_1$ qubits. Then any Lindbladian Gibbs sampler satisfying $\trange(\cL) \leq \ell_2 n$ with $\ell_2 \ell_1<\alpha$ satisfies $\tmix = e^{\Omega(n)}$ for $\beta>\frac{\log(2)}{\gamma(\alpha-\ell_1 \ell_2)}$.\footnote{Note that $\ell_2 = o(1)$ in many cases. For example, the Davies generator with single-qubit Pauli jump operators being the single qubit Paulis has bounded transition degree between eigenbases in \Cref{fact:eigenbasis}.}
\end{proposition}
\begin{proof}
    As before, let $C_1, \dots, C_m$ denote the checks for the code ordered WLOG so that $C_1, \dots, C_{n-k}$ correspond to a set of $n-k$ independent stabilizer generators for the code. To explicitly write a full eigenbasis, we construct a set of additional commuting logical Pauli operators $C_1',\dots, C_{k}' $ which uniquely determine $2^n$ eigenvectors using \Cref{fact:eigenbasis}. Note that $C_1',\dots, C_{k}' $ do not contribute to the Hamiltonian.

    Each vector in this eigenbasis can now be written using a string $s\in\{0,1\}^n$ where $s=s^a\oplus s^b$ and $s^a\in\{0,1\}^{n-k}$ and $s^b\in\{0,1\}^k$. $s^a$ denotes violations of the independent stabilizer generators and $s^b$ denotes violations of the additional nonlocal checks. The code subspace corresponds to all possible strings of the form $0^{\oplus n-k}\oplus s^b$. Let us arbitrarily split the bitstrings $s^b$ in half by defining the set 
    $S_R=\{0^{n-k}\oplus s^b: s^b_1 = 0\}$ containing those strings
    where the first bit of $s^b$ is equal to $0$. Let $S_L=\{0^{n-k}\oplus s^b:s_1^b=1\}$ be the other half.

     Let $\ket{\psi_s}$ denote the unique eigenstate corresponding to bitstring $s$. Let $G(\mathcal{C}')$ now correspond to the bipartite graph (\Cref{def:code_graph}) of the code with the additional checks. This fixes a complete eigenbasis from \Cref{fact:eigenbasis}. We define $E_w(S_R)$ as the set of eigenstates reachable by Pauli operations to eigenstates in $S_R$ that act non-trivially on at most $w$ qubits, i.e. using the notation of \Cref{def:small_set_expander}:
     \begin{equation}
         E_w(S_R) = \left\{ \ket{\psi_s} : \exists a \in \{I,X,Y,Z\}^n, s' \in S_R, \ket{\psi_s} \propto P_a \ket{\psi_{s'}}   , w(a) \leq w \right\}.
     \end{equation}

     We will use \Cref{prop:mix_time_discrete} with the starting state subspace $A$ and bottleneck subspace $B$ defined as:
     \begin{equation}
         \begin{split}
             A = \operatorname{span}\left( E_{(\alpha-\ell_1 \ell_2)n}(S_R) \right), \quad
             B = \operatorname{span}\left( E_{\alpha n}(S_R) \setminus E_{(\alpha-\ell_1 \ell_2)n}(S_R) \right).
         \end{split}
     \end{equation}
     In words, the subspace $A$ are the eigenstates reachable by Paulis applied to at most $(\alpha-\ell_1 \ell_2)n$ qubits and the bottleneck $B$ are the eigenstates not in $A$ and reachable by Paulis applied to at most $\alpha n$ qubits. 
     Since any linear transformation can be written in terms of Pauli operators, this decomposition into local Paulis covers all the operations of a local Gibbs sampler.
     Subspace $C$ spans the rest of the Hilbert space not spanned by $A$ and $B$.

    Let $\Pi_A,\Pi_B$ be projectors onto the subspaces $A,B$ respectively. 
    For an eigenbasis vector in $B$, choose a minimum-weight Pauli $P_a$ taking some ground-state basis vector indexed by $S_R$ to that eigenbasis vector in $B$ up to a phase. Its weight satisfies $(\alpha-\ell_1\ell_2)n<w(a)\le\alpha n$. Multiplying by a stabilizer generator leaves this state invariant up to a phase, so the minimality assumption gives $w_{\mathrm{red}}(a)=w(a)$. Small-set expansion therefore gives energy $V(a)\ge\gamma w(a)>\gamma(\alpha-\ell_1\ell_2)n$. Every state $\ket{\psi}\in B$ therefore satisfies $\bra{\psi}H\ket{\psi}\ge\gamma(\alpha-\ell_1\ell_2)n$.
      Thus, we have
     \begin{equation}\label{eq:code_pib_bound}
         \Tr[\Pi_B \exp(-\beta H)] \leq 2^n \exp(-\beta \gamma (\alpha-\ell_1 \ell_2)n).
     \end{equation}
     The $2^n$ factor trivially bounds the dimension of the subspace of $\Pi_B$ which is at most $2^n$. Therefore, for $\beta>\frac{\log(2)}{\gamma(\alpha-\ell_2 \ell_1)}$, it holds that $\Tr[\Pi_B \exp(-\beta H)] = \exp(-\Omega(n))$.
     

    We will now use \Cref{prop:mix_time_discrete}.
    Note that $\Tr[\Pi_A \exp(-\beta H)] \geq 2^{k}/2$ since there are at least $2^{k}/2$ orthogonal code states in $A$ of energy $0$. 
    Setting $\Pi_C = I - \Pi_A - \Pi_B$, any state in $A$ cannot reach $C$ by a single action of the Gibbs sampler. Indeed, expand an input jump operator in Paulis on its support. A nonzero transition matrix element in the complete stabilizer eigenbasis is realized by at least one Pauli in this expansion. A product of at most $\ell_2n$ transitions is therefore realized by a Pauli of weight at most $\ell_1\ell_2n$. Its action on $A$ remains in $E_{\alpha n}(S_R)$. 
    Furthermore, since $d>2\alpha n$, the sets $E_{\alpha n}(S_R)$ and $E_{\alpha n}(S_L)$ are disjoint. Otherwise a product of two errors of total weight at most $2\alpha n$ would connect orthogonal code states. We also have $\Tr[\Pi_C\rho_\beta]\ge\Tr[\Pi_A\rho_\beta]$, and thus $\Tr[\Pi_C\rho_\beta]\ge(1-\Tr[\Pi_B\rho_\beta])/2$.
    The previous bound in \Cref{eq:code_pib_bound} gives $\Tr[\Pi_B\rho_\beta]/\Tr[\Pi_A\rho_\beta]=e^{-\Omega(n)}$. Putting these facts together, we conclude by applying \Cref{prop:mix_time_discrete}. The same support argument for a Hermiticity-preserving generator gives the continuous-time conclusion from \Cref{lem:slow}, with $\widetilde\cL=\cL$.
\end{proof}

\subsection{Non-commuting Hamiltonians with a linear free energy barrier}
\label{sec:linear}
Informally, we show in this subsection that for a classical Hamiltonian $H_0$ (i.e., diagonal in the $Z$ basis) with slow mixing due to a linear free energy barrier, the (stoquastic) non-commuting Hamiltonian with transverse field
\begin{equation}
    H = H_0 - h \sum_i X_i
\end{equation}
also experiences slow mixing for sufficiently small constant $h$. Our proof uses the classical bottleneck, i.e., projectors $\Pi_A, \Pi_B, \Pi_C$ that are diagonal in the $Z$ basis. If the regions $A$ and $C$ are sufficiently far in Hamming distance, then local Lindbladian Gibbs samplers must pass through the bottleneck region. To use \Cref{lem:slow}, we must show that the projectors satisfy two properties relative to the Hamiltonian: the bottleneck region must be exponentially small under the Gibbs measure, and the regions $A$ and $C$ must be decoupled.

To show the first property, we use the Feynman-Kac representation of the thermal state. The quantity $e^{-\beta H}$ is written as a path integral from its Trotter expansion, yielding independent Poisson processes. The following lemma, based on the techniques of~\cite{crawford2007thermodynamics,leschke2021free}, holds for any stoquastic model in a transverse field.

\begin{lemma}[Poisson Feynman-Kac representation of Hamiltonian~\cite{crawford2007thermodynamics,leschke2021free}]
\label{lem:fk}
    Let $H$ be a Hamiltonian of the form
    \begin{equation}
        H = -\sum_{i_1, \dots, i_k = 1}^n c_{i_1 \cdots i_k} Z_{i_1} \cdots Z_{i_k} - h \sum_{i=1}^n X_i
    \end{equation}
    for real coefficients $c_{i_1 \cdots i_k}$ and $h \geq 0$. Let $\sigma \in \{\pm 1\}^n$ denote a spin configuration in the computational basis. For all $i \in [n]$, define $s_i(t) \in \{\pm1\}$ by the spin-flip process
    \begin{equation}
        s_i(t) = (-1)^{\cN_i(t)}, \quad \cN_i(t) \sim \mathrm{Pois}(\beta h t),
    \end{equation}
    where $\cN_i(t)$ are independent $\mathbb{N}$-valued Poisson processes on the positive half-line with mean $\beta h t$ and rate $\beta h$. Then
    \begin{equation}
        \bra{\sigma}e^{-\beta H}\ket{\sigma'} = \lr{\cosh \beta h}^n \lr{\tanh\beta h}^{|\sigma-\sigma'|_H} \left\langle \exp\left[\beta\int_0^1 \sum_{i_1, \dots, i_k = 1}^n c_{i_1 \cdots i_k} \sigma_{i_1} s_{i_1}(t) \cdots \sigma_{i_k} s_{i_k}(t) dt\right] \right\rangle_{\beta h},
    \end{equation}
    where $\langle \cdot \rangle_{\beta h}$ denotes averaging over Poisson point processes conditioned on $\sigma_is_i(1)=\sigma_i'$ for all $i \in [n]$, and $|\sigma-\tilde\sigma|_H$ is the Hamming distance between configurations.
\end{lemma}
\begin{proof}
Observe that the matrix
\begin{equation}
    Q = -h \sum_{i=1}^n X_i + D, \quad D = hn \cdot I
\end{equation}
has non-positive off-diagonal elements and rows that sum to zero. Thus $-Q$ is the generator matrix of a continuous-time Markov chain of independent spin flips at rate $h$. Letting
\begin{equation}
    A = -D - \sum_{i_1, \dots, i_k = 1}^n c_{i_1 \cdots i_k} Z_{i_1} \cdots Z_{i_k}, \quad B = D - h\sum_{i=1}^n X_i,
\end{equation}
we consider the Lie-Trotter expansion
\begin{equation}
    e^{-\beta H} = \lim_{k \to \infty} \lr{e^{-\beta A/k} e^{-\beta B/k}}^k.
\end{equation}
Let $\sigma \in \{\pm 1\}^n$ denote an element of the computational basis; we use this spin configuration notation rather than the standard bitstring notation in order to write the Hamiltonian more simply.
At some given $k$, the matrix element of $\lr{e^{-\beta A/k} e^{-\beta B/k}}^k$ indexed by basis elements $\sigma, \tilde \sigma \in \{\pm 1\}^n$ is given by
\begin{equation}
    \bra{\sigma} \lr{e^{-\beta A/k}e^{-\beta B/k}}^k \ket{\tilde \sigma} = \sum_{\substack{\sigma(0), \dots, \sigma(k) \in \{\pm 1\}^n\\ \sigma(0) = \sigma, \sigma(k) = \tilde \sigma}} \prod_{i=0}^{k-1} \bra{\sigma(i)}e^{-\beta A/k}e^{-\beta B/k} \ket{\sigma(i+1)}.
\end{equation}
The key observation of the Poisson Feynman-Kac formula is to evaluate this quantity as a path integral determined by the rate matrix $Q$. Formally, the integral is performed over the measurable space $\Omega$ of $\{\pm 1\}^n$ valued c\'adl\'ag paths, i.e., paths that are right continuous with left limits. Denoting a spin path $\{\sigma(s):s\in[0,1]\}$ by $\underline{\sigma}$, where $\sigma(s)$ is the value of the process at time $s$, we let $\dd\nu_\sigma(\underline{\sigma})$ be the induced Markov chain path measure starting from $\sigma$. For eigenvalues of the diagonal part of the Hamiltonian
\begin{equation}
    \mu(\sigma) = -\sum_{i_1, \dots, i_k = 1}^n c_{i_1 \cdots i_k} \sigma_{i_1} \cdots \sigma_{i_k} - hn,
\end{equation}
we can write a matrix element of the Trotterization as
\begin{equation} \label{eq:integral_pre_poisson}
    \bra{\sigma} \lr{e^{-\beta A/k}e^{-\beta B/k}}^k \ket{\tilde \sigma} = \int_\Omega \exp\left[\sum_{i=0}^{k-1} -\frac{\beta}{k}\mu\lr{\sigma\lr{\frac{i}{k}}}\right] 1_{\sigma(1)=\tilde \sigma} \dd\nu_\sigma(\underline{\sigma}).
\end{equation}
By the bounded convergence theorem, we can take the $k\to\infty$ limit inside the integral, giving
\begin{equation}
    \bra{\sigma} e^{-\beta H}\ket{\tilde \sigma} = \int_\Omega \exp\left[-\beta\int_0^1 \mu(\sigma(s)) ds\right] 1_{\sigma(1)=\tilde \sigma} \dd\nu_\sigma(\underline{\sigma}).
\end{equation}
To evaluate the integral, we note that $\dd\nu_\sigma(\underline{\sigma})$ corresponds to starting from an initial configuration $\sigma$ and evolving in time via spin flips which occur independently at each site according to the arrivals of a Poisson process of rate $\beta h$. The probability the $i$th spin has flipped $k$ times over the time interval $[0, 1]$ is then $e^{-\beta h} (\beta h)^k/k!$. We make this Poisson point process explicit by writing the matrix element as an expectation over the space of paths
\begin{equation}
\begin{split}
    \bra{\sigma}e^{-\beta H}\ket{\tilde \sigma} &= \lr{\cosh(\beta h)}^n \lr{\tanh(\beta h)}^{|\sigma-\tilde\sigma|_H} \\
    &\qquad \times \left\langle \exp\left[\beta\int_0^1 \sum_{i_1, \dots, i_k = 1}^n c_{i_1 \cdots i_k} \sigma_{i_1} \cdots \sigma_{i_k} (-1)^{\cN_{i_1}(s)+\cdots+\cN_{i_k}(s)} ds\right] \right\rangle_{\beta h},
\end{split}
\end{equation}
where each $\cN_i(s)$ for $i \in [n]$ and $t \geq 0$ is $\mathbb{N}$-valued simple Poisson process with mean $\beta h t$. The processes $\cN_1, \dots, \cN_n$ are mutually independent. The brackets $\langle\cdot \rangle_{\beta h}$ indicate an expectation over paths conditional on $\sigma_i (-1)^{\cN_i(1)} = \tilde\sigma'_i$. To obtain the prefactor, note that the probability of this condition being satisfied for any $i$ is equal to
\begin{equation}
    e^{-\beta h} \sum_{k=0}^\infty \frac{(\beta h)^{2k+(1-\sigma_i\tilde\sigma_i)/2}}{(2k+(1-\sigma_i\tilde\sigma_i)/2)!} = \begin{cases}
        e^{-\beta h} \cosh(\beta h) & \sigma_i = \tilde\sigma_i\\
        e^{-\beta h} \sinh(\beta h) & \sigma_i \neq \tilde \sigma_i.
    \end{cases}
\end{equation}
Taking a product over all $i \in [n]$, we get a prefactor $e^{-\beta hn} \lr{\sinh(\beta h)}^{|\sigma-\tilde\sigma|_H} \lr{\cosh(\beta h)}^{n - |\sigma-\tilde\sigma|_H}$; this is multiplied by $e^{\beta h n}$ from the eigenvalue $\mu$. The result reported in the lemma statement is obtained by a change of variables.
\end{proof}

We use the above result to bound the change in support of a projector diagonal in the computational basis.

\begin{lemma}
\label{lem:hbound}
Let
\begin{equation}
    H = -\sum_{i_1, \dots, i_k = 1}^n c_{i_1 \cdots i_k} Z_{i_1} \cdots Z_{i_k} - h \sum_{i=1}^n X_i
\end{equation}
for real coefficients satisfying $|c_{i_1 \cdots i_k}| \leq c/d$, where each qubit appears in at most $d$ Hamiltonian terms. Let $H_0=-\sum c_{i_1\cdots i_k} Z_{i_1}\cdots Z_{i_k}$ and $V=-\sum_i X_i$. Then for all $0 \leq h < 1/\beta$, and for any projector $\Pi$ diagonal in the computational basis,
\begin{equation}
    \left|\log\frac{\Tr[\Pi e^{-\beta (H_0+hV)}]}{\Tr[\Pi e^{-\beta H_0}]}\right| \leq \lr{h\beta e^{\beta c}}^2 n.
\end{equation}
Writing $\rho_\beta^{(0)}$ for the Gibbs state of $H_0$, then
\begin{equation}
    e^{-n(h\beta e^{\beta c})^2}
    \le \frac{\Tr[\Pi\rho_\beta]}{\Tr[\Pi\rho_\beta^{(0)}]}
    \le e^{n(h\beta e^{\beta c})^2}.
\end{equation}
\end{lemma}
\begin{proof}
Since $\Pi$ is diagonal in the computational basis, Lemma~\ref{lem:fk} gives for spin configurations (computational basis states) $\sigma \in \{\pm1\}^n$,
\begin{equation}
    \Tr[\Pi e^{-\beta(H_0+hV)}] = (\cosh \beta h)^n \sum_{\sigma \in \Pi} \left\langle\exp\left[\beta \int_0^1 \sum_{i_1,\dots,i_k=1}^n c_{i_1,\dots,i_k}\sigma_{i_1}s_{i_1}(t)\cdots\sigma_{i_k}s_{i_k}(t)dt\right] \right\rangle_{\beta h},
\end{equation}
where the expectation is conditioned on $s_i(1)=1$. The probability of $m$ spins flipping an even positive number of times and everything else flipping 0 times is
\begin{equation}
\begin{split}
    Q_m &= e^{-\beta h n}\binom{n}{m} \lr{\sum_{k \geq 1} \frac{(\beta h)^{2k}}{(2k)!}}^m \lr{\frac{(\beta h)^{0}}{0!}}^{n-m}\\
    &= e^{-\beta h n}\binom{n}{m} (-1 + \cosh \beta h)^m
\end{split}
\end{equation}
and hence, the probability of $m$ spins flipping a positive even number of times conditioned on $s_i(1)=1$ for all $i$ is
\begin{equation}
\begin{split}
    P_m &= \frac{e^{-\beta h n}\binom{n}{m} (-1 + \cosh \beta h)^m}{\sum_{m=0}^n Q_m}\\
    &= (\cosh \beta h)^{-n} \binom{n}{m} (-1 + \cosh \beta h)^m.
\end{split}
\end{equation}
Assuming these spins flip adversarially, we obtain
\begin{equation}
    \frac{\Tr[\Pi e^{-\beta(H_0+hV)}]}{\Tr[\Pi e^{-\beta H_0}]} \leq \left[ 1+e^{2\beta c}(\cosh(\beta h) - 1)  \right]^n.
\end{equation}
We can bound this with inequalities $\cosh(x) \leq \exp(x^2/2)$ and $\log(1+x)\leq x$ to obtain 
\begin{equation}
    \frac{\Tr[\Pi e^{-\beta(H_0+hV)}]}{\Tr[\Pi e^{-\beta H_0}]} \leq \exp\left[ n e^{2 \beta c}(e^{\beta^2h^2/2} - 1) \right],
\end{equation}
and using the bound $e^x\le 1+2x$ for $0\le x\le 1/2$, we have assuming $h<1/\beta$, 
\begin{equation}
    \frac{\Tr[\Pi e^{-\beta(H_0+hV)}]}{\Tr[\Pi e^{-\beta H_0}]} \leq \exp\left[ n e^{2 \beta c}\beta^2h^2 \right].
\end{equation}
A similar bound in the opposite direction gives
\begin{equation}
    \frac{\Tr[\Pi e^{-\beta(H_0+hV)}]}{\Tr[\Pi e^{-\beta H_0}]} \geq \left[ 1+e^{-2\beta c}(\cosh(\beta h) - 1)  \right]^n \geq 1.
\end{equation}
If the ratio of traces is positive, then the above inequalities imply
\begin{equation}
    \frac{\Tr[\Pi e^{-\beta H_0}]}{\Tr[\Pi e^{-\beta(H_0+hV)}]} \geq \exp\left[- n e^{2 \beta c}\beta^2h^2 \right], \quad \frac{\Tr[\Pi e^{-\beta H_0}]}{\Tr[\Pi e^{-\beta(H_0+hV)}]} \leq 1.
\end{equation}
Normalizing to obtain bounds on the Gibbs state, we take $\Pi=I$ to get
\begin{equation}
    \exp\left[-n \lr{h \beta e^{\beta c}}^2\right] \leq \frac{\Tr[\Pi e^{-\beta(H_0+hV)}]}{\Tr[e^{-\beta(H_0+hV)}]} \cdot \frac{\Tr[e^{-\beta H_0}]}{\Tr[\Pi e^{-\beta H_0}]} \leq \exp\left[n \lr{h \beta e^{\beta c}}^2\right],
\end{equation}
completing the proof.
\end{proof}

The above lemma allows us to bound the support of $\Pi_B$ on the Gibbs state $e^{-\beta H}$ based on its support on $e^{-\beta H_0}$. We now show that the bottleneck weight ratios and decoupling persist.

\begin{lemma}\label{lem:linearstrong}
    Let $H_0 = -\sum_{i_1\cdots i_k} c_{i_1\cdots i_k} Z_{i_1}\cdots Z_{i_k}$ be an $n$-qubit classical Hamiltonian with $|c_{i_1\cdots i_k}| \leq 1/\deg(H_0)$, where $\deg(H_0)$ is the maximum number of Hamiltonian terms that any qubit appears in. Let $H = H_0 - h \sum_i X_i$. Let $A,C \subseteq \mathbb{H}_n$ be orthogonal linear subspaces of $\bH_n$ with projectors $\Pi_A, \Pi_C \in \mathcal{O}_{\bH_n}$ diagonal in the computational basis states such that the Hamming distance satisfies $|A-C|_H>d n$ for some constant $d > 0$. Denote the bottleneck subspace $B \subseteq \mathbb{H}_n$ by projector $\Pi_B = I - \Pi_A - \Pi_C$.
    Let $\rho_\beta$ and $\rho_\beta^{(0)}$ correspond to Gibbs states of $H$ and $H_0$ respectively. Assume there exists a constant $\beta$ such that the classical Gibbs distribution satisfies
    \begin{equation}
    \label{eq:lineargap}
        \Tr[\Pi_A \rho_\beta^{(0)}] = \Omega(1), \quad \Tr[\Pi_C \rho_\beta^{(0)}] = \Omega(1), \quad \Tr[\Pi_B \rho_\beta^{(0)}] \leq \exp[-a n]
    \end{equation}
    for some constant $a > 0$. Then for all
    \begin{equation}
        0 \leq h < \min\left\{\frac{1}{\beta}, \frac{\sqrt a}{\sqrt 3 \beta e^\beta}\right\}.
    \end{equation}
    After possibly swapping $A$ and $C$, the quantum Gibbs state satisfies
    \begin{equation}
         \Tr[\Pi_C\rho_\beta]=\Omega(1),\qquad
         \frac{\Tr[\Pi_A\rho_\beta]\Tr[\Pi_C\rho_\beta]}
              {\Tr[\Pi_B\rho_\beta]}=\exp[\Omega(n)],
    \end{equation}
    and the initialization $\sigma_0$ satisfies $\Tr[\Pi_C\sigma_0]\le e^{-\Omega(n)}$.
\end{lemma}
\begin{proof}
    Put $\kappa=(h\beta e^\beta)^2<a/3$, and write
    $w_X=\Tr[\Pi_X\rho_\beta]$ and
    $w_X^{(0)}=\Tr[\Pi_X\rho_\beta^{(0)}]$ for $X=A,B,C$.
    By Lemma~\ref{lem:hbound},
    \begin{equation}
        \begin{split}
            w_A,w_C&\geq e^{-\kappa n}\Omega(1),\\
            w_B&\leq e^{-(a-\kappa)n},\\
            \frac{w_Aw_C}{w_B}
            &\geq\frac{w_A^{(0)}w_C^{(0)}}{w_B^{(0)}}
                       e^{-3\kappa n}
             =e^{\Omega(n)}.
        \end{split}
    \end{equation}
    Since $w_B=o(1)$ and $w_A+w_B+w_C=1$, exchange $A,C$ if necessary so that $w_C\geq(1-w_B)/2=\Omega(1)$.
    We also have
    \begin{equation}
        \frac{w_B}{w_A}\leq O(1)e^{-(a-2\kappa)n}.
    \end{equation}
    The lower bound on $w_A$ may be exponentially small, so this argument does not assert weak clustering of these regions.

    We now prove the required decoupling. The case $h=0$ is immediate because the projectors commute with $H_0$. For $h>0$, we use a Taylor-series argument. Although $H_0$ need not be geometrically local, it is diagonal, while each term of $-h\sum_iX_i$ flips only one spin. Therefore, setting $m=\lceil dn\rceil$,
    $\Pi_AH^j\Pi_C=0$ for every integer $j<m$.
    For $s=\Delta\beta/2$, the triangle inequality, the tail bound, and $m!\geq(m/e)^m$ give
    \begin{equation}
        \begin{split}
            \norm{e^{sH}\Pi_Ae^{-sH}\Pi_C}
            &\leq e^{s\norm{H}}
                       \sum_{j\geq m}\frac{(s\norm{H})^j}{j!}\\
            &\leq e^{2s\norm{H}}
                       \left(\frac{es\norm{H}}{m}\right)^m\\
            &\leq\exp\left[n\left(
                \Delta\beta(1+h)
                -d\log\frac{2d}{e\Delta\beta(1+h)}
                          \right)\right]
        \end{split}
    \end{equation}
    whenever $\Delta$ is sufficiently small. The coefficient of $n$ in the last exponent tends to $-\infty$ as $\Delta\downarrow0$. Choose a constant $\Delta>0$, with $1/\Delta$ an integer, so that this norm is at most $e^{-(\kappa+1)n}$.

    Apply the stoquastic decoupling estimate of Lemma~\ref{lem:sigma0}, proved below. The lower bound on $w_A$ and the ratio estimate above imply
    \begin{equation}
        \begin{split}
            \Tr[\Pi_C\sigma_0]
            &\leq\frac1\Delta\left[
                \max\{1,w_C/w_A\}
                \norm{e^{sH}\Pi_Ae^{-sH}\Pi_C}
                +\frac{w_B}{w_A}\right]\\
            &\leq O(1)\left[e^{-n}+e^{-(a-2\kappa)n}\right]
             =e^{-\Omega(n)}.
        \end{split}
    \end{equation}
    This proves the stated weight and decoupling estimates.
\end{proof}

We use the above lemmas to show that if $H_0$ has a linear free energy barrier --- \eqref{eq:lineargap} in the proposition below --- and two regions $A$ and $C$ separated with linear Hamming distance, then slow mixing in the classical model implies slow mixing in the quantum model with a (small) constant transverse field.

\begin{corollary}\label{cor:linear}
     Let $\cT$ denote an $\ell$-local quantum Gibbs sampler with respect to $\rho_\beta$, where $\ell < d n$. Then $\tmix(\cT) = \exp[\Omega(n)]$.
\end{corollary}
\begin{proof}
    We apply Proposition~\ref{prop:mix_time_discrete} under the observation that $\cT^\dagger(\Pi_C) \preccurlyeq \Pi_B + \Pi_C$ due to the locality of the Gibbs sampler being smaller than the Hamming distance between $A$ and $C$.
\end{proof}

Similarly to the previous section, the same result holds for a continuous-time Lindbladian Gibbs sampler $\cL$ instead of the discrete Gibbs sampler $\cT$. Note that the above argument requires the Lindbladian to have locality smaller than the Hamming distance between regions $A$ and $C$. The locality of a Lindbladian Gibbs sampler is a natural property: for example, given a geometrically local Hamiltonian defined on a lattice, the Lindbladian of~\cite{chen2023efficient,gilyen2024quantum} are quasi-local. In practice, locality may be desirable for efficiency of implementation, making Corollary~\ref{cor:linear} relevant.

We remark that this proof only holds for a linear free energy barrier due to the size of $\beta|h|\norm{V} = \Theta(n)$ for constant $\beta, h$. Since $h$ can always be chosen such that $\beta|h|\norm{V}$ is smaller than a free energy barrier of size $\Theta(n)$, the bottleneck lemma can be successfully applied. If instead the free energy barrier were $o(n)$, however, then for constant $h$, the mixing time lower bound would vanish for asymptotically large $n$.

\subsection{Non-commuting Hamiltonians with a sublinear free energy barrier}
\label{sec:localtfim}

Here, we will develop techniques to establish clustering and lower-bound the mixing time of (stoquastic) non-commuting models with a \emph{sub}linear free energy barrier, which prevents standard matrix inequalities (such as those applied in the previous subsection) from showing slow mixing. As an example, we prove slow mixing of the 2D transverse field Ising model on the $\sqrt n \times \sqrt n$ lattice
\begin{equation}
\label{eq:tfim}
    H = -\sum_{\langle i, j \rangle} Z_i Z_j - h\sum_{i=1}^n X_i,
\end{equation}
where $\langle i, j \rangle$ denotes unique pairs (e.g., taking $i < j$) of neighboring vertices on a 2D lattice of dimension $\sqrt n \times \sqrt n$ with free boundary conditions, and we take $h \geq 0$.
Like the purely classical model, we will show a mixing time lower bound of $\exp(\Omega(\sqrt{n}))$ below a critical constant temperature and transverse field strength, proving Theorem~\ref{thm:tfimslow} and Corollary~\ref{cor:tfimslowchen}. 

To use the slow mixing result of \Cref{lem:slow}, we need to establish strong clustering for clusters $A, C$ and bottleneck $B = (A \cup C)^c$. The support of the bottleneck region $B$ must be superpolynomially suppressed in the Gibbs state, and the support of regions $A$ and $C$ must be large. Regions $A$ and $C$ must be separated by a distance sufficiently large such that Lindbladian operators acting on $A$ cannot reach $C$ and instead fall into the bottleneck $B$. Here, the notion of distance will be provided by the number of local (or quasi-local) unitary operations.

Technically, the construction of the region $B$ contains two main components. First, we write the transverse field Ising model in the path integral language of the Feynman-Kac representation introduced in Lemma~\ref{lem:fk}, which interprets the transverse field as a probabilistic point process and enables us to evaluate matrix elements in the computational basis instead of the eigenbasis of the Hamiltonian. Second, we consider the metastable domains introduced by fault lines, which are sufficient to show slow mixing in the classical 2D Ising model. Here, we show fault-line configurations are also exponentially suppressed in the presence of a transverse field via a Peierls-like argument, even though the configurations are addressed by the path integral given by Feynman-Kac.

Spin flip symmetry gives equal Gibbs weight to regions $A$ and $C$, completing the weak-clustering argument. We then discretize imaginary time to bound $\norm{\Pi_A\sqrt{\rho_\beta}\Pi_C}_2$ and use this estimate to prove strong clustering. Finally, we compare the geometric separation of $A$ and $C$ with the locality of the Gibbs sampler, using Lieb-Robinson bounds for quasi-local samplers.

We now specialize our analysis to the 2D transverse field Ising model to provide a concrete result using the technique of Lemma~\ref{lem:fk}. We follow the methodology outlined in \cite{randall2006slow} (see also \cite{levin2017markov}). The configurations $\sigma$ we wish to consider correspond to metastable configurations, which produce a bottleneck for mixing. In the case of the 2D classical Ising model, these configurations correspond to fault lines, where spins adjacent to one side of the fault line are all up and spins adjacent to the other side are all down, with few exceptions (Fig.~\ref{fig:fault}). 

\begin{definition}[Fault line]
    A \emph{fault line} with at most $k$ defects is a self-avoiding\footnote{The self-avoiding property states that each new element in the sequence for the path is a unique point. Therefore, the path never returns to a previously visited point in the path. Moreover, the requirement of two adjacent faces to each edge prevents edges on the boundary from being in a fault line; this implies that any fault line only intersects the boundary of the lattice at either end.} lattice path from the left side to the right side or from the top to the bottom of a $\sqrt n \times \sqrt n$ lattice, where each edge of the path is adjacent to two faces with different spins on them; with at most $k$ exceptions, the two adjacent faces can be assigned the same spins.
\end{definition}

\begin{figure}
    \centering
    \begin{tikzpicture}[scale=1.2]

\fill[green!40] (1,0) rectangle (2,1);
\fill[green!40] (0,1) rectangle (1,2);
\fill[green!40] (1,1) rectangle (2,2);
\fill[green!40] (2,1) rectangle (3,2);
\fill[green!40] (0,2) rectangle (1,3);
\fill[green!40] (1,2) rectangle (2,3);

\fill[red!40] (0,0) rectangle (1,1);
\fill[red!40] (2,0) rectangle (3,1);
\fill[red!40] (3,0) rectangle (4,1);
\fill[red!40] (3,1) rectangle (4,2);
\fill[red!40] (2,2) rectangle (3,3);
\fill[red!40] (3,2) rectangle (4,3);
\fill[red!40] (0,3) rectangle (1,4);
\fill[red!40] (1,3) rectangle (2,4);
\fill[red!40] (2,3) rectangle (3,4);
\fill[red!40] (3,3) rectangle (4,4);

\draw[very thick] (0,0) grid (4,4);

\draw[black, line width=2 mm] (0,3) -- (2,3) -- (2,2) -- (4,2);

\end{tikzpicture}
    \caption{Fault line of length $\ell = 5$ with a single defect ($k=1$). Green squares indicate $\sigma_i = +1$ and red squares indicate $\sigma_i = -1$. All pairs of adjacent spins across the fault line have opposite signs, except at the location of the defect (at the right end of the fault line).}
    \label{fig:fault}
\end{figure}
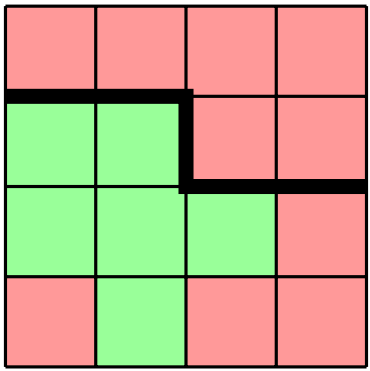

We will soon show that configurations with fault lines have small support in the Gibbs state; such configurations will form our bottleneck region $B$ and the remaining configurations will be split symmetrically across two regions $A$ and $C$. We define the regions $A, B, C$ formally here.

\begin{definition}
    Let $S_+$ be the set of configurations $\sigma \in \{\pm 1\}^n$ that have a top-to-bottom and a left-to-right crossing of pluses\footnote{A ``crossing" refers to a contiguous region of spins that spans from one border to another border of the 2D lattice.}, and similarly define $S_-$. Let $A \subset S_+$ consist of configurations $\sigma \in S_+$ such that every fault line in $\sigma$ has at least $c_0\sqrt n$ defects for some fixed $c_0>0$. Analogously define $C \subset S_-$, and let $B$ consist of all remaining configurations in $\{\pm 1\}^n \setminus (A \cup C)$. We refer to $(A, B, C)$ as the $c_0$-regions of the Ising model.
\end{definition}

We proceed to show that configurations with few-defect fault lines produce a bottleneck with inverse superpolynomial support on the Gibbs state. We will use a Peierls-like argument that constructs a one-to-one map between configurations with fault lines to configurations with exponentially larger support in the Gibbs state. Specifically, we consider a configuration with a fault line, and the same configuration but with all spins on one side of the fault line flipped. When applying the Feynman-Kac representation in this setting, the following result will be useful; it relates the two configurations for a given sample of the spin-flip process $s(t)$ defined in Lemma~\ref{lem:fk}.

\begin{lemma}\label{lem:midratio}
    Let $\sigma$ be a configuration with a fault line of length $\ell$ and $k$ defects. Let $S$ denote the set of spins on one side of the fault line; let $\partial S$ denote the set of spins adjacent to the fault line. Let $\sigma'$ be a configuration with all spins in $S$ flipped, i.e., $\sigma'_i = (-1)^{1_{i \in S}}\sigma_i$. For Feynman-Kac spin-flip process $s(t)$, let $m_{\partial S}(s(t))$ denote the number of spins in $\partial S$ that are flipped in any time $t \in [0, 1]$; i.e.,
    \begin{equation}
        m_{\partial S}(s(t)) = \left|\{i:i \in \partial S \; \mathrm{and} \; \exists \; t_* \in [0, 1] \; \mathrm{s.t.} \; s_i(t_*) = -1\}\right|.
    \end{equation}
    The following holds for any $s(t)$:
    \begin{equation}
        \exp\left[2\beta(\ell-2k-6m_{\partial S}(s(t)))\right] \leq \frac{\exp\left[\beta\int_0^1 dt \sum_{\langle i,j \rangle} \sigma_i' s_i(t) \sigma_j' s_j(t)\right]}{\exp\left[\beta\int_0^1 dt \sum_{\langle i,j \rangle} \sigma_i s_i(t) \sigma_j s_j(t)\right]} \leq \exp\left[2\beta(\ell-2k+6m_{\partial S}(s(t)))\right].
    \end{equation}
\end{lemma}
\begin{proof}
    Let $\partial S$ denote the set of spins on either side of the fault line.
    Divide the set of neighbors on the lattice into $\langle i, j \rangle = N_{\partial S} \cup N_{\partial S}^c$, where $N_{\partial S}$ is the set of $\langle i, j \rangle$ such that both $i, j \in \partial S$ and are on opposite sides of the fault line, and $N_{\partial S}^c$ is the remaining neighbors on the lattice. Note that $|N_{\partial S}| = \ell$ since each edge in the fault line has neighboring lattice sites on either side. However, the number of unique spins in $N_{\partial S}$ may be less than $2\ell$ due to bends in the fault line; we write the number of spins in $\partial S$ as $|\partial S|$, which satisfies $|\partial S| \leq 2\ell$. Finally, note that a spin $i$ may appear in both $N_{\partial S}$ (due to an interaction between two spins in $\partial S$) and $N_{\partial S}^c$ (due to an interaction between one spin in $\partial S$ and one spin outside $\partial S$). Since $\sigma_i' \sigma_j' = \sigma_i \sigma_j$ for $(i, j) \in N_{\partial S}^c$ and $\sigma_i' \sigma_j' = -\sigma_i \sigma_j$ for $(i,j) \in N_{\partial S}$, we can rewrite the above ratio as
    \begin{equation}
    \begin{split}
        R(s(t)) &= \frac{\exp\left[\beta \int_0^1 dt \lr{\sum_{(i, j) \in N_{\partial S}^c} \sigma_i s_i(t) \sigma_j s_j(t) - \sum_{(i, j) \in N_{\partial S}} \sigma_i s_i(t) \sigma_j s_j(t)}\right]}{\exp\left[\beta \int_0^1 dt \lr{\sum_{(i, j) \in N_{\partial S}^c} \sigma_i s_i(t) \sigma_j s_j(t) + \sum_{(i, j) \in N_{\partial S}} \sigma_i s_i(t) \sigma_j s_j(t)}\right]} \\
        &= \frac{\exp\left[\beta \int_0^1 dt \lr{-\sum_{(i, j) \in N_{\partial S}} \sigma_i s_i(t) \sigma_j s_j(t)}\right]}{\exp\left[\beta \int_0^1 dt \lr{\sum_{(i, j) \in N_{\partial S}} \sigma_i s_i(t) \sigma_j s_j(t)}\right]}.
    \end{split}
    \end{equation}
    Since $\sigma$ has a fault line of length $\ell$ with $k$ defects, we have that $\sum_{(i, j) \in N_{\partial S}} \sigma_i \sigma_j = -\ell + 2k$. Flipping a spin can change the energy by at most 6, since a spin can be surrounded by 3 edges of the fault line. The energy is then bounded by $-\ell+2k - 6m_{\partial S}(s(t)) \leq \sum_{(i, j) \in N_{\partial S}} \sigma_i s_i(t) \sigma_j s_j(t) \leq -\ell + 2k + 6m_{\partial S}(s(t))$. This gives bounds
    \begin{equation}
        \exp\left[2\beta(\ell-2k-6m_{\partial S}(s(t)))\right] \leq R(s(t)) \leq \exp\left[2\beta(\ell-2k+6m_{\partial S}(s(t)))\right].
    \end{equation}
\end{proof}

We use another fact, proven in Appendix~\ref{app:localtfim} but stated here. In what follows, the distribution $\cD s(t)$ denotes the measure over Poisson processes $s_1(t), \dots, s_n(t)$.

\begin{lemma}\label{lem:fm}
    The quantity
    \begin{equation}
        f_m = \mathbb E_{\cD s(t)}\left[
        \exp\left[\beta\int_0^1 dt\sum_{\langle i,j\rangle}
                    \sigma_i s_i(t)\sigma_j s_j(t)\right]
        \,\middle|\,s(1)=1,\ m_{\partial S}(s(t))=m\right]
    \end{equation}
    satisfies
    \begin{equation}
        e^{-8\beta m} f \leq f_m \leq e^{8\beta m} f
    \end{equation}
    for some $f$ that is independent of $m$.
\end{lemma}
\begin{proof}
    See Appendix~\ref{app:localtfim}.
\end{proof}

The proof is conceptually straightforward: the quantity $f$ corresponds to averaging over spin-flip processes that act as specified by the Feynman-Kac representation, but with the substitution of no flips occurring on any spins in $\partial S$. Since the integral conditions on $m$ spins flipping in $\partial S$, this substitution incurs errors of at most $e^{\pm O(\beta m)}$, producing the claimed result. We are now prepared to compare the support of $\sigma'$ to $\sigma$ in the Gibbs state.

\begin{lemma}[Fault line support, single configuration]
\label{lem:ratiosupport}
    Let $\sigma$ be a configuration with a fault line of length $\ell$ and $k$ defects. Let $S$ denote the set of spins on one side of the fault line. Let $\sigma'$ be a configuration with all spins in $S$ flipped, i.e., $\sigma'_i = (-1)^{1_{i \in S}}\sigma_i$. Then for $H$ defined in~\eqref{eq:tfim},
    \begin{equation}
        \exp\left[2\beta\lr{\ell\lr{1 + g_{\beta,h}^{-}}-2k}\right] \leq \frac{\bra{\sigma'}e^{-\beta H}\ket{\sigma'}}{\bra{\sigma}e^{-\beta H}\ket{\sigma}} \leq \exp\left[2\beta\lr{\ell\lr{1 + g_{\beta,h}^{+}}-2k}\right].
    \end{equation}
    for
    \begin{equation}
        \begin{split}
        g^-_{\beta,h}
        &=\frac1\beta\log\frac{1+e^{-20\beta}(\cosh(\beta h)-1)}
                                  {1+e^{8\beta}(\cosh(\beta h)-1)},\\
        g^+_{\beta,h}
        &=\frac1\beta\log\frac{1+e^{20\beta}(\cosh(\beta h)-1)}
                                  {1+e^{-8\beta}(\cosh(\beta h)-1)}.
        \end{split}
    \end{equation}
\end{lemma}
\begin{proof}
    We apply Lemma~\ref{lem:fk}, i.e., the identity
    \begin{equation}
        \bra{\sigma}e^{-\beta H}\ket{\sigma} = \lr{\cosh(\beta h)}^n \left\langle \exp\left[\beta\int_0^1 \sum_{\langle i,j \rangle}^n \sigma_i s_i(t) \sigma_j s_j(t) dt\right] \right\rangle_{\beta h},
    \end{equation}
    where $\langle \cdot \rangle_{\beta h}$ denotes averaging over Poisson point processes conditioned on $s_i(1) = 1$ for all $i \in [n]$. We note that since the choice of $s(t)$ is independent of whether the configuration is $\sigma$ or $\sigma'$, we can expand the expectation as
    \begin{equation}
        \frac{\bra{\sigma'}e^{-\beta H}\ket{\sigma'}}{\bra{\sigma}e^{-\beta H}\ket{\sigma}} = \frac{\int_{s(t)} \cD s(t)\,\pr{s(t) \, | \, s(1) = 1} \exp\left[\beta \int_0^1 dt \sum_{\langle i, j \rangle} \sigma_i' s_i(t) \sigma_j' s_j(t)\right]}{\int_{s(t)} \cD s(t)\,\pr{s(t) \, | \, s(1) = 1} \exp\left[\beta \int_0^1 dt \sum_{\langle i, j \rangle} \sigma_i s_i(t) \sigma_j s_j(t)\right]}.
    \end{equation}
    Letting
    \begin{equation}
        R(s(t)) = \frac{\exp\left[\beta\int_0^1 dt \sum_{\langle i,j \rangle} \sigma_i' s_i(t) \sigma_j' s_j(t)\right]}{\exp\left[\beta\int_0^1 dt \sum_{\langle i,j \rangle} \sigma_i s_i(t) \sigma_j s_j(t)\right]},
    \end{equation}
    this is equivalent to
    \begin{equation}
         \frac{\bra{\sigma'}e^{-\beta H}\ket{\sigma'}}{\bra{\sigma}e^{-\beta H}\ket{\sigma}} = \frac{\int_{s(t)} \cD s(t)\,R(s(t)) \pr{s(t) \, | \, s(1) = 1} \exp\left[\beta \int_0^1 dt \sum_{\langle i, j \rangle} \sigma_i s_i(t) \sigma_j s_j(t)\right]}{\int_{s(t)} \cD s(t)\,\pr{s(t) \, | \, s(1) = 1} \exp\left[\beta \int_0^1 dt \sum_{\langle i, j \rangle} \sigma_i s_i(t) \sigma_j s_j(t)\right]}.
    \end{equation}
    Applying \Cref{lem:midratio} to bound $R(s)$, we obtain lower bound
    \begin{equation}
    \label{eq:lbmid}
    \begin{split}
        &\frac{\int_{s(t)} \cD s(t)\,\exp\left[2\beta(\ell-2k-6 m_{\partial S}(s(t)))\right] \pr{s(t)\,|\,s(1)=1} \exp\left[\beta\int_0^1 dt \sum_{\langle i,j \rangle} \sigma_i s_i(t) \sigma_j s_j(t)\right]}{\int_{s(t)} \cD s(t)\,\pr{s(t)\,|\,s(1)=1} \exp\left[\beta\int_0^1 dt \sum_{\langle i,j \rangle} \sigma_i s_i(t) \sigma_j s_j(t)\right]} \\
        &\leq \frac{\bra{\sigma'}e^{-\beta H}\ket{\sigma'}}{\bra{\sigma}e^{-\beta H}\ket{\sigma}}
    \end{split}
    \end{equation}
    and analogous upper bound
    \begin{equation}
    \begin{split}
        &\frac{\bra{\sigma'}e^{-\beta H}\ket{\sigma'}}{\bra{\sigma}e^{-\beta H}\ket{\sigma}} \\
        &\leq \frac{\int_{s(t)} \cD s(t)\,\exp\left[2\beta(\ell-2k+6 m_{\partial S}(s(t)))\right] \pr{s(t)\,|\,s(1)=1} \exp\left[\beta\int_0^1 dt \sum_{\langle i,j \rangle} \sigma_i s_i(t) \sigma_j s_j(t)\right]}{\int_{s(t)} \cD s(t)\,\pr{s(t)\,|\,s(1)=1} \exp\left[\beta\int_0^1 dt \sum_{\langle i,j \rangle} \sigma_i s_i(t) \sigma_j s_j(t)\right]}.
    \end{split}
    \end{equation}
    The remaining argument will be symmetric in upper and lower bounds, so we only explicitly compute the lower bound. Observe that we can group time-paths $s(t)$ according to $m_{\partial S}(s(t))$, rewriting the LHS of~\eqref{eq:lbmid} as
    \begin{equation}
        \frac{\sum_{m=0}^{|\partial S|} \exp\left[2\beta(\ell-2k-6 m)\right] \pr{m_{\partial S}(s(t)) = m \,|\,s(1)=1} f_m}{\sum_{m=0}^{|\partial S|} \pr{m_{\partial S}(s(t)) = m\,|\,s(1)=1} f_m} \leq \frac{\bra{\sigma'}e^{-\beta H}\ket{\sigma'}}{\bra{\sigma}e^{-\beta H}\ket{\sigma}},
    \end{equation}
    where $f_m$ is the conditional expectation in \Cref{lem:fm}.
    Let $p_m=\pr{m_{\partial S}(s(t))=m\,|\,s(1)=1}$.
    Applying the lower bound on $f_m$ in the numerator and the upper bound in the denominator gives
    \begin{equation}
        e^{2\beta(\ell-2k)}
        \frac{\sum_{m=0}^{|\partial S|}p_m e^{-20\beta m}}
             {\sum_{m=0}^{|\partial S|}p_m e^{8\beta m}}
        \leq
        \frac{\bra{\sigma'}e^{-\beta H}\ket{\sigma'}}
             {\bra{\sigma}e^{-\beta H}\ket{\sigma}}.
    \end{equation}
    Similarly, the upper bound is
    \begin{equation}
        \frac{\bra{\sigma'}e^{-\beta H}\ket{\sigma'}}
             {\bra{\sigma}e^{-\beta H}\ket{\sigma}}
        \leq e^{2\beta(\ell-2k)}
        \frac{\sum_{m=0}^{|\partial S|}p_m e^{20\beta m}}
             {\sum_{m=0}^{|\partial S|}p_m e^{-8\beta m}}.
    \end{equation}
    We now evaluate $p_m$. For one spin, conditioned on an even number of flips, the probability of no flips is $1/\cosh(\beta h)$.
    Independence of the paths therefore gives
    \begin{equation}
        p_m=\binom{|\partial S|}{m}
             \frac{(\cosh(\beta h)-1)^m}{\cosh(\beta h)^{|\partial S|}}.
    \end{equation}
    Hence, for real-valued $u$,
    \begin{equation}
        \sum_m p_m e^{u\beta m}
        =\left[\frac{1+e^{u\beta}(\cosh(\beta h)-1)}
                     {\cosh(\beta h)}\right]^{|\partial S|}.
    \end{equation}
    Substituting these sums gives a lower factor $\exp[\beta g^-_{\beta,h}|\partial S|]$ and an upper factor $\exp[\beta g^+_{\beta,h}|\partial S|]$, multiplied by $e^{2\beta(\ell-2k)}$. Since $g^-_{\beta,h}\leq0$, $g^+_{\beta,h}\geq0$, and $|\partial S|\leq2\ell$, we obtain the bounds in the statement. 
\end{proof}

The above lemma shows, roughly speaking, that for small enough $h$ and large enough $\beta$, every configuration containing a fault line with few defects corresponds to another configuration with $e^{\Theta(\sqrt n)}$ more support in the Gibbs state. The following lemma then sums over configurations with fault lines and few defects to show that these configurations must be exponentially rare in the Gibbs state, allowing us to form a bottleneck.

\begin{lemma}[Fault line support, all configurations]
\label{lem:faultsupport}
    Fix constants $\beta>0$, $h\geq0$, and $0<c_0<1/2$ such that
    \begin{equation}\label{eq:tfim-clustering-region}
        2\beta(1+g^-_{\beta,h})-4\beta c_0>\log3.
    \end{equation}
    Let $\Pi$ denote the projector onto all configurations
    $\sigma\in\{\pm1\}^n$ containing a fault line with at most
    $c_0\sqrt n$ defects. Then for the Gibbs state of
    \eqref{eq:tfim},
    \begin{equation}
        \Tr(\Pi\rho_\beta)\leq\exp[-\Omega(\sqrt n)].
    \end{equation}
\end{lemma}
\begin{proof}
    Fix a self-avoiding path $\pi$ of length $\ell$, and let $\Pi_{\pi,j}$ project onto all configurations having exactly $j$ defects along $\pi$. Flip all spin configurations in $\Pi_{\pi,j}$ on one fixed side and then collect these flipped configurations in $\Pi'_{\pi,j}$. \Cref{lem:ratiosupport} gives
    \begin{equation}
        \Tr(\Pi_{\pi,j}\rho_\beta)
        \leq e^{-2\beta(\ell(1+g^-_{\beta,h})-2j)}
                    \Tr(\Pi'_{\pi,j}\rho_\beta)
        \leq e^{-2\beta(\ell(1+g^-_{\beta,h})-2j)}.
    \end{equation}
    Let $\alpha_\ell$ be the number of self-avoiding lattice paths of
    length $\ell$ starting from the origin. There are at most
    $2(\sqrt n+1)\alpha_\ell$ spanning paths, and
    $\alpha_\ell\leq4\cdot3^{\ell-1}\leq4\cdot3^\ell$.
    Taking $k=\lfloor c_0\sqrt n\rfloor$, a union bound gives
    \begin{equation}
        \Tr(\Pi\rho_\beta)
        \leq8(\sqrt n+1)
        \sum_{\ell\geq\sqrt n}
            \left[3e^{-2\beta(1+g^-_{\beta,h})}\right]^\ell
        \sum_{j=0}^k e^{4\beta j}.
    \end{equation}
    Both sums are geometric. In particular,
    \begin{equation}
    \begin{split}
        \Tr(\Pi\rho_\beta)
        &\leq
        \frac{8(\sqrt n+1)}
        {(1-3e^{-2\beta(1+g^-_{\beta,h})})(1-e^{-4\beta})}\\
        &\qquad\times
        \exp\left[\sqrt n\bigl(\log3-2\beta(1+g^-_{\beta,h})
                                    +4\beta c_0\bigr)\right].
    \end{split}
    \end{equation}
    The inequality in \eqref{eq:tfim-clustering-region} makes the exponent equal to $-\Omega(\sqrt n)$.
\end{proof}

So far, we have shown that configurations with fault lines form a bottleneck; in the language of the slow mixing result (\Cref{lem:slow}), this is the region $B$. We now must show that at constant temperature, there exist two regions $A$ and $C$ with large support in the thermal state such that $A$ and $C$ are further apart than the range of the Lindbladian. Moreover, we must show that traveling between $A$ and $C$ must pass through the bottleneck.

We begin by constructing regions $A$ and $C$ from various paths on the 2D model in the computational basis. First, we require some (slight generalization of) standard properties of these paths~\cite{randall2006slow,levin2017markov}. Recall that a lattice path refers to a collection of edges, where each edge is adjacent to a site on either side of the fault line (as in Fig.~\ref{fig:fault}).

\begin{lemma}[Properties of paths (similar to Lemma 15.17~\cite{levin2017markov})]
\label{lem:pathprops}
    For a 2D $\sqrt{n} \times \sqrt{n}$ lattice with configuration $\sigma \in \{\pm 1\}^n$, the following properties hold.
    \begin{enumerate}
        \item If in $\sigma$ there is neither an all-plus nor an all-minus crossing from the left to the right side of the lattice, then there is a fault line with no defects from the top to the bottom.
        \item Let $\Gamma_+$ be a path of lattice sites such that $\sigma_i = 1$ for all $i \in \Gamma_+$, and $\Gamma_+$ starts at site $q$ and ends at a site at the top of the lattice. Similarly define $\Gamma_-$ from $q'$ to the top of the lattice, where $\sigma_i = -1$ for all $i \in \Gamma_-$. Assume $q$ and $q'$ are on the same row of the lattice, and let $\Gamma_{qq'}$ be the horizontal path of sites directly joining $q$ and $q'$. Then there exists a lattice path $\xi$ from the boundary of $\Gamma_{qq'}$ to the top of the lattice such that every edge in $\xi$ is adjacent to two lattice sites with different labels in $\sigma$.
    \end{enumerate}
\end{lemma}
\begin{proof}
    See Appendix~\ref{app:localtfim}.
\end{proof}

\begin{corollary}
\label{cor:defect3}
    Let $(A, B, C)$ be the $c_0$-regions of the $\sqrt n \times \sqrt n$ Ising model. Every $\sigma \in A$ and $\sigma' \in C$ satisfies
    \begin{equation}
        |\sigma-\sigma'|_H \geq \frac23 (c_0\sqrt n-3).
    \end{equation}
    In particular, $|\sigma-\sigma'|_H\geq c_0\sqrt n/2$ whenever $c_0\sqrt n\geq12$.
\end{corollary}
\begin{proof}
    See Appendix~\ref{app:localtfim}.
\end{proof}

We now show weak clustering. The following argument is similar to that of~\cite{randall2006slow}, but we use the Feynman-Kac results above to control the effect of the transverse field. 

\begin{lemma}[Weak clustering of the 2D transverse field Ising model]
\label{lem:tfimbottleneck}
    Fix $\beta,h,c_0$ satisfying \eqref{eq:tfim-clustering-region} from \Cref{lem:faultsupport}.
    Let $\Pi_A,\Pi_C$ be the projectors onto the $c_0$-regions $A,C$, and let $\Pi_B=I-\Pi_A-\Pi_C$. Then
    \begin{equation}
    \begin{split}
        \Tr(\Pi_B\rho_\beta)&\leq e^{-\Omega(\sqrt n)},\\
        \Tr(\Pi_A\rho_\beta)=\Tr(\Pi_C\rho_\beta)
        &=\frac{1-\Tr(\Pi_B\rho_\beta)}2,
        \qquad d_H(A,C)\geq c_0\sqrt n/2
    \end{split}
    \end{equation}
    for all sufficiently large $n$. In particular, $\rho_\beta$ is weakly clustered with clusters $\Pi_A,\Pi_C$ and bottleneck $\Pi_\perp=\Pi_B$.
    Moreover, for every operator $O$ supported in an $R\times R$ square of qubits, where $R\leq c_0\sqrt n/8$, we have
    \begin{equation}
        \Pi_AO\Pi_C=0.
    \end{equation}
\end{lemma}
\begin{proof}
    Let $(A,B,C)$ be the $c_0$-regions of the Ising model.
    By the first property of Lemma~\ref{lem:pathprops}, every configuration outside $S_+\cup S_-$ has a defect-free fault line.
    Configurations removed from $S_+$ or $S_-$ when defining $A,C$ have a fault line with fewer than $c_0\sqrt n$ defects.
    Consequently, Lemma~\ref{lem:faultsupport} gives $\Tr(\Pi_B\rho_\beta)\leq e^{-\Omega(\sqrt n)}$.
    Conjugation by $X^{\otimes n}$ fixes $H$ and exchanges $A,C$.
    Therefore their Gibbs weights are equal and sum to $1-\Tr(\Pi_B\rho_\beta)$.
    Corollary~\ref{cor:defect3} gives quantum Hamming separation $d_H(A,C)\geq c_0\sqrt n/2$, completing the weak-clustering claim.

    It remains to establish $\Pi_AO\Pi_C=0$.
    Suppose configurations $\sigma\in A$ and $\sigma'\in C$ agree outside an $R\times R$ square. Change its spins one at a time to go from $\sigma$ to $\sigma'$. 
    The single-spin crossing argument in the proof of Corollary~\ref{cor:defect3} gives an intermediate configuration with a spanning fault line having at most three defects. 
    Outside the square this configuration agrees with $\sigma$.
    Replace the part of this fault line between its first and last encounters with the square by an arc along the perimeter of the square. This arc uses at most $4R$ edges.
    If the square meets the outer boundary, use the perimeter inside the lattice to build this arc and, when necessary, move the corresponding endpoint along the same outer side. 
    Since $R<\sqrt n$, the endpoints remain on opposite sides.
    The resulting fault line in $\sigma$ has at most $4R+3$ defects. But $4R+3\leq7R\leq7c_0\sqrt n/8<c_0\sqrt n$, contradicting $\sigma\in A$.
    Thus no such pair of configurations exists, and every operator supported in the square has $\Pi_AO\Pi_C=0$.
\end{proof}

Finally, to establish strong clustering and apply the bottleneck lemma, we must satisfy the decoupling condition
\begin{equation}
\label{eq:cacond2}
    \Tr[\Pi_C \frac{e^{-\beta H/2} \Pi_A e^{-\beta H/2}}{\Tr[\Pi_A e^{-\beta H}]}] = o(1).
\end{equation}
This is done by discretizing the Poisson point process of the Feynman-Kac representation over imaginary time. For $\Delta=1/\poly{n}$, we control the total contribution from transitions between $A$ and $C$ during one interval of imaginary time $\Delta\beta/2$. Iteratively inserting the cluster projectors shows that a contribution connecting $A$ and $C$ either enters the bottleneck or includes one such short-time transition.

To show \eqref{eq:cacond2}, we first provide some generic lemmas for stoquastic systems composed from classical Hamiltonians in a transverse field. The probabilistic interpretation of Lemma~\ref{lem:fk} implies that $e^{-\beta H}$ has matrix elements that are only nonnegative (in the computational basis), which aids in showing \eqref{eq:cacond2}. In particular, the two following facts hold (proofs deferred to Appendix~\ref{app:localtfim}).

\begin{fact}
\label{fact:pos}
    For arbitrary matrices $A, B$ with only nonnegative entries and for diagonal projector $\Pi$,
    \begin{equation}
        \tr(A \Pi B) \leq \tr(A B).
    \end{equation}
\end{fact}
\begin{fact}
\label{fact:pow}
    Let $A$ be a PSD matrix such that $A^p$ only has nonnegative entries for all $p \geq 0$, and let $\Pi$ be a diagonal projector. Then, for any $\gamma \geq 1$,
    \begin{equation}
    \lr{\Tr\lr{(A \Pi A)^\gamma}}^{1/\gamma} \leq \lr{\Tr(\Pi A^{2\gamma})}^{1/\gamma}.
    \end{equation}
\end{fact}

\begin{lemma}[Decoupling between regions]
\label{lem:sigma0}
    Assume for all $\beta > 0$ that $e^{-\beta H}$ has only nonnegative matrix elements and that $\Pi_A, \Pi_B, \Pi_C$ are orthogonal projectors that are diagonal in the computational basis and sum to $I$, with $\Tr[\Pi_Ae^{-\beta H}]>0$. Then for any integer $1/\Delta \geq 1$,
    \begin{equation}
        \Tr[\Pi_C \frac{e^{-\beta H/2} \Pi_A e^{-\beta H/2}}{\Tr[\Pi_A e^{-\beta H}]}] \leq \frac{1}{\Delta}\lr{
        \max\left\{1,\frac{\Tr[\Pi_Ce^{-\beta H}]}{\Tr[\Pi_Ae^{-\beta H}]}\right\}
        \norm{e^{\Delta\beta H/2}\Pi_Ae^{-\Delta\beta H/2}\Pi_C}
        +\frac{\Tr[\Pi_Be^{-\beta H}]}{\Tr[\Pi_Ae^{-\beta H}]}}.
    \end{equation}
\end{lemma}
\begin{proof}
    Define $T(a) = e^{-a \beta H/2}$. We will prove an inequality of the form
    \begin{equation}
        \Tr[\lr{\Pi_A T(\Delta)}^m \lr{\Pi_B + \Pi_C} T(1-m\Delta)\Pi_C T(1)] \leq f(m).
    \end{equation}
    This gives
    \begin{equation}
    \begin{split}
        \Tr[\Pi_C e^{-\beta H/2} \Pi_A e^{-\beta H/2}] &= \Tr[\Pi_A T(1) \Pi_C T(1)]\\
        &= \Tr[\Pi_A T(\Delta) \lr{\Pi_A + \Pi_B + \Pi_C} T(1-\Delta) \Pi_C T(1)]\\
        &\leq \Tr[\Pi_A T(\Delta) \Pi_A T(1-\Delta)\Pi_C T(1)] + f(1)\\
        &= \Tr[\lr{\Pi_A T(\Delta)}^2 (\Pi_A + \Pi_B + \Pi_C) T(1-2\Delta) \Pi_C T(1)] + f(1)\\
        &\leq \Tr[\lr{\Pi_A T(\Delta)}^2 \Pi_A T(1-2\Delta) \Pi_C T(1)] + f(2) + f(1)\\
        &\leq \sum_{m=1}^{1/\Delta} f(m),
    \end{split}
    \end{equation}
    where the last inequality is obtained by iteratively inserting the identity as in the prior lines; the final remaining term contains $\Pi_A\Pi_C=0$. To show the bound $f(m)$, we separately bound the two terms $\Pi_B$ and $\Pi_C$ with $f(m) = \xi_B(m) + \xi_C(m)$. For $\Pi_B$, Fact~\ref{fact:pos} implies
    \begin{equation}
    \begin{split}
        \xi_B(m) &= \Tr[\lr{\Pi_A T(\Delta)}^m \Pi_B T(1-m\Delta)\Pi_C T(1)]\\
        &\leq \Tr[T(m\Delta) \Pi_B T(2-m\Delta)]\\
        &= \Tr[\Pi_B T(2)].
    \end{split}
    \end{equation}
    For $\Pi_C$ and $m\Delta<1$, we have
    \begin{equation}
    \begin{split}
        \xi_C(m) &= \Tr[\lr{\Pi_A T(\Delta)}^m \Pi_C T(1-m\Delta)\Pi_C T(1)]\\
        &\leq \Tr[\Pi_A T((m-1)\Delta) \Pi_A T(\Delta) \Pi_C T(1-m\Delta) \Pi_C T(1)]\\
        &= \Tr[T(m\Delta) \Pi_A T(m\Delta) T(-\Delta) \Pi_A T(\Delta) \Pi_C T(1-m\Delta) \Pi_C T(1-m\Delta)]\\
        &\leq \norm{T(-\Delta) \Pi_A T(\Delta) \Pi_C} \Tr[\lr{T(m\Delta) \Pi_A T(m\Delta)}^{\frac{1}{m\Delta}}]^{m\Delta}\\
        &\qquad\times\Tr[\lr{T(1-m\Delta) \Pi_C T(1-m\Delta)}^{\frac{1}{1-m\Delta}}]^{1-m\Delta}\\
        &\leq \norm{T(-\Delta) \Pi_A T(\Delta) \Pi_C} \Tr[\Pi_A T(2)]^{m\Delta} \Tr[\Pi_C T(2)]^{1-m\Delta}\\
        &\leq \norm{T(-\Delta)\Pi_A T(\Delta)\Pi_C}
        \max\{\Tr[\Pi_A T(2)],\Tr[\Pi_C T(2)]\},
    \end{split}
    \end{equation}
    where the first inequality holds by Fact~\ref{fact:pos}, the second inequality holds by H\"older's inequality, and the third inequality holds by Fact~\ref{fact:pow}. 
    Hence, denoting $w_X=\Tr[\Pi_Xe^{-\beta H}]$, we have
    \begin{equation}
    \begin{split}
        \Tr\left[\Pi_C
        \frac{e^{-\beta H/2}\Pi_Ae^{-\beta H/2}}{w_A}\right]
        &\leq\frac1{w_A}\sum_{m=1}^{1/\Delta}
                              \bigl(\xi_B(m)+\xi_C(m)\bigr)\\
        &\leq\frac1\Delta\left[
        \max\{1,w_C/w_A\}\norm{T(-\Delta)\Pi_A T(\Delta)\Pi_C}
        +\frac{w_B}{w_A}\right].
    \end{split}
    \end{equation}
    This is the required bound. 
\end{proof}

To apply Lemma~\ref{lem:sigma0}, we bound $\norm{e^{\Delta \beta H / 2}\Pi_A e^{-\Delta\beta H/2}\Pi_C}$.

\begin{lemma}
\label{lem:tfimnorm}
    Let $H$ be the $\sqrt n \times \sqrt n$ 2D transverse field Ising model of \eqref{eq:tfim} with $c_0$-regions $(A,B,C)$. Then for any $\Delta > 0$ such that $\Delta \beta h n \leq 1$, and for sufficiently large $n$ such that $c_0 \sqrt n \geq 12$,
    \begin{equation}
        \norm{e^{\Delta \beta H / 2}\Pi_A e^{-\Delta\beta H/2}\Pi_C} \leq 2e^{2\Delta \beta(1+h)n} \lr{\frac{\Delta \beta h n}{2}}^{c_0\sqrt n / 2}.
    \end{equation}
\end{lemma}
\begin{proof}
    Since $\norm{H} \leq 2n(1+h)$, we have
    \begin{equation}
        \norm{e^{\Delta \beta H / 2}\Pi_A e^{-\Delta\beta H/2}\Pi_C} \leq \norm{e^{\Delta \beta H / 2}}\norm{\Pi_A e^{-\Delta\beta H/2}\Pi_C} \leq e^{\Delta \beta n(1+h)}\norm{\Pi_A e^{-\Delta\beta H/2}\Pi_C}.
    \end{equation}
    To simplify notation, let $M=\Pi_Ae^{-\Delta\beta H/2}\Pi_C$. 
    The matrix $M+M^\dagger$ is real symmetric and entrywise nonnegative, and $\norm{M}=\norm{M+M^\dagger}$.
    Thus Perron--Frobenius gives
    \begin{equation}
    \begin{split}
        \norm{\Pi_Ae^{-\Delta\beta H/2}\Pi_C}
        &\leq\max_\sigma\sum_{\sigma'}
              \bra{\sigma}(M+M^\dagger)\ket{\sigma'}\\
        &\leq\max_\sigma
          \sum_{\sigma':|\sigma-\sigma'|_H\geq c_0\sqrt n/2}
              \bra{\sigma}e^{-\Delta\beta H/2}\ket{\sigma'}.
    \end{split}
    \end{equation}
    The second inequality follows from Corollary~\ref{cor:defect3} for $c_0\sqrt n\geq12$, in both directions between $A$ and $C$. We apply Lemma~\ref{lem:fk} to compute this quantity. Let $\langle\cdot\rangle_{\Delta\beta h/2}$ denote the expectation over independent Poisson processes $s_i(t)$ with rate $\Delta\beta h/2$, conditioning on $\sigma_i s_i(1)=\sigma_i'$ for all $i$.
    We obtain for $\sigma, \sigma'$ satisfying $|\sigma-\sigma'|_H \geq c_0 \sqrt n / 2$ that
    \begin{equation}
    \begin{split}
        \bra{\sigma} e^{-\Delta\beta H/2} \ket{\sigma'} &= \lr{\cosh \frac{\Delta \beta h}{2}}^n \lr{\tanh\frac{\Delta \beta h}{2}}^{|\sigma-\sigma'|_H}\left\langle \exp\left[\frac{\Delta\beta}{2} \int_0^1 \sum_{\langle i,j \rangle} \sigma_i \sigma_j s_i(t) s_j(t) \, \dd t\right] \right\rangle_{\Delta \beta h /2},\\
        &\leq \lr{\cosh \frac{\Delta \beta h}{2}}^n \lr{\tanh\frac{\Delta \beta h}{2}}^{|\sigma-\sigma'|_H} \exp\left[\Delta \beta n\right]\\
        &\leq \lr{\frac{\Delta \beta h}{2}}^{|\sigma-\sigma'|_H}\exp\left[\Delta \beta (1+h) n\right],
    \end{split}
    \end{equation}
    where the last inequality holds for all $\Delta \beta h/2 \leq 1$. Summing over all bitstrings within Hamming distance $d$, this gives
    \begin{equation}
    \begin{split}
        \norm{\Pi_A e^{-\Delta\beta H/2}\Pi_C} &\leq \sum_{d=c_0 \sqrt n / 2}^n \binom{n}{d}\lr{\frac{\Delta \beta h}{2}}^{d}\exp\left[\Delta \beta (1+h) n\right]\\
        &\leq 2\exp\left[\Delta \beta (1+h) n\right] \lr{\frac{\Delta \beta h n}{2}}^{c_0\sqrt n / 2},
    \end{split}
    \end{equation}
    where we apply $\Delta \beta h n \leq 1$ to obtain the last inequality. This gives a final bound of
    \begin{equation}
        \norm{e^{\Delta \beta H / 2}\Pi_A e^{-\Delta\beta H/2}\Pi_C} \leq 2e^{2\Delta \beta(1+h)n} \lr{\frac{\Delta \beta h n}{2}}^{c_0\sqrt n / 2}.
    \end{equation}
\end{proof}

\begin{corollary}[Strong clustering of the 2D transverse field Ising model]
\label{cor:tfimstrong}
    Fix $\beta,h,c_0$ satisfying \eqref{eq:tfim-clustering-region}.
    The Gibbs state $\rho_\beta$ of \eqref{eq:tfim} is strongly clustered with the same clusters $\Pi_A,\Pi_C$ and bottleneck $\Pi_\perp=\Pi_B$ as in Lemma~\ref{lem:tfimbottleneck}.
    In particular, with
    \begin{equation}
        \tilde\rho_\beta
        =\Pi_A\rho_\beta\Pi_A+\Pi_B\rho_\beta\Pi_B
         +\Pi_C\rho_\beta\Pi_C,
    \end{equation}
    we have
    \begin{equation}
        \frac12\norm{\rho_\beta-\tilde\rho_\beta}_1
        \leq\exp[-\Omega(\sqrt n)].
    \end{equation}
\end{corollary}
\begin{proof}
    Weak clustering was proved in Lemma~\ref{lem:tfimbottleneck}.
    Let $b=\Tr[\Pi_B\rho_\beta]$ and $a=\Tr[\Pi_A\rho_\beta]=\Tr[\Pi_C\rho_\beta]=(1-b)/2$.
    First, applying Lemma~\ref{lem:tfimnorm} with $\Delta=1/n^2$ to Lemma~\ref{lem:sigma0} gives, for $\sigma_0=e^{-\beta H/2}\Pi_Ae^{-\beta H/2}/ \Tr[\Pi_Ae^{-\beta H}]$,
    \begin{equation}
    \begin{split}
        q:=\Tr[\Pi_C\sigma_0]
        &\leq n^2\left[
        2e^{2\beta(1+h)/n}
              \left(\frac{\beta h}{2n}\right)^{c_0\sqrt n/2}
        +\frac ba\right]\\
        &\leq\exp[-\Omega(\sqrt n)].
    \end{split}
    \end{equation}
    Here $a=\Omega(1)$ and $b\leq e^{-\Omega(\sqrt n)}$ by Lemma~\ref{lem:tfimbottleneck}. 

    It remains to bound contributions of the off-diagonal blocks of $\rho_\beta$. We have $\norm{\Pi_A \sqrt{\rho_\beta}\Pi_C}_2^2=aq$ and
    \begin{equation}
        \Pi_A\rho_\beta\Pi_C
        =\Pi_A \sqrt{\rho_\beta}\Pi_A \sqrt{\rho_\beta}\Pi_C
         +\Pi_A \sqrt{\rho_\beta}\Pi_B \sqrt{\rho_\beta}\Pi_C
         +\Pi_A \sqrt{\rho_\beta}\Pi_C \sqrt{\rho_\beta}\Pi_C.
    \end{equation}
    Since $\norm{\Pi_A \sqrt{\rho_\beta}\Pi_A}_2,\norm{\Pi_C \sqrt{\rho_\beta}\Pi_C}_2
    \leq\sqrt a$, H\"older's inequality gives
    \begin{equation}
    \begin{split}
        \norm{\Pi_A\rho_\beta\Pi_C}_1
        &\leq2a\sqrt q
          +\norm{\Pi_A \sqrt{\rho_\beta}\Pi_B}_2\norm{\Pi_B \sqrt{\rho_\beta}\Pi_C}_2\\
        &\leq(1-b)\sqrt q+\frac b2.
    \end{split}
    \end{equation}
    The last inequality uses
    \begin{equation}
        \norm{\Pi_A \sqrt{\rho_\beta}\Pi_B}_2^2
        +\norm{\Pi_B \sqrt{\rho_\beta}\Pi_C}_2^2
        =\Tr[\Pi_B \sqrt{\rho_\beta}(\Pi_A+\Pi_C)\sqrt{\rho_\beta}\Pi_B]\leq b.
    \end{equation}
    Another application of H\"older's inequality gives
    $\norm{\Pi_B\rho_\beta(\Pi_A+\Pi_C)}_1
        \leq\sqrt{b(1-b)}$.
    Thus the triangle inequality yields
    \begin{equation}
    \begin{split}
        \frac12\norm{\rho_\beta-\tilde\rho_\beta}_1
        &\leq\norm{\Pi_B\rho_\beta(\Pi_A+\Pi_C)}_1
             +\norm{\Pi_A\rho_\beta\Pi_C}_1\\
        &\leq\sqrt{b(1-b)}+(1-b)\sqrt q+\frac b2\\
        &\leq\exp[-\Omega(\sqrt n)].
    \end{split}
    \end{equation}
    This establishes strong clustering.
\end{proof}


\begin{corollary}[Slow mixing of the 2D transverse field Ising model]
\label{cor:slowtfim}
    Fix $\beta,h,c_0$ satisfying \eqref{eq:tfim-clustering-region}.
    Let $\cL$ be a Lindbladian with fixed point $\cL(\rho_\beta)=0$.
    Suppose every Lindblad operator $L^a$ and coherent term in the Lindbladian is supported in an $R\times R$ square, where $1\leq R\leq c_0\sqrt n/8$. Then
    \begin{equation}
        \tmix(\cL)\geq\exp[\Omega(\sqrt n)].
    \end{equation}
\end{corollary}
\begin{proof}
    Corollary~\ref{cor:tfimstrong} gives strong clustering for $\Pi_A,\Pi_C$ with bottleneck $\Pi_B$. 
    The proof of Lemma~\ref{lem:decouple} gives
    $\Tr[\Pi_C\sigma_0]\leq e^{-\Omega(\sqrt n)}$.
    Lemma~\ref{lem:tfimbottleneck} gives $\Pi_C O\Pi_A=0$ for every local operator in the generator, which is true for each Lindblad operator $(L^a)^\dagger L^a$, which has the same geometric support as $L^a$.
    Thus $\cL^\dagger(\Pi_C)=(\Pi_B+\Pi_C)\cL^\dagger(\Pi_C)(\Pi_B+\Pi_C)$.
    Apply \Cref{lem:slow} to conclude.
\end{proof}


Finally, we use Lieb-Robinson bounds to show that the locality assumption on $R$ is satisfied by Gibbs sampling algorithms such as that of~\cite{chen2023efficient}. As shown in Appendix~\ref{app:localtfim}, Lieb-Robinson bounds give the following property for the Gibbs sampler.

\begin{lemma}[Locality of Gibbs sampler of~\cite{chen2023efficient}]
\label{lem:chen-local}
Let $\cL$ denote the Lindbladian Gibbs sampler of~\cite{chen2023efficient} with $J$ constant-local jump operators and defined at inverse temperature $\beta$ with respect to a constant-local Hamiltonian on an constant-dimensional lattice. There exists a Hermiticity-preserving linear map $\cE$ satisfying $\norm{\cL-\cE}_\diamond\leq\epsilon$ whose signed operator-sum representation uses operators supported on neighborhoods of radius $R$, where
\begin{equation}
    R = O\lr{(\beta+1) \log \frac{J}{\epsilon}}.
\end{equation}
\end{lemma}
\begin{proof}
    See Appendix~\ref{app:localtfim}.
\end{proof}

The proof is omitted from the main text as it consists of standard arguments concerning properties of the CPTP map $e^{\cL t}$ constructed from Lindbladian map $\cL$ (see Appendix~\ref{app:lindblad}) and the following well-known Lieb-Robinson bound (see, e.g.,~\cite{nachtergaele2006lieb,hastings2006spectral,lieb1972finite,hastings2004lieb,haah2021quantum}).

\begin{lemma}[Lieb-Robinson bound]
For any finite-range, local Hamiltonian $H = \sum_Z h_Z$ on a constant-dimensional lattice and with bounded-norm terms, and for any operator $A_X$ acting nontrivially only on region $X$, and for any time $t$ and distance $\ell$, there exists constants $C, v, \mu > 0$ such that
\begin{equation}
    \norm{A_X(t) - A_{X(\ell)}(t)} \leq C |X| \norm{A_X} e^{-\mu \ell} \quad \text{for} \quad X(\ell) = \left\{i : d(i, X) \leq v |t| + \ell \right\}.
\end{equation}
\end{lemma}

The algorithm of~\cite{chen2023efficient} uses linear combinations of time-evolved jump operators (much like the $A_X$ of the Lieb-Robinson bound) convolved with a Gaussian filter whose variance scales with the inverse temperature. At late times (much larger than $\beta$), the decay of the Gaussian filter suppresses the contributions of distant sites, allowing the Lieb-Robinson light cone to be truncated at time $t\sim \beta$ with small error. This gives the $\beta$-dependence in Lemma~\ref{lem:chen-local}. Although we use~\cite{chen2023efficient} as a concrete example in Lemma~\ref{lem:chen-local}, similar algorithms based on a linear combination of Hamiltonian-simulated terms $e^{-iHt}Ae^{iHt}$ produce similarly quasi-local Lindblad operators when placed on a lattice~\cite{haah2021quantum,rouze2024efficient,gilyen2024quantum}.

Finally, we combine Corollary~\ref{cor:slowtfim} and Lemma~\ref{lem:chen-local} to show that this Gibbs sampler mixes slowly for the 2D transverse field Ising model.

\begin{corollary}[Slow mixing of~\cite{chen2023efficient} for 2D transverse field Ising model]
    For $\beta,h,c_0$ satisfying \eqref{eq:tfim-clustering-region}, the mixing time of the Gibbs sampler $\cL$ of~\cite{chen2023efficient} to prepare thermal state $\rho_\beta$ of the 2D transverse field Ising model with constant inverse temperature $\beta > \beta_*$ and constant field strength $0 \leq h < h_*$ satisfies
    \begin{equation}
        \tmix(\cL) \ge \exp[n^{1/2-o(1)}]
    \end{equation}
    for sufficiently large $n$.
    In particular, when $h<e^{-5}$, this mixing time bound holds when $1\leq\beta\leq\frac{1}{5}\log\frac{1}{h}$.
\end{corollary}
\begin{proof}
We first verify the second part of the claim that the parameter condition \eqref{eq:tfim-clustering-region} holds with the specified bounds on $h,\beta$ and $c_0=1/16$. Since $h\leq e^{-5\beta}$, we have $2\beta h e^{4\beta}\leq 2/e<1$, and therefore
\begin{equation}
    e^{8\beta}(\cosh(\beta h)-1)
    \leq e^{8\beta}(\beta h)^2
    \leq \frac14.
\end{equation}
By the definition of $g^-_{\beta,h}$, this gives $\beta g^-_{\beta,h}\geq-\log(5/4)$. Hence
\begin{equation}
    2\beta(1+g^-_{\beta,h})-4\beta c_0
    \geq \frac74\beta-2\log(5/4)
    >\log3.
\end{equation}
Thus the strong-clustering and bottleneck estimates established above apply.

Strong clustering from Corollary~\ref{cor:tfimstrong}, together
with Lemma~\ref{lem:decouple}, gives $a,c=\Omega(1)$,
$b\leq e^{-\Omega(\sqrt n)}$, and $q\leq e^{-\Omega(\sqrt n)}$
in the notation of \Cref{lem:slow}.
Set
\begin{equation}
    r_n=\frac{\sqrt n}{\log(n+2)},\qquad
    \epsilon_n=e^{-4r_n},\qquad t_n=e^{r_n}.
\end{equation}
By Lemma~\ref{lem:chen-local}, there is a Hermiticity-preserving linear map approximating the Lindbladian with error $\epsilon_n$ supported on neighborhoods of radius $O((\beta+1)(\log n+r_n))=o(\sqrt n)$.
Use  Lemma~\ref{lem:tfimbottleneck}, so the approximation obeys the
support condition in \Cref{lem:slow}. Apply \Cref{lem:slow} with $\Lambda=\max\{1,\norm{\cL}_\diamond\}$, $\delta=e^{-2r_n}/\Lambda$, and
$j=\lceil t_n/\delta\rceil$. Then, the mixing-time bound in Proposition~\ref{prop:mix_time_discrete} shows that the state has not mixed by time $j\delta\geq t_n$.
\end{proof}


\section*{Note}
The work of~\cite{rakovszky2026bottlenecks}, first posted around the same time as the current work, provides a quantum bottleneck lemma that shares similarities with our results. Their slow mixing lemma follows from conditions similar to weak clustering and does not require the further condition strengthening a weakly clustered state to a strongly clustered state. In particular, they provide a bottleneck lemma that shows that Gibbs samplers with sufficiently local Kraus operators (in the Pauli sense of Definition~\ref{def:local}) mix slowly for Hamiltonians with a linear free energy barrier defined by commuting projectors perturbed by a local non-commuting Hamiltonian. 

\section*{Declarations}
\paragraph{Funding}
DG gratefully acknowledges funding from NSF Grant CISE-2233897.
AZ gratefully acknowledges funding from the Hertz Foundation.

\paragraph{Author contributions}
This project was conceived jointly, and all authors contributed to the results, development of the theory, and derivation of the proofs.
In the most recent revision, we acknowledge use of AI models to assist in correcting typos and errors and cleaning up the draft.

\paragraph{Acknowledgments}
The authors thank Eric Anschuetz, Shankar Balasubramanian, Chi-Fang Chen, David Gosset, Aram Harrow and Seth Lloyd for enlightening discussions. We especially thank Tibor Rakovszky for pointing out an error in an earlier version of this manuscript.

\clearpage
\bibliography{References.bib}
\bibliographystyle{alpha}

\clearpage

\appendix

\section{Properties of Lindbladian evolution}

\subsection{General Lindbladians}
\label{app:lindblad}
We summarize basic properties of CPTP maps from Lindbladian evolution and provide accompanying proofs. Any Lindbladian can be written in the form
\begin{equation}\label{eq:lb}
    \cL(\rho) = -i\Big[\sum_j B_j, \rho\Big] + \sum_k L_k \rho L_k^\dagger - \frac{1}{2}\{L_k^\dagger L_k, \rho\}
\end{equation}
for Hermitian operators $B_j$. The first order approximation of $e^{\delta t \cL}$ is given by $\udt = 1 + \delta t \cL$. 
\begin{lemma}\label{lem:evolve-lindblad-dist}
Let $\cL$ and $\tilde \cL$ be any two Lindbladians. Then
\begin{equation}
    \norm{e^{t \cL} - e^{t \tilde \cL}}_\diamond \leq t \norm{\cL - \tilde \cL}_\diamond.
\end{equation}
\end{lemma}
\begin{proof}
Duhamel's identity gives
\begin{equation}
    e^{t\cL}-e^{t\tilde\cL}
    =\int_0^t e^{(t-s)\cL}(\cL-\tilde\cL)e^{s\tilde\cL}\,ds.
\end{equation}
$e^{(t-s)\cL}$ and $e^{s\tilde\cL}$ have diamond norm one. The triangle inequality and submultiplicativity of the diamond norm prove the claim.
\end{proof}
In this work, we are interested in the setting where $\cL$ is a Lindbladian and $\tilde \cL$ is a local approximation of it, i.e., one where $L_k$ and $B_j$ are replaced by local (or quasi-local) operators. In such a setting, $\cL$ and $\tilde \cL$ have the same number of terms $B_j$ and $L_k$. For convenience, we write the next property under this assumption.

\begin{lemma}\label{lem:lindblad-dist}
Let $\cL$ and $\tilde \cL$ be any two Lindbladians of the form in \eqref{eq:lb}, where both Lindbladians have $m$ operators $B_j$ and $\tilde B_j$, and $m'$ operators $L_k$ and $\tilde L_k$. Then
\begin{equation}
    \norm{\cL - \tilde \cL}_\diamond \leq 2 \sum_{j=1}^m \norm{B_j - \tilde B_j} + 2\sum_{k=1}^{m'} \lr{\bnorm{L_k}+\norm{\tilde L_k}}\norm{L_k-\tilde L_k}.
\end{equation}
\end{lemma}
\begin{proof}
We use the triangle inequality and submultiplicativity after adjoining an arbitrary identity ancilla map. A commutator with $B_j-\tilde B_j$ has diamond norm at most $2\norm{B_j-\tilde B_j}$. For arbitrary $V,\tilde V$, put $D=V-\tilde V$. The identity
\begin{equation}
    VXV^\dagger-\tilde VX\tilde V^\dagger
       =DXV^\dagger+\tilde VXD^\dagger
\end{equation}
gives diamond norm at most $(\norm{V}+\norm{\tilde V})\norm{D}$ for the corresponding map. Similarly,
\begin{equation}
    \norm{V^\dagger V-\tilde V^\dagger\tilde V}
       \leq (\norm{V}+\norm{\tilde V})\norm{D}.
\end{equation}
The difference of the two anticommutator maps, including its factor $1/2$, has diamond norm at most this last quantity. Therefore the difference of the dissipators has diamond norm at most
$2(\norm{V}+\norm{\tilde V})\norm{D}$.
\end{proof}

\subsection{Lindbladian of~\cite{chen2023efficient}}
\label{app:chenlindblad}
We show here useful properties of the Lindbladian of~\cite{chen2023efficient}, which we use as an example for our techniques. 
Let $A_a(t)=e^{iHt}A_ae^{-iHt}$ and consider throughout the choice of jump operators $\{A_a\}$ as the set of single qubit Pauli operators.
Define
\begin{equation}
    \begin{split}
        f_\beta(t)&=\frac{e^{-t^2/\beta^2}}
                           {\sqrt{\beta\sqrt{\pi/2}}},\\
        \gamma_\beta(\omega)&=e^{-(\beta\omega+1)^2/2},\\
        L_a(\omega)&=\frac{\sqrt{\gamma_\beta(\omega)}}{\sqrt{2\pi}}
                  \int_{\mathbb R} f_\beta(t)A_a(t)e^{-i\omega t}\,dt.
    \end{split}
\end{equation}
The Lindbladian then takes the form
\begin{equation}\label{eq:chenlind}
    \begin{split}
        \cL(\rho)
        &=-i\Big[\sum_{a=1}^J B_a,\rho\Big]\\
        &\quad+\sum_{a=1}^J\int_{\mathbb R}
          \left(L_a(\omega)\rho L_a(\omega)^\dagger
           -\frac12\{L_a(\omega)^\dagger L_a(\omega),\rho\}\right)d\omega.
    \end{split}
\end{equation}
Here the coherent term is given by the Gaussian time-domain representation in~\cite[Corollary~III.1]{chen2023efficient}:
\begin{equation}\label{eq:chendefs}
    \begin{split}
        B_a&=\int_{\mathbb R^2}B_a(t,t')\,dt\,dt',\\
        B_a(t,t')&=b_1(t)b_2(t')e^{-i\beta Ht}
                 A_a^\dagger(\beta t')A_a(-\beta t')e^{i\beta Ht},\\
        b_1(t)&=2\sqrt\pi e^{1/8}\int_{\mathbb R}
             \frac{\sin(s-t)e^{-2(s-t)^2}}{\cosh(2\pi s)}\,ds,\\
        b_2(t')&=\pi^{-3/2}e^{-4t'^2-2it'}.
    \end{split}
\end{equation}

We truncate the Lindblad operators and the coherent terms as follows:
\begin{equation}\label{eq:lops}
    \begin{split}
        \tilde L_a(\omega)
        &=\frac{\sqrt{\gamma_\beta(\omega)}}{\sqrt{2\pi}}
           \int_{-R\beta}^{R\beta}
                 f_\beta(t)A_a(t)e^{-i\omega t}\,dt,\\
        \tilde B_a
        &=\int_{-R}^{R}\int_{-R}^{R}B_a(t,t')\,dt\,dt'.
    \end{split}
\end{equation}
The functions $b_1$ and $b_2$ obey $b_1(t)\in\mathbb R$ and $b_2(-t')=\overline{b_2(t')}$. Hence both $B_a$ and $\tilde B_a$ are Hermitian by the substitution $t'\mapsto-t'$. Replacing $L_a,B_a$ in \eqref{eq:chenlind} by $\tilde L_a,\tilde B_a$ therefore defines a genuine Lindbladian $\tilde\cL$, although the arguments in the main text only require a Hermiticity-preserving approximation.

\begin{lemma}\label{lem:chenops}
For $R\geq1$, the operators in \eqref{eq:lops} satisfy
\begin{equation}
    \begin{split}
        \norm{L_a(\omega)-\tilde L_a(\omega)}
            &\leq C\sqrt\beta\,
              e^{-(\beta\omega+1)^2/4}e^{-R^2},\\
        \norm{L_a(\omega)}+\norm{\tilde L_a(\omega)}
            &\leq C\sqrt\beta\,
              e^{-(\beta\omega+1)^2/4},\\
        \norm{B_a-\tilde B_a}
            &\leq C e^{-2\pi R}.
    \end{split}
\end{equation}
\end{lemma}
\begin{proof}
The first two estimates follow by taking absolute values in the Fourier integrals, using $\norm{A_a(t)}\leq1$, and applying the Gaussian tail estimate
\begin{equation}
    \int_{|t|>R\beta}e^{-t^2/\beta^2}\,dt
       \leq C\beta e^{-R^2}.
\end{equation}
The Fourier phase has modulus one; it can be omitted from this norm estimate but not from the definition of $L_a(\omega)$.
For the coherent term, $|b_2(t')|\leq C e^{-4t'^2}$ and
\begin{equation}
    \begin{split}
        |b_1(t)|
        &\leq C\int_{\mathbb R}e^{-2(s-t)^2-2\pi|s|}\,ds\\
        &\leq C e^{-2\pi|t|}
                     \int_{\mathbb R}e^{-2u^2+2\pi|u|}\,du
         \leq C' e^{-2\pi|t|}.
    \end{split}
\end{equation}
The omitted region is contained in $\{|t|>R\}\cup\{|t'|>R\}$. Since the operator factor has norm at most one,
\begin{equation}
    \norm{B_a-\tilde B_a}
       \leq C(e^{-2\pi R}+e^{-4R^2})
       \leq C' e^{-2\pi R}.
\end{equation}
\end{proof}

We combine the above results with Lemma~\ref{lem:lindblad-dist}.
\begin{theorem}\label{thm:chen-bounds}
Let $\cL$ be the Gaussian Lindbladian in \eqref{eq:chenlind}, and let $\tilde\cL$ use the operators in \eqref{eq:lops}. For $0<\delta<1$, there is a universal constant $C$ such that
\begin{equation}
    R\geq C\log\frac{CJ}{\delta}
    \quad\Longrightarrow\quad
    \norm{\cL-\tilde\cL}_\diamond\leq\delta.
\end{equation}
\end{theorem}
\begin{proof}
The coherent contribution is $O(Je^{-2\pi R})$. For the dissipative contribution, Lemma~\ref{lem:chenops} gives
\begin{equation}
    \begin{split}
        &\sum_{a=1}^J\int_{\mathbb R}
          \bigl(\norm{L_a(\omega)}+\norm{\tilde L_a(\omega)}\bigr)
          \norm{L_a(\omega)-\tilde L_a(\omega)}\,d\omega\\
        &\qquad\leq CJ e^{-R^2}
                \int_{\mathbb R}\beta e^{-(\beta\omega+1)^2/2}\,d\omega
         =O(Je^{-R^2}).
    \end{split}
\end{equation}
Apply Lemma~\ref{lem:lindblad-dist}, with the sum over dissipators replaced by its absolutely convergent integral, and choose $R$ as in the statement.
\end{proof}

\section{Example bound on operator range of Gibbs sampler}
\label{app:ex_operator_range_bound}

The sampler of \cite{jiang2024quantum} is a discrete sampler which follows in principle the reject-accept procedure that is standard in Metropolis like procedures. Though we will not detail the algorithm in full, we will note the basic properties that allow us to prove it is a bounded degree operation in the jump operators.  We note however that the Gibbs sampler of~\cite{jiang2024quantum} is only proved to approximately set $\rho_\beta$ as a fixed point.

Throughout this section, we assume the Hamiltonian $H$ from which we would like to sample has eigenbasis
\begin{equation}
    H = \sum_i E_i \ket{\psi_i}\bra{\psi_i}.
\end{equation}

The sampler acts on a quantum state with $4$ registers. Let us assume the first three registers have $n, n_2, n_3,$ qubits. The last register will have $1$ qubit. An initial state $\ket{\phi}$ is inputted into the sampler into register $1$ and the remaining three registers are initialized as $\ket{0^{n_2}}\ket{0^{n_3}}\ket{0^{n_4}}$. The algorithm follows the steps below. 
\begin{itemize}
    \item Apply quantum phase estimation $\operatorname{QPE}_{1,2}$ to the state in register $1$ outputting the value in register $2$. See below for description of $\operatorname{QPE}$. Then, apply a jump operator $A^a$ followed again by quantum phase estimation $\operatorname{QPE}_{1,3}$ outputting the value now in register $3$. We denote this sequence of unitary operations as $U_C=\operatorname{QPE}_{1,3}A^a\operatorname{QPE}_{1,2}$.
    \item Let $W$ be a unitary acting on registers $2,3,4$. Measure register $4$ which outputs a binary value. The form of this unitary will not matter for our purposes as all that matters for our purposes is that it does not act on register $1$.
    \item If the outcome of the measurement above is $0$, then apply $W^\dagger$ and measure again the last qubit.
    \item If this outcome of this measurement is also equal to $0$, then apply $U_C^\dagger$ to the first three registers.
\end{itemize}

Let $\operatorname{QPE}_{a,b}$ denote quantum phase estimation on a state in register $a$ and output of the phase estimation procedure in state $b$. The particular form of phase estimation used is stated in \cite{jiang2024quantum}. What is important for us is to note the general property that for a given Hamiltonian $H$, if the state in register $a$ is an eigenstate $\ket{\psi_i}$ of the Hamiltonian, then this procedure does not change the state in the first register:
\begin{equation}
    \operatorname{QPE}_{a,b} \ket{\psi_i}\ket{A} = \ket{\psi_i} U(E_i) \ket{A}.
\end{equation}
We use the notation $U(E_i)$ to stress that the form of the unitary only depends on the energy of the eigenstate.

Now, let us show that the action of the sampler $\cM$ in \cite{jiang2024quantum} has Kraus operators that are degree $2$ in the jump operators $A_{j \to i}^a$. Assume w.l.o.g. we have a starting state $\ket{\phi}\bra{\phi}$ which is pure. Any mixed state will be a mixture of such pure states so we only need to consider the action of $\cM$ on one of these pure states. We write $\ket{\phi}$ in the eigenbasis of $H$ as
\begin{equation}
    \ket{\phi}= \sum_i c_i \ket{\psi_i}
\end{equation}
for coefficients $c_i$. The action of $U_C$ is then
\begin{equation}
    U_C \ket{\phi} = \sum_i \sum_j  c_i \bra{\psi_j} A^a \ket{\psi_i} \ket{\psi_j}  \otimes U(E_i)\ket{0^{n_2}} \otimes U(E_j)\ket{0^{n_3}} \otimes \ket{0}
\end{equation}

Upon applying $W$ and measuring to obtain an outcome $O_W \in \{0,1\}$, we get one of two states depending on this value. Let us denote this operation as $M_W$
\begin{equation}
    M_W U_C \ket{\phi} = \sum_i \sum_j  c_i c_W(E_i,E_j,O_W) \bra{\psi_j} A^a \ket{\psi_i} \ket{\psi_j}  \otimes \ket{\phi_{2,3}(O_W,i,j)} \otimes \ket{O_W},
\end{equation}
where $c_W(E_i,E_j,O_W)$ are coefficients depending on $i,j,$ and the output measurement and $\ket{\phi_{2,3}(O_W,i,j)}$ denotes the states in registers $2$ and $3$ upon collapse after measurement. This process may be repeated again upon application of $W^\dagger$ and a measurement of the last qubit which we will denote as $O_W'\in \{0,1\}$. Note that the above is still linear in the transitions $\bra{\psi_j} A^a \ket{\psi_i}$ even after this second measurement. Finally if one applies $U_C^\dagger$, we obtain a (potentially mixed) state where the first register spans over states of the form
\begin{equation}
   \sum_{i,j,k} c_k \bra{\psi_j} A^a \ket{\psi_i} \bra{\psi_k} (A^a)^\dagger \ket{\psi_j} \ket{\psi_k}
\end{equation}
for coefficients $c_k$ which depend on the outcomes of the measurements. Therefore, it is clear to see that the Kraus operators of the sampler $\cM$ can only induce up to two jumps in the jump operators $A^a$.

\section{Omitted proofs for \Cref{sec:localtfim}}
\label{app:localtfim}

Here, we provide the proofs omitted in Sec.~\ref{sec:localtfim} of the main text. We begin by showing the linear algebra Facts~\ref{fact:pos} and~\ref{fact:pow}, reproduced here.

\begin{fact}
    For arbitrary matrices $A, B$ with only nonnegative entries and for diagonal projector $\Pi$,
    \begin{equation}
        \tr(A \Pi B) \leq \tr(A B).
    \end{equation}
\end{fact}
\begin{proof}
    \begin{equation}
        \tr(A \Pi B) = \sum_{i,j} A_{ij} \Pi_{jj} B_{ji} \leq \sum_{i,j} A_{ij} B_{ji}.
    \end{equation}
\end{proof}

\begin{fact}
\label{fact:PSD_projector_trace_app}
    Let $A$ be a PSD matrix such that $A^p$ only has nonnegative entries for all $p \geq 0$, and let $\Pi$ be a diagonal projector. Then, for any $\gamma \geq 1$,
    \begin{equation}
    \lr{\Tr\lr{(A \Pi A)^\gamma}}^{1/\gamma} \leq \lr{\Tr(\Pi A^{2\gamma})}^{1/\gamma}.
    \end{equation}
\end{fact}
\begin{proof}
    \begin{equation}
    \begin{split}
        \left[\Tr\lr{(A \Pi A)^\gamma}\right]^{1/\gamma} &= \left[\Tr\lr{(A \Pi A)^{\lfloor\gamma\rfloor} \lr{A \Pi A}^{\gamma - \lfloor \gamma \rfloor}}\right]^{1/\gamma}\\
        &\leq \left[\Tr\lr{(A \Pi A)^{\lfloor\gamma\rfloor} \lr{A^2}^{\gamma - \lfloor \gamma \rfloor}}\right]^{1/\gamma}\\
        &= \left[\Tr\lr{(A \Pi A)(A \Pi A)^{\lfloor\gamma\rfloor-1} \lr{A^2}^{\gamma - \lfloor \gamma \rfloor}}\right]^{1/\gamma}\\
        &\leq \left[\Tr((A \Pi A) (A^2)^{\lfloor \gamma \rfloor - 1} A^{2(\gamma - \lfloor \gamma \rfloor)})\right]^{1/\gamma}\\
        &= \left[\Tr((A\Pi A) A^{2(\gamma-1)})\right]^{1/\gamma}\\
        &= \left[\Tr(\Pi A^{2\gamma})\right]^{1/\gamma},
    \end{split}
    \end{equation}
    where the first inequality uses the following facts: (1) $\Tr(M_1 M_2) \leq \Tr(M_1 M_2')$ if $M_1, M_2, M_2'$ are PSD with $M_2 \leq M_2'$, (2) $M_1 \geq M_2$ implies $M_1^p \geq M_2^p$ for $p \in (0, 1]$~\cite{bhatia2013matrix}, and (3) $A\Pi A \leq A^2$. To prove (1), observe $\Tr((M_2'-M_2)M_1) = \Tr(\sqrt{M_1}(M_2'-M_2)\sqrt{M_1}) \geq 0$. The second inequality follows by repeatedly applying Fact~\ref{fact:pos} to remove projectors one by one from the integer power $\lfloor \gamma \rfloor - 1$.
\end{proof}

We review the notation used in the main text for the $\sqrt n \times \sqrt n$ transverse field Ising model $H$ defined in~\eqref{eq:tfim}. Let $\sigma \in \{\pm 1\}^n$ be a configuration with a fault line of length $\ell$ and $k$ defects. Let $S$ denote the set of spins on one side of the fault line, and let $\partial S$ denote the set of spins adjacent to the fault line. Let $\sigma'$ be a configuration with all spins in $S$ flipped, i.e., $\sigma'_i = (-1)^{1_{i \in S}}\sigma_i$.

For all $i \in [n]$, define $s_i(t) \in \{\pm1\}$ by the spin-flip process
\begin{equation}
    s_i(t) = (-1)^{\cN_i(t)}, \quad \cN_i(t) \sim \mathrm{Pois}(\beta h t),
\end{equation}
where $\cN_i(t)$ are independent $\mathbb{N}$-valued Poisson processes on the positive half-line with mean $\beta h t$ and rate $\beta h$. Let $s(t)$ denote the full sequence $s_1(t), \dots, s_n(t)$ defined for $t\in[0,1]$. We use $\cD s(t)$ to denote the measure over Poisson processes $s_1(t), \dots, s_n(t)$. Finally, for a given $s(t)$, let $m_{\partial S}(s(t))$ denote the number of spins in $\partial S$ that are flipped in any time $t \in [0, 1]$; i.e.,
\begin{equation}
    m_{\partial S}(s(t)) = \left|\{i:i \in \partial S \; \mathrm{and} \; \exists \; t_* \in [0, 1] \; \mathrm{s.t.} \; s_i(t_*) = -1\}\right|.
\end{equation}
We show here the proof of Lemma~\ref{lem:fm}, whose statement is reproduced as follows.
\begin{lemma}
    The quantity
    \begin{equation}
        f_m=\mathbb E_{\cD s(t)}\left[
        \exp\left[\beta\int_0^1\sum_{\langle i,j\rangle}
                    \sigma_i s_i(t)\sigma_j s_j(t)\,dt\right]
        \,\middle|\,s(1)=1,\ m_{\partial S}(s(t))=m\right].
    \end{equation}
    satisfies
    \begin{equation}
        e^{-8\beta m} f \leq f_m \leq e^{8\beta m} f
    \end{equation}
    for some $f$ that is independent of $m$.
\end{lemma}
\begin{proof}
Let $M_{\partial S}$ contain all lattice neighbors $\langle i, j\rangle$ such that $i \in \partial S$ and/or $j \in \partial S$, and let $M_{\partial S}^c$ contain all remaining lattice neighbors. For any $R \subseteq \partial S$, let $M_{\partial S} = M_{\partial S}(R) \cup M_{\partial S}(R)^c$ where $M_{\partial S}(R)$ is the subset of $M_{\partial S}$ with at least one site in $R$, and $M_{\partial s}(R)^c$ is the subset with neither site in $R$.

We use $p_{s_i(t)}$ to denote the probability of a particular time-path $s_i(t)$ defined over $t \in [0, 1]$. For simplicity of notation, we will also write the integral over the measure of paths as a sum.

For normalization factor
\begin{equation}
\label{eq:am}
    A_m = \sum_{\substack{s_i(t):i \notin \partial S\\ s_i(1)=1}} \sum_{\substack{R \subseteq \partial S \\|R| = m}} \sum_{\substack{s_i(t):i\in\partial S, \, i\notin R\\ s_i(1)=1, \, s_i(t)=1}} \sum_{\substack{s_i(t):i \in R\\ s_i(1)=1, \, \exists t_* \in [0,1] \;\text{st}\; s_i(t_*)=-1}} \lr{\prod_{i \notin \partial S} p_{s_i(t)}}\lr{\prod_{i \in R} p_{s_i(t)}} \lr{\prod_{\substack{i \in \partial S\\i \notin R}} p_{s_i(t)}},
\end{equation}
we have
\begin{equation}
\begin{split}
    f_m &= \frac{1}{A_m} \sum_{\substack{s_i(t):i \notin \partial S\\ s_i(1)=1}} \lr{\prod_{i \notin \partial S} p_{s_i(t)}} \exp\left[\beta \int_0^1 dt \sum_{\langle i, j \rangle \in M_{\partial S}^c} \sigma_i s_i(t) \sigma_j s_j(t)\right]\\
    &\quad \times \sum_{\substack{R \subseteq \partial S \\|R| = m}} \sum_{\substack{s_i(t):i\in\partial S, \, i\notin R\\ s_i(1)=1, \, s_i(t)=1}} \lr{\prod_{i \in \partial S, \, i \notin R} p_{s_i(t)}} \exp\left[\beta \int_0^1 dt \sum_{\langle i, j \rangle \in M_{\partial S}(R)^c} \sigma_i s_i(t) \sigma_j s_j(t)\right]\\
    &\quad \times \sum_{\substack{s_i(t):i \in R\\ s_i(1)=1, \, \exists t_* \in [0,1] \;\text{st}\; s_i(t_*)=-1}} \lr{\prod_{i \in R} p_{s_i(t)}} \exp\left[\beta \int_0^1 dt \sum_{\langle i, j \rangle \in M_{\partial S}(R)} \sigma_i s_i(t) \sigma_j s_j(t)\right].
\end{split}
\end{equation}
Since there are at most $m$ sites in $R$ and each site can have at most 4 neighbors, $|M_{\partial S}(R)| \leq 4m$ and thus for any $s(t)$,
\begin{equation}
\label{eq:4mrep}
    \exp[-4\beta m] \leq \exp\left[\beta \int_0^1 dt \sum_{\langle i, j \rangle \in M_{\partial S}(R)} \sigma_i s_i(t) \sigma_j s_j(t)\right] \leq \exp[4\beta m].
\end{equation}
Since the steps for the upper and lower bounds are identical up to the sign of factors of $\beta m$ in the exponential, we only write the upper bound here. We replace the spins in the last exponential with the bounds of~\eqref{eq:4mrep}, and we explicitly write the conditioned choice of $s_i(t)=1$ for sites satisfying both $i \in \partial S, i \notin R$ to obtain
\begin{equation}
\begin{split}
    f_m &\leq \frac{\exp[4\beta m]}{A_m} \sum_{\substack{s_i(t):i \notin \partial S\\ s_i(1)=1}} \lr{\prod_{i \notin \partial S} p_{s_i(t)}} \exp\left[\beta \int_0^1 dt \sum_{\langle i, j \rangle \in M_{\partial S}^c} \sigma_i s_i(t) \sigma_j s_j(t)\right]\\
    &\quad \times \sum_{\substack{R \subseteq \partial S \\|R| = m}} \exp\left[\beta \int_0^1 dt \sum_{\langle i, j \rangle \in M_{\partial S}(R)^c} \sigma_i \sigma_j \lr{\delta_{i \in \partial S} + \delta_{i \notin \partial S} s_i(t)} \lr{\delta_{j \in \partial S} + \delta_{j \notin \partial S} s_j(t)}\right]\\
    &\quad \times \sum_{\substack{s_i(t):i\in\partial S, \, i\notin R\\ s_i(1)=1, \, s_i(t)=1}} \lr{\prod_{i \in \partial S, \, i \notin R} p_{s_i(t)}} \sum_{\substack{s_i(t):i \in R\\ s_i(1)=1, \, \exists t_* \in [0,1] \;\text{st}\; s_i(t_*)=-1}} \lr{\prod_{i \in R} p_{s_i(t)}},
\end{split}
\end{equation}
where the membership indicator functions $\delta_{i \in \partial S}$ set $s_i(t)=1$ for all $t$ if $i \in \partial S$, and otherwise use the $s_i(t)$ defined in the first sum. Note that in the second line, the condition $i\in \partial S$ implies that $i \notin R$ due to the sum over $M_{\partial S}(R)^c$. We now substitute summing over $M_{\partial S}(R)^c$ with summing over $M_{\partial S}$ in the second line, incurring an additional $\exp[\pm 4\beta m]$ in the upper and lower bounds due to the inequality of~\eqref{eq:4mrep}:
\begin{equation}
\begin{split}
    f_m &\leq \frac{\exp[8\beta m]}{A_m} \sum_{\substack{s_i(t):i \notin \partial S\\ s_i(1)=1}} \lr{\prod_{i \notin \partial S} p_{s_i(t)}} \exp\left[\beta \int_0^1 dt \sum_{\langle i, j \rangle \in M_{\partial S}^c} \sigma_i s_i(t) \sigma_j s_j(t)\right]\\
    &\quad \times \exp\left[\beta \int_0^1 dt \sum_{\langle i, j \rangle \in M_{\partial S}} \sigma_i \sigma_j \lr{\delta_{i \in \partial S} + \delta_{i \notin \partial S} s_i(t)} \lr{\delta_{j \in \partial S} + \delta_{j \notin \partial S} s_j(t)}\right]\\
    &\quad \times \sum_{\substack{R \subseteq \partial S \\|R| = m}} \sum_{\substack{s_i(t):i\in\partial S, \, i\notin R\\ s_i(1)=1, \, s_i(t)=1}} \lr{\prod_{i \in \partial S, \, i \notin R} p_{s_i(t)}} \sum_{\substack{s_i(t):i \in R\\ s_i(1)=1, \, \exists t_* \in [0,1] \;\text{st}\; s_i(t_*)=-1}} \lr{\prod_{i \in R} p_{s_i(t)}}.
\end{split}
\end{equation}
Operationally, this is equivalent to adding the interactions in $M_{\partial S}(R)$, but specifying that the spins in $R$ never flip. Inserting the quantity $A_m$ from~\eqref{eq:am} and performing the same argument for lower bounds, we obtain
\begin{equation}
    \exp[-8\beta m] f \leq f_m \leq \exp[8\beta m] f
\end{equation}
for $m$-independent quantity
\begin{equation}
\begin{split}
    f &= \lr{{\sum_{\substack{s_i(t):i \notin \partial S\\ s_i(1)=1}} \prod_{i \notin \partial S} p_{s_i(t)}}}^{-1}\sum_{\substack{s_i(t):i \notin \partial S\\ s_i(1)=1}} \lr{\prod_{i \notin \partial S} p_{s_i(t)}} \exp\left[\beta \int_0^1 dt \sum_{\langle i, j \rangle \in M_{\partial S}^c} \sigma_i s_i(t) \sigma_j s_j(t)\right]\\
    &\quad \times \exp\left[\beta \int_0^1 dt \sum_{\langle i, j \rangle \in M_{\partial S}} \sigma_i \sigma_j \lr{\delta_{i \in \partial S} + \delta_{i \notin \partial S} s_i(t)} \lr{\delta_{j \in \partial S} + \delta_{j \notin \partial S} s_j(t)}\right].
\end{split}
\end{equation}
\end{proof}

The remaining proofs are similar to known results, but we include them here for completeness.
The following lemma and corollary are reproduced from Lemma~\ref{lem:pathprops} and Corollary~\ref{cor:defect3}.
\begin{lemma}[Properties of paths (similar to Lemma 15.17~\cite{levin2017markov})]
    For a 2D $\sqrt{n} \times \sqrt{n}$ lattice with configuration $\sigma \in \{\pm 1\}^n$, the following properties hold.
    \begin{enumerate}
        \item If in $\sigma$ there is neither an all-plus nor an all-minus crossing from the left to the right side of the lattice, then there is a fault line with no defects from the top to the bottom.
        \item Let $\Gamma_+$ be a path of lattice sites such that $\sigma_i = 1$ for all $i \in \Gamma_+$, and $\Gamma_+$ starts at site $q$ and ends at a site at the top of the lattice. Similarly define $\Gamma_-$ from $q'$ to the top of the lattice, where $\sigma_i = -1$ for all $i \in \Gamma_-$. Assume $q$ and $q'$ are on the same row of the lattice, and let $\Gamma_{qq'}$ be the horizontal path of sites directly joining $q$ and $q'$. Then there exists a lattice path $\xi$ from the boundary of $\Gamma_{qq'}$ to the top of the lattice such that every edge in $\xi$ is adjacent to two lattice sites with different labels in $\sigma$.
    \end{enumerate}
\end{lemma}
\begin{proof}
We show the first property (shown in Lemma 2.2 of~\cite{randall2006slow} and Lemma 15.17 of~\cite{levin2017markov}). Let $A$ be the collection of lattice sites that can be reached from the left side by a path of lattice sites of the same label in $\sigma$. Let $A^*$ equal $A$ together with the set of sites that are separated from the right side by $A$. Then the boundary of $A^*$ consists of part of the boundary of the lattice and a fault line.

We show the second property. Regard each spin site as a unit square,
with neighboring sites represented by squares sharing a side. A site
path moves between neighboring squares, whereas the desired path $\xi$
follows their shared sides. When referring to the boundary of
$\Gamma_{qq'}$, we mean the boundary of the one-square-high strip formed
by its sites.

Let $A_+$ be the set of sites reachable from $\Gamma_+$ by a path of plus
spins. Let $A_+^*$ be $A_+$ together with all sites outside $A_+$ from
which no site path to the top avoids $A_+$. (Visually, $A_+^*$ corresponds to the union of $A_+$ and all the islands inside it.) Note that the paths used in this
definition may contain both plus and minus spins.

Every site outside $A_+^*$ can reach the top while staying outside
$A_+^*$: along a path avoiding $A_+$, every site also has such a path
to the top, obtained by taking the remaining portion. Every site in
$A_+^*$ can reach the top while staying in $A_+^*$. Indeed, take any site
path from that site to the top and stop at its first visit to $A_+$.
Such a visit must occur by the definition of $A_+^*$. Before that visit,
the path cannot leave $A_+^*$, since the preceding observation shows
that the inside site at any such crossing would already lie in $A_+$.
Then continue to the top within $A_+$.

Mark every side shared by a square in $A_+^*$ and a square outside
$A_+^*$. Two marked sides are joined when they share a corner.
We claim that each marked side separates opposite spins. To see this, suppose $u\in A_+^*$ and $v\notin A_+^*$ are neighboring sites. Then the site
$v$ has a path to the top avoiding $A_+$. If $u\notin A_+$, then $u$ could
also reach the top while avoiding $A_+$, by first stepping to $v$.
This contradicts $u\in A_+^*$. Thus $u\in A_+$ and $\sigma_u=+1$.
Moreover, $\sigma_v=-1$, since a plus spin adjacent to $A_+$ would itself
belong to $A_+$.

Every connected component of the marked sides meets the top.
Otherwise, such a component either contains a simple closed curve or,
if it contains no closed curve, has at least two endpoints and contains
a simple path between them, meeting the outer sides only at its
endpoints. These endpoints must lie on the left,
right, or bottom side, since no interior corner has exactly one incident
marked side. In the latter case, join the endpoints by the outer-side
arc that avoids the top. In either case, the resulting simple closed
curve encloses at least one spin square and no square in the top row.
A site path from an enclosed square to the top must cross a marked
side: it cannot leave the lattice through the outer-side arc.
Crossing a marked side changes membership in $A_+^*$, contradicting
the fact that every site can reach the top while staying either entirely
in $A_+^*$ or entirely outside it.

Finally, $q\in A_+^*$, whereas $q'\notin A_+^*$ because $\Gamma_-$
provides a path from $q'$ to the top avoiding $A_+$. Hence some two
successive sites of $\Gamma_{qq'}$ lie on different sides of $A_+^*$.
Let $\xi_1$ be their shared side, which is marked. Its lower corner
lies on the lower boundary of the horizontal strip $\Gamma_{qq'}$.
The connected component containing $\xi_1$ meets the top, so choose
a shortest path of marked sides from this corner to the top.
This path has no repeated corner, and every side in it separates
opposite spins. It is the required path $\xi$.
\end{proof}
\begin{corollary}
    Let $(A, B, C)$ be $c_0$-regions of the $\sqrt n \times \sqrt n$ Ising model. Every $\sigma \in A$ and $\sigma' \in C$ satisfies $|\sigma - \sigma'|_H \geq \frac23 (c_0\sqrt n - 3)$.
\end{corollary}
\begin{proof}
If $c_0\sqrt n\le3$, the claimed bound is immediate. Otherwise, fix
$\sigma\in A$ and $\sigma'\in C$, and choose a shortest sequence of
single-spin flips
\[
    \sigma_0=\sigma,\ \sigma_1,\ \ldots,\ \sigma_m=\sigma',
    \qquad m=|\sigma-\sigma'|_H.
\]
Each $\sigma_i\in\{\pm1\}^n$ is a full configuration, and
$\sigma_i(q)\in\{\pm1\}$ is its spin at site $q$.
Let $j$ be the first index for which $\sigma_j\notin S_+$.
Such an index exists because $\sigma_0\in A\subseteq S_+$ and
$\sigma_m\in C\subseteq S_-$, and $S_+\cap S_-=\varnothing$.
We show that $\sigma_j$ has a fault line with at most three defects.

If $\sigma_j\notin S_-$, then in at least one direction there is neither
an all-plus nor an all-minus crossing. Indeed, if there were a crossing
of one sign in each direction, the two crossings would intersect and
therefore have the same sign, putting $\sigma_j$ in $S_+$ or $S_-$.
The first property of Lemma~\ref{lem:pathprops}, with the lattice
rotated if necessary, gives a fault line with no defects.

Now suppose $\sigma_j\in S_-$, and let $q$ be the site flipped between
$\sigma_{j-1}\in S_+$ and $\sigma_j$. Every left-right crossing whose
spins all have the same sign is minus and contains $q$. A plus crossing
would intersect the minus top-bottom crossing in $\sigma_j$. A minus
crossing avoiding $q$ would remain unchanged in $\sigma_{j-1}$ and
would intersect a plus top-bottom crossing there. Both are impossible.

Let $U$ be the set of sites reachable from the left side by a path of
spins all of the same sign that avoids $q$, allowing paths consisting
of a single site. Also include $q$ in $U$ if it lies on the left side.
Then the whole left column belongs to $U$, and no site in the right
column belongs to $U$. If neighboring sites $u\in U$ and $v\notin U$
are both different from $q$, their spins are opposite: otherwise the
path reaching $u$ could be extended to $v$.

Represent each spin site by a square, with neighboring sites sharing
a side. Mark every side shared by a square in $U$ and a square outside
$U$; two marked sides are joined when they share a corner.
At each interior corner, an even number of marked sides meet, because
membership in $U$ changes zero, two, or four times around the four
adjacent squares. No marked side ends on the left or right side of the
lattice, since the whole left column is in $U$ and the whole right
column is outside $U$. The number of marked sides ending at the top
is odd: membership in $U$ changes an odd number of times along the top
row, from its leftmost site to its rightmost site.

Each connected component of the marked sides has an even number of
endpoints. To see this, sum the numbers of marked sides meeting at
its corners. This counts every marked side twice, and every corner
other than an endpoint contributes an even number. Since the total
number of endpoints at the top is odd, some connected component has
endpoints at both the top and the bottom. Choose a shortest path of
marked sides joining these two sides of the lattice. It has no repeated
corner and meets the outer sides only at its two endpoints. Its only
possible defects lie along sides of the square at $q$. It cannot use
all four sides of that square, since those four sides form a closed
loop. Thus it is a fault line with at most three defects.

Let $f(\tau)$ be the minimum number of defects among all fault lines
in a configuration $\tau$. Flipping one spin changes the defect count
of any fixed fault line by at most three, since a path with no repeated
corner cannot contain all four sides of one square. Taking minima and
applying this observation in both directions shows that one spin flip
changes $f$ by at most three. The intermediate configuration above
satisfies $f(\sigma_j)\le3$, while
$f(\sigma_0),f(\sigma_m)\ge c_0\sqrt n$. Hence
\[
    c_0\sqrt n-3\le3j,
    \qquad
    c_0\sqrt n-3\le3(m-j).
\]
Adding these inequalities and using $m=|\sigma-\sigma'|_H$ proves
the claimed bound.
\end{proof}

We now bound the locality of the Gibbs sampler via Lieb-Robinson bounds. Let $H=\sum_Zh_Z$ be a finite-range Hamiltonian on a fixed-dimensional lattice, with interaction diameters bounded independently of $n$ and with $\sup_x\sum_{Z\ni x}\norm{h_Z}=O(1)$. We use lattice distance between supports. Under these hypotheses the following standard form of the Lieb-Robinson estimate applies~\cite{nachtergaele2006lieb,hastings2006spectral,lieb1972finite,hastings2004lieb,haah2021quantum}.

\begin{lemma}[Lieb-Robinson bounds]\label{lem:lieb-robinson}
For an operator $A_X$ supported on $X$, there is an operator $A_{X(\ell)}(t)$ supported on
\begin{equation}
    X(\ell)=\{i:d(i,X)\leq v|t|+\ell\}
\end{equation}
such that
\begin{equation}
    \norm{A_X(t)-A_{X(\ell)}(t)}
       \leq C|X|\norm{A_X}e^{-\mu\ell}.
\end{equation}
The constants $C,v,\mu>0$ are independent of system size.
\end{lemma}
We may choose a norm-contracting local approximation, for example by normalized partial trace onto the indicated region. Fixed multiplicative changes in $C$ are harmless here.

\begin{corollary}\label{cor:lieb-time}
Under the preceding hypotheses, for every $0<\epsilon<1$ and time $t$ there is a local approximation $\tilde A(t)$ satisfying
\begin{equation}
    \norm{A(t)-\tilde A(t)}\leq\epsilon|X|\norm{A}
\end{equation}
and supported on
\begin{equation}
    \tilde X=\left\{j:d(j,X)\leq
                   c\left(|t|+\log\frac{C}{\epsilon}\right)\right\}.
\end{equation}
\end{corollary}
\begin{proof}
Choose $\ell=\mu^{-1}\log(C/\epsilon)$ in Lemma~\ref{lem:lieb-robinson}, and absorb the fixed constants into $c$.
\end{proof}

We now specialize to the Gaussian construction of~\cite{chen2023efficient}, using the corrected formulas in Appendix~\ref{app:chenlindblad}. Take each input $A_a$ to have norm at most one and support of uniformly bounded diameter and size.

\begin{theorem}
There exists a Hermiticity-preserving map $\cE$ such that
\begin{equation}
    \norm{\cL-\cE}_\diamond\leq\epsilon
\end{equation}
and $\cE$ admits a signed operator-sum representation whose operators are supported in neighborhoods of the seed supports with radius
\begin{equation}
    \ell=O\left((\beta+1)\log\frac{CJ}{\epsilon}\right).
\end{equation}
No stationary-state property of $\cE$ is required.
\end{theorem}
\begin{proof}
Use Theorem~\ref{thm:chen-bounds} to choose
$R=C\log(CJ/\epsilon)$ so that the time-truncated Lindbladian $\tilde\cL$ has diamond distance at most $\epsilon/2$ from $\cL$.
All real times in $\tilde L_a(\omega)$ are at most $R\beta$ in absolute value. For the coherent integrand, the identity
\begin{equation}
    e^{-i\beta Ht}A_a^\dagger(\beta t')A_a(-\beta t')e^{i\beta Ht}
       =A_a^\dagger\bigl(\beta(t'-t)\bigr)
        A_a\bigl(-\beta(t'+t)\bigr)
\end{equation}
shows that the individual times have absolute value at most $2R\beta$.

Choose local approximations to $A_a(u)$ uniformly for $|u|\leq2R\beta$, of operator-norm error at most $\xi$ and norm at most one, on a single neighborhood of the seed support. Corollary~\ref{cor:lieb-time} permits radius $O(R\beta+\log(C/\xi))$. Insert these approximations into the truncated integrals. Denote the resulting Lindblad operators by $L_a^{\rm loc}(\omega)$ and the resulting coherent terms by $B_a^{\rm loc}$. Take the Hermitian part of a coherent term if needed; this does not increase its distance from the Hermitian $\tilde B_a$ and does not enlarge its support.

The same integral bounds used in Lemma~\ref{lem:chenops} give
\begin{equation}
    \begin{split}
        \norm{\tilde L_a(\omega)-L_a^{\rm loc}(\omega)}
           &\leq C\xi\sqrt\beta\,e^{-(\beta\omega+1)^2/4},\\
        \norm{\tilde B_a-B_a^{\rm loc}}&\leq C\xi.
    \end{split}
\end{equation}
Together with the norm bound on $\tilde L_a(\omega)$, integration in $\omega$ and Lemma~\ref{lem:lindblad-dist} imply
$\norm{\tilde\cL-\cE}_\diamond\leq CJ\xi$.
Choose $\xi=\epsilon/(2CJ)$ and apply the triangle inequality.
The radius is
\begin{equation}
    O\left(R\beta+\log\frac{C}{\xi}\right)
        =O\left((\beta+1)\log\frac{CJ}{\epsilon}\right).
\end{equation}
The local dissipators and local coherent terms define a genuine Lindbladian; in particular the map is Hermiticity-preserving. Each can be written in a signed operator-sum form using operators on the same neighborhood. Thus the claimed signed support property holds.
\end{proof}

\end{document}